\documentclass{article}
\input{tikzit.tikzdefs}
\usepackage{authblk}
\usepackage{subfiles}     
\AtBeginDocument{%
}
\AtEndDocument{%
}

\title{Quantum-like Cognition in Process Theories: \\ An Analysis} 

\author[1]{{Sean Tull}}
\author[2,3,4,5]{{Masanao Ozawa}}

\affil[1]{Quantinuum, Partnership House, London, United Kingdom} 
\affil[2]{Department of Mathematical Informatics, Nagoya University, Japan}
\affil[3]{Kinugasa Research Organization, Ritsumeikan University, Japan}
\affil[4]{Center for Mathematical Science and Artificial Intelligence, Academy of Emerging Sciences, Chubu University, Japan}
\affil[5]{Fundamental Quantum Science Program, RIKEN, Japan}

\date{}
\begin{document}

\maketitle
\begin{abstract}
Various effects in human cognition, often considered `non-classical', have been argued to be most naturally modelled by quantum-like models of decision making. We extend this approach to describe models of cognition and decision-making in general probabilistic process theories, which include both classical probabilistic models and quantum instrument models as special cases. We show how many aspects of quantum-like cognition can be described diagrammatically in process theories, before using our approach to assess the arguments for quantum-like models. While standard Bayesian classical models are insufficient, we prove that any sequential decision data can in fact be given a more general form of classical instrument model, and see that even simple deterministic models can exhibit all cognitive effects. Restricting attention to instruments induced by measurements, such as classical Bayesian and quantum POVM models, rules out such a result, but is challenged by the fact that such instruments cannot account for certain effects. Finally, we argue that to strictly rule out classical instrument models one should make use of parallel composition in the modelling of joint decisions, and find real world cognitive data violating Bell inequalities. 
\end{abstract}

\section{Introduction} \label{sec:intro}
In recent years, a proposal that has gained increasing attention within cognitive science is that the mathematics of \emph{quantum theory} can provide a natural high-level modelling language for various aspects of human cognition, decision-making and reasoning \cite{aerts1995applications,aerts2009quantum,busemeyer2009empirical,wang2014context,busemeyer2012quantum,pothos2022quantum,huang2025overview}. 
Importantly, such claims are distinct from `quantum brain' hypotheses which postulate the relevance of quantum physics to the brain or mind \cite{hameroff1996orchestrated,jedlicka2017revisiting}. Instead, these \emph{quantum-like} models of cognition claim only that quantum representations provide a natural high-level account of mental processes, which may or may not emerge from an underlying classical system such as a neural network
\cite{khrennikov2025coupling}. 

A basic motivation 
comes from the modelling of various effects and so-called `fallacies' in human decision making, each of which violate the rules of classical logic and probability theory. For example, decisions are often found to \emph{interfere} with each other, leading to \emph{order effects} whereby the order in which decisions are made influences their outcome probabilities \cite{moore2002measuring,wang2014context}. 
These effects motivate researchers to pursue alternative models besides classical logic. A natural candidate has been to instead use `quantum logic', and hence model decisions as non-commuting quantum projective measurements which famously exhibit interference and order effects. To capture yet further cognitive effects researchers have more recently extended these to more general quantum \emph{instruments}, describing arbitrary quantum-classical interactions \cite{ozawa2019application,jacobs2017quantum}.

However, quantum-like cognition remains a controversial topic, and it is natural to wonder: is it really necessary to invoke the full apparatus of quantum theory specifically to explain such cognitive effects? 

Within the foundations of physics, great progress has been made in the understanding of \emph{general probabilistic theories} (GPTs), which generalise both classical and quantum information theory \cite{davies1970operational,gudder1973convex,barrett2007information,plavala2023general}. While often merely specifying a single isolated system, in their modern form a GPT is instead best described as a \emph{process theory} \cite{coecke2018picturing}, containing an entire class of systems and processes between them, including rules for how such systems and processes compose and interact with classical systems in measurement. Initially motivated by the desire for operational and informational understandings of quantum theory \cite{chiribella2011informational,hardy2011reformulating}, there are now various categorical formalisations of GPTs as process theories such as \cite{tull2016operational,gogioso2017categorical} beginning with the framework of operational-probabilistic theories \cite{chiribella2010probabilistic}.  Formally, a process theory is a special case of a \emph{(symmetric monoidal) category} \cite{heunen2019categories}, and as a result process theories benefit from an intuitive but formal graphical calculus based on \emph{string diagrams} \cite{selinger2010survey}.

Now, a priori, all we know of a cognitive decision is that it forms some probabilistic interaction process between an unknown system, describing the participant's mental state, and a known classical system describing the reported outcome of the decision. It is natural to thus initially regard a decision as such a process within an unknown process theory. By studying the properties of decisions one can then aim to deduce whether they are best described via classical theory, quantum theory or some other process theory. As a result clarity on the current status of the arguments in quantum-like cognition would be gained by giving an account of them in the language of process theories. 
In this article we give such an analysis of quantum-like cognition in terms of probabilistic process theories. We use these to give high-level diagrammatic  accounts of many cognitive effects, aiming to give a succinct presentation while also generalising beyond classical and quantum models. Following this, we use the process-theoretic perspective to assess the status of existing arguments for the use of quantum-like models. Overall, we find that while these cognitive effects rule out standard `Bayesian' classical models, which feature no state updating, they do not rule out more general classical instrument models, which are in fact always available for any data, as we prove in Theorem \ref{thm:-simple-M-Model}. We then explore two potential routes forward for arguments for non-classicality of cognition, based on either restricting to `measurement models' or considering `Bell-type' arguments for joint decision data.

\paragraph{Process theories.} 
Let us now explain our approach in more detail. 
Throughout we work in a process theory $\catC$, also known as symmetric monoidal category. Specifically, we work in what we call a \emph{probabilistic process theory} which contains a subset of objects which we call \emph{classical} objects, and which we will use to model the outcomes of decisions. Any such process theory $\catC$ consists of a class of objects (systems) drawn generally as thick wires, with classical objects drawn as thinner wires, and processes drawn as boxes read from bottom to top, which may be composed in string diagrams.

As examples, classical processes are described by the process theory $\Class$ of finite sets and sub-probability channels. The process theory $\QC$ additionally contains finite-dimensional quantum systems (Hilbert spaces) and quantum processes (completely positive trace non-increasing maps) \cite{coecke2018picturing}. 

\paragraph{Decision models.}
Consider a subject making a cognitive decision $A$ which leads to an observable outcome from a finite set of possible outcomes $\XAintro$. For example, the subject could simply be asked a question, to which they give an answer from the set $\XAintro$. We model such a decision as a process of the following form, called an \emph{instrument} in $\catC$, read bottom to top.
\[
\input{introA.tikz}
\]
The input is given by some `hidden' object $S$, describing the state space for their mental state before the decision. The outputs are another instance of $S$, corresponding to their mental state after the decision, and a classical object $\XAintro$ corresponding to the set of possible outcomes of the decision. Given some initial state $\Mo$ of the object $S$, we obtain a probability distribution over the outcomes $\XAintro$ of the decision $A$ as follows. 
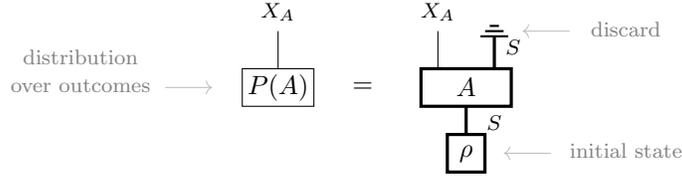
\[
\input{introPA.tikz}
\]
Here the symbol $\discard{}$ corresponds to `discarding' or `ignoring' $S$ (classically via marginalisation or in the quantum case via the partial trace). In a probabilistic process theory, the left-hand side will then correspond to a probability distribution $\PA$ over the possible outcomes $\XAintro$.

More generally, \emph{decision data} from an experiment consists of a set of probability distributions over the outcomes of various possible \emph{sequences} of decisions. An \emph{instrument model} of the data in a process theory $\catC$ is then given by specifying an object $S$, an initial state $\Mo$, and for each decision a corresponding instrument, such that each probability distribution in  the data is given by composing the instruments in sequence. For example, data for decision $A$ followed by decision $B$ is given by a distribution $\PAB$ over their respective outcomes $X_A \times X_B$, and must satisfy the following. 
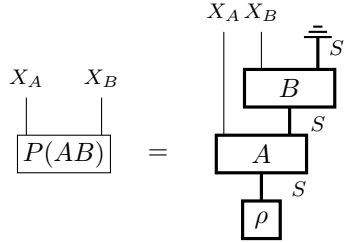
\[
\input{PABintro.tikz}
\]

\paragraph{Cognitive effects.}
Quantum-like cognition often concerns studying models of decisions which account for various effects in decision data. Examples include order effects in which $\PAxBy \neq \PByAx$ for some outcomes $x,y$, and interference effects\footnote{Not to be confused with interference in the sense of quantum superpositions, see Section \ref{subsec:interference-effect}.} in which the marginal of $\PAB$ on $B$ is distinct from $\PB$ alone. Each such effect is closely related to a similar (in)equality at the level of instruments in $\catC$, allowing us to describe their key properties diagrammatically, as in the example diagrams below (explained fully in later sections). 
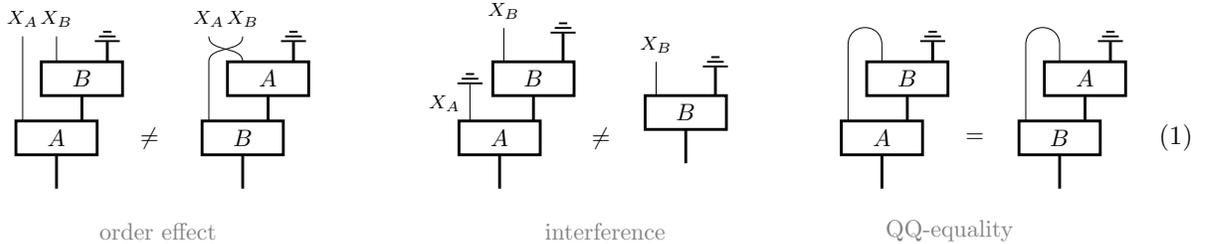
\begin{equation} \label{eq:effects-simple}
\scalebox{0.9}{
\input{intro-OE.tikz}
\qquad \qquad
\input{intro-int.tikz}
\qquad \qquad
\input{intro-QQ.tikz}
}
\end{equation}
In Section \ref{sec:cog-effects}, we give a high-level diagrammatic account of several such cognitive effects including interference, order and replicability effects, conjunction fallacies, and the QQ-equality. Appendices \ref{app:cond-data} and \ref{app:further-effects} extend this to also cover disjunction fallacies, double stochasticity, reciprocity, `changes of context' and the Temporal Bell inequality.

\paragraph{Classical and quantum models.}
By varying our process theory and choice of instruments, we reach various classes of models which one may use to attempt to model these cognitive effects, outlined in Table \ref{table:instruments}. In a classical model, given by working in the classical process theory $\Class$, $S$ is given by a finite set, $\rho$ is an initial probability distribution over $S$, and an instrument with outcome set $X$ in general corresponds to a probability channel $P(X, S' \mid S)$ from the previous state to the next state along with the outcome $X$. Special cases are \emph{Bayesian} instruments, which simply copy and return the state on $S$ while yielding a probabilistic outcome in $X$, and \emph{deterministic} instruments 
which map each state of $S$ via a function to an outcome in $X$ and updated state of $S$. 

A quantum model, in the process theory $\QC$, replaces $S$ by a finite dimensional Hilbert space $\hilbH$, with $\rho$ given by a density matrix over $\hilbH$. A quantum instrument consists of a collection of \emph{completely positive} maps $(\instAI_x)_{x \in X}$ on $\hilbH$ with trace-preserving sum. Special cases include those induced by \emph{projective measurements} $(P_x)_{x \in X}$, including orthonormal basis measurements.

\begin{center}
\begin{table}[h]
\centering
\begin{tabular}{l|l|l}
Instrument type &  State space &  \\ \hline
Classical& Finite set $S$ & Prob.~channel $P(\XAI, S \mid S)$ \\ 
 \quad Bayesian& 
& Prob.~channel $P(\XAI \mid S)$ \\ 
\quad Deterministic & 
& Function $S \to \XAI \times S$ \\
\quad Markov & 
& Prob. channels $P(S \mid S)$, $P(\XAI \mid S)$  \\ 
\hline
Quantum & F.d. Hilbert space $\hilbH$ & CP maps $(\instAI_x)_{x \in \XAI}$ on $L(\hilbH)$
with trace-preserving sum \\ 
\quad Projective& 
& Projective measurement $(\projPAI_x)_{x \in \XAI}$ on $\hilbH$\\ 
\quad POVM & 
& Positive operator valued measure $(\povmAI_x)_{x \in \XAI}$ on $\hilbH$
\end{tabular}
\caption{Overview of classes of instruments with set of outcomes $\XAI$.}\label{table:instruments}
\end{table}
\end{center}

\paragraph{Arguments for quantum-like models.}

A key question is then: to what extent do these cognitive effects considered in the literature  enforce the use of quantum models specifically? Firstly, we find that the effects do indeed rule out classical Bayesian models, which do not disturb the underlying state and always commute, and hence cannot satisfy order or interference effects. In contrast, thanks to their non-commutativity, quantum projective models can exhibit both order and interference effects, and this has formed a major argument for quantum models in the literature \cite{busemeyer2012quantum}. 

However, we find that more general classical instruments, which allow for state disturbance, are in fact able to account for all of these effects. Indeed, it is well-known within the study of GPTs and process theories that general classical processes and models, unlike Bayesian models more specifically, need not satisfy non-disturbance nor commutativity, and so one should not equate classicality with these properties. 

More strongly, we prove in Theorem \ref{thm:-simple-M-Model} that essentially any sequential decision data can in fact be given a classical Markov model. In fact even simple deterministic classical instrument models can exhibit all of the cognitive effects listed in Section \ref{sec:cog-effects} together, as shown in \cite{ozawa2021modeling}. Further arguments are thus required for those wishing to strictly rule out the use of classical models, and we here identify two possible routes forward.

\paragraph{Measurement models.}
One approach to ruling out classical models which we discuss in Section \ref{sec:meas-models} is to limit the class of allowed instruments to only those canonically associated to \emph{measurements}. Classically this means only allowing Bayesian instruments, induced by probability channels $S \to X$, and in the quantum case only instruments induced by projective measurements or more general POVMs. 
Most of the literature has indeed implicitly focused on such \emph{measurement models}, comparing projective quantum models to classical Bayesian ones. These models benefit from a generally simpler representation with fewer parameters. However, a challenge is that quantum measurement models cannot satisfy all desirable cognitive effects, most notably the combination of order and replicability effects \cite{khrennikov2014quantum}. 

\paragraph{Joint decisions.}
Alternatively, if we wish to strictly rule out classical models whilst allowing for the use of arbitrary instruments, we are then forced to go beyond sequential data. 
It is then natural to thus consider the modelling of \emph{joint} decision processes. These make use of the parallel composition $\otimes$ in the process theory, depicted by drawing processes side-by-side, which in the quantum case corresponds to the tensor product. A \emph{parallel decision model} is now given by a set of distributions denoted $\{P(A \jdec B)\}_{A \in \decsetA, B \in \decsetB}$ over $\Xone \times \Xtwo$, modelled by parallel measurements as on the right-hand side:
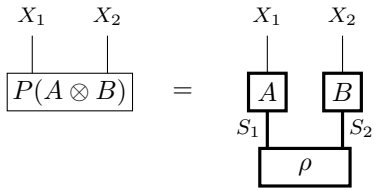
\[
\input{intro-bell-2.tikz}
\] 
where $P(A \jdec B)$ is interpreted as the decisions $A, B$ being made in parallel. Here $\Hione$ and $\Hitwo$ are typically interpreted as separate aspects of the mental representation $\Hi$ (rather than separate agents).

In Section \ref{sec:Bell-arguments} we outline how such models could be applied to cognition to rule out the use of classical models. Specifically, for any such data violating a \emph{Bell inequality} there is provably no classical model of the right-hand form, i.e.~given by a probability distribution $\Mo$ and probability channels $A, B$ as above \cite{bell1964einstein,clauser1969proposed} \cite[Chapter 5]{busemeyer2012quantum}. Such arguments based on \emph{Bell's theorem} would provide genuinely stricter arguments ruling out classical models and forcing the use of further probabilistic theories. 

\subsection{Related \sscap{w}ork}
We use \cite{busemeyer2012quantum} as our main source on quantum-like cognition including both cognitive effects and examples. 
See for example \cite{pothos2022quantum,huang2025overview} for a more full review of quantum-like modelling. The arguments for (quantum) instruments and our perspective on order effects and RRE comes from \cite{ozawa2019application}. Our main process theory $\QC$ and graphical notation is essentially that of \cite{coecke2018picturing}, which also serves as an introduction to process theories; for a more categorical introduction see \cite{heunen2019categories}. Our categorical approach to probabilistic process theories is closely related to others such as \cite{tull2016operational,gogioso2017categorical}. 
 The use of GPTs in quantum-like cognition has been previously suggested in \cite{khrennikov2025coupling}, in that case for modelling mental state spaces derived from underlying neuronal dynamics. Here we extend this beyond single systems to entire process theories.

An earlier formal high-level categorical presentation of quantum-like cognition, which served as inspiration and source of examples, is due to Jacobs \cite{jacobs2017quantum}. Jacobs' `effect logic' essentially corresponds to our `measurement models' and `instrument maps' to our `update structures'. Our approach here is more comprehensive and considers more general instrument models rather than merely measurement models. 
 Key differences between our arguments around Bell inequalities and prior uses in quantum-like modelling, particularly in relation to entangled concepts \cite{
gabora2002contextualizing, 
aerts2005theory,
aerts2009quantum,
aerts2018spin}, are discussed in detail in Appendix \ref{sec:entangled-concepts}.

\subsection{Structure of \sscap{a}rticle}
We begin by introducing process theories in Section \ref{sec:setup} and both classical and quantum decision models in Section \ref{sec:decision-models}. We give diagrammatic accounts of cognitive effects in Section \ref{sec:cog-effects} and prove the existence of classical Markov models in Section \ref{sec:classical-models}. We then explore possible approaches based on measurement models in Section \ref{sec:meas-models} and joint decisions and Bell arguments in Section \ref{sec:Bell-arguments}. In Section \ref{sec:Discussion} we close with an overall discussion on the status of quantum vs classical decision models and future directions. Appendices \ref{app:cond-data} and \ref{app:further-effects} contain process-theoretic accounts of yet futher aspects of quantum-like modelling. Proofs are left for Appendix \ref{app:proofs} and \ref{sec:general-RRE-Markov-models} and related work on entangled concepts discussed in Appendix \ref{sec:entangled-concepts}.

\section{Process Theories} \label{sec:setup}

Let us now begin by describing our formal setup. Throughout we will work in a \emph{process theory}, or in other words a \emph{symmetric monoidal category} $(\catC, \otimes)$,
making use of the graphical calculus of \emph{string diagrams} \cite{selinger2010survey}. For a more detailed introduction to process theories with a quantum focus see e.g.~\cite{coecke2018picturing,heunen2019categories}.

Recall that a category consists of a collection of \emph{objects} $\obA, \obB,\dots$ (also referred to as \emph{systems}) and \emph{morphisms} $\morf, \morg \dots$ (also referred to as \emph{processes}). In diagrams, objects are depicted as labelled wires and a morphism $\morf \colon \obA \to \obB$ as a box with input wire $\obA$ and output wire $\obB$, read from bottom to top. 
\[
\begin{tikzpicture}
	\begin{pgfonlayer}{nodelayer}
		\node [style=map] (0) at (0, 0) {$\morf$};
		\node [style=none] (1) at (0, -1.25) {};
		\node [style=none] (2) at (0, 1.25) {};
		\node [style=label] (3) at (0, 1.75) {$\obB$};
		\node [style=label] (4) at (0, -1.75) {$\obA$};
	\end{pgfonlayer}
	\begin{pgfonlayer}{edgelayer}
		\draw (2.center) to (1.center);
	\end{pgfonlayer}
\end{tikzpicture}

\]
Given another morphism $\morg \colon \obB \to \obC$ we can compose both to yield a morphism $\morg \circ \morf \colon \obA \to \obC$ depicted as left-hand below. Each object $\obA$ also has an identity morphism $\id{\obA}$ depicted as a blank wire (right-hand below), such that $\id{\obB} \circ \morf = \morf = \morf \circ \id{\obA}$ for all $\morf \colon \obA \to \obB$. 
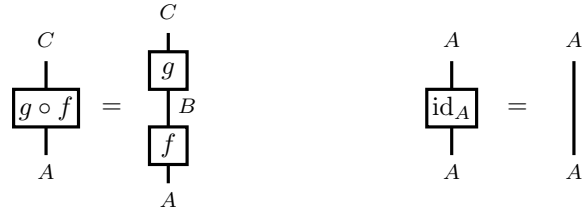
\[
\input{composite-1.tikz} \qquad \qquad  \qquad \qquad  \input{identity.tikz} 
\]
In a monoidal category, given objects $\obA, \obB$ we can form their \emph{monoidal product} $\obA \otimes \obB$ depicted by drawing wires side by side. Similarly, given morphisms $\morf \colon \obA \to \obB, \morg \colon \obC \to \obD$ we can form their monoidal product $\morf \otimes \morg \colon \obA \otimes \obC \to \obB \otimes \obD$ depicted as below. 
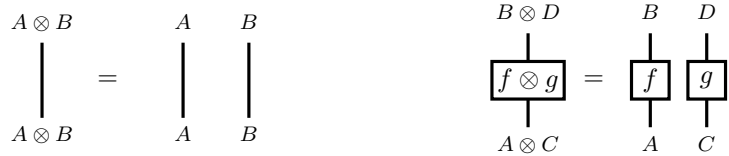
\[
\input{ob-tensor.tikz} 
\qquad \qquad \qquad \qquad 
\input{tensor.tikz}
\] 
Intuitively, we think of such a product as consisting of two independent processes $f, g$ occuring in parallel. Since our monoidal category is symmetric it also comes with natural\footnote{Meaning that $\tinyswap_{C,D} \circ (f \otimes g) = (g \otimes f) \circ \tinyswap_{A,B}$ for all $f \colon A \to C, g \colon B \to D$.} \emph{swap} morphisms allowing us to cross wires over each other: 
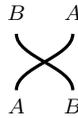
\[
\input{swap.tikz}
\]
which satisfy $\tinyswap \circ \tinyswap = \id{}$. 
 A monoidal category also comes with a \emph{unit object} denoted by $I$ which is not shown in diagrams, depicted simply as empty space, with identity denoted $1 := \id{I}$. Morphisms $\omega \colon I \to \obA$, called \emph{states} of $\obA$, are depicted as follows, appearing as boxes with `no input'. 
\[
\begin{tikzpicture}
	\begin{pgfonlayer}{nodelayer}
		\node [style=map] (0) at (0, -0.5) {$\omega$};
		\node [style=none] (1) at (0, 1) {};
		\node [style=label] (2) at (0, 1.5) {$A$};
	\end{pgfonlayer}
	\begin{pgfonlayer}{edgelayer}
		\draw (1.center) to (0);
	\end{pgfonlayer}
\end{tikzpicture}

\]
Morphisms $r, s \colon I \to I$, drawn without inputs or outputs, are called \emph{scalars}. For any morphism $f \colon A \to B$ and scalar $r$ we can define a scalar multiplication $r \cdot f \colon A \to B$ by: 
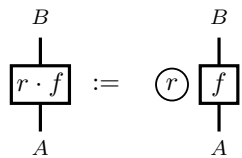
\[
\input{scalarmult.tikz}
\]
In particular, scalars come with a 
commutative and associative multiplication operation given by 
\[
\begin{tikzpicture}
	\begin{pgfonlayer}{nodelayer}
		\node [style=scalar] (0) at (2.75, 0) {$r$};
		\node [style=scalar] (1) at (4, 0) {$s$};
		\node [style=none] (2) at (1.5, 0) {$=$};
		\node [style=scalar] (3) at (-0.25, 0) {$r \cdot s$};
	\end{pgfonlayer}
\end{tikzpicture}

\]
with unit $1=\id{I}$.
We will assume $\catC$ comes with a \emph{zero} scalar $0$ satisfying $r \cdot 0 = 0$ for all scalars $r$, and more generally $0 \cdot f = 0 \cdot g$ for all $f, g \colon A \to B$.

Throughout we will work in an SMC with \emph{discarding}, meaning that each object $\obA$ comes with a distinguished \emph{discarding} process $\discard{\obA} \colon \obA \to I$ depicted as a ground symbol:
\[
\begin{tikzpicture}[tikzfig]
	\begin{pgfonlayer}{nodelayer}
		\node [style=upground] (28) at (12, 0.75) {};
		\node [style=none] (29) at (12, -0.5) {};
		\node [style=label] (30) at (12, -1) {$\obA$};
	\end{pgfonlayer}
	\begin{pgfonlayer}{edgelayer}
		\draw (29.center) to (28);
	\end{pgfonlayer}
\end{tikzpicture}

\]
such that the following holds. 
\[
\input{disc-nat.tikz} \qquad \qquad \input{disc-I.tikz}
\]
We call a morphism $\morf \colon \obA \to \obB$ a \emph{channel} when it preserves discarding, in that the following holds.\footnote{Channels are also often called `causal'  \cite{coecke2016terminality} or `total' \cite{cho2015introduction}.}
\[
\input{channel.tikz}
\]
As a special case we call a state $\omega$ \emph{normalised} when:
\[
\begin{tikzpicture}
	\begin{pgfonlayer}{nodelayer}
		\node [style=map] (0) at (-1.5, -0.75) {$\omega$};
		\node [style=upground] (1) at (-1.5, 1) {};
		\node [style=none] (2) at (0.25, 0) {$=$};
		\node [style=none] (3) at (1.75, 0) {$1$};
	\end{pgfonlayer}
	\begin{pgfonlayer}{edgelayer}
		\draw (0) to (1);
	\end{pgfonlayer}
\end{tikzpicture}

\]

Next, it will be convenient to assume the following extra structure, allowing us to explicitly capture classical `information' carried by finite sets. 
\begin{definition}[Classical Objects]
We say that a process theory $\catC$ \emph{has classical objects} when for every finite set $X$ it contains an object also denoted $\cob$, the \emph{classical object} corresponding to $X$, depicted as a thin wire:
\[
\begin{tikzpicture}
	\begin{pgfonlayer}{nodelayer}
		\node [style=none] (0) at (0, 1) {};
		\node [style=none] (1) at (0, -1) {};
		\node [style=label] (2) at (0, 1.5) {$\cob$};
		\node [style=label] (3) at (0, -1.5) {$\cob$};
	\end{pgfonlayer}
	\begin{pgfonlayer}{edgelayer}
		\draw [style=cwire] (1.center) to (0.center);
	\end{pgfonlayer}
\end{tikzpicture}

\]
and satisfying the following.\footnote{Categorically, the condition states that $X$ is a finite copower $X \cdot I$ and these copowers distribute over $\otimes$, so that $A \otimes X = X \cdot A$.}
\[
\begin{tikzpicture}
	\begin{pgfonlayer}{nodelayer}
		\node [style=none] (0) at (0, 1) {};
		\node [style=classpoint] (1) at (0, -0.5) {$x$};
		\node [style=label] (2) at (0, 1.5) {$\cob$};
	\end{pgfonlayer}
	\begin{pgfonlayer}{edgelayer}
		\draw [style=cwire] (0.center) to (1);
	\end{pgfonlayer}
\end{tikzpicture}

\]
and such that for all objects $\obA, \obB$ and any collection of morphisms $(f_x \colon \obA \to \obB)_{x \in \cob}$ there is a unique morphism  $f \colon \obA \otimes \cob \to \obB$ such that the following holds for all $x \in X$: 
\[
\input{class-coprod.tikz}
\]
\end{definition}
Intuitively, the above condition tells us that each finite classical object $X$ corresponds to the corresponding finite set, with elements $x \in X$ corresponding to the states $x$.
For each $x \in X$ we define another process $x\colon \cob \to I$ depicted as an inverted triangle:
\begin{equation} \label{eq:deltaxy}
\begin{tikzpicture}
	\begin{pgfonlayer}{nodelayer}
		\node [style=none] (0) at (0, -0.75) {};
		\node [style=classcopoint] (1) at (0, 0.75) {$x$};
		\node [style=label] (2) at (0, -1.25) {$\cob$};
	\end{pgfonlayer}
	\begin{pgfonlayer}{edgelayer}
		\draw [style=cwire] (0.center) to (1);
	\end{pgfonlayer}
\end{tikzpicture}
 
\qquad
\qquad \text{with}
\qquad
\qquad
\begin{tikzpicture}
	\begin{pgfonlayer}{nodelayer}
		\node [style=classcopoint] (0) at (0, 0.75) {$x$};
		\node [style=classpoint] (1) at (0, -0.75) {$y$};
		\node [style=label] (2) at (-0.75, 0) {$X$};
	\end{pgfonlayer}
	\begin{pgfonlayer}{edgelayer}
		\draw [style=cwire] (1) to (0);
	\end{pgfonlayer}
\end{tikzpicture}
 \ \ = \ \ 
\begin{cases} 1 & x =y \\ 0 & \text{otherwise} 
\end{cases}
\end{equation}
Next we define the \emph{copy} morphism $\tinycopy \colon \cob \to \cob \otimes \cob$ by: 
\[
\input{copy-points2.tikz} \ \ \forall x 
\]
It follows that classical objects are closed under tensor via $X \otimes Y = X \times Y$, and copy morphisms satisfy the following.\footnote{That is, copying forms a commmutative comonoid with discarding, and the subcategory of classical objects forms a \emph{cd-category} \cite{cho2019disintegration} and its further subcategory of classical channels a \emph{Markov category} \cite{fritz2020synthetic}.
Though not necessary for this article,  one could more broadly consider infinite or continuous classical objects, for example corresponding to commutative C*-algebras or von Neumann algebras. These should satisfy that they merely form a cd-subcategory with chosen copying morphisms $\tinycopy$.} 
\[
\input{markov-axioms.tikz} \qquad \qquad 
\input{copy-nat.tikz} 
\]

Finally, we will assume that $\catC$ carries the structure of a probabilistic theory in the following sense. By a \emph{probability sub-distribution} over a finite set $X$ we mean a collection $(\omega(x))_{x \in X}$ of probabilities $\omega(x) \in [0,1]$ with $\sum_{x \in X} \omega(x) \leq 1$. When this is an equality we simply call the collection a \emph{distribution}.

\begin{definition}[Probabilistic Process Theory] \label{def:prob-PT}
We say $\catC$ forms a \emph{probabilistic process theory} when it satisfies the above properties and the following:
\begin{enumerate}
\item \label{enum:scalars-are-probs}
scalars $\catC(I,I) = [0,1]$ correspond to probabilities $p \in [0,1]$ under $p \cdot q = pq$ and $1=\id{I}$;
\item \label{enum:scalar-cancellative}
$p \cdot f = p \cdot g \implies f = g$ for all morphisms $f, g$ and $p \in (0,1]$;
\item \label{enum:states-to-sub-dist}
states (resp. normalised states) $\omega$ of a classical object $X$ are in one-to-one correspondence with sub-probability (resp. probability) distributions $(\omega(x))_{x \in X}$ over the set $X$ via:
\begin{equation} \label{eq:probs}
\omega(x)  \ \ :=  \ \ \begin{tikzpicture}
	\begin{pgfonlayer}{nodelayer}
		\node [style=cmap] (0) at (-3, -0.75) {$\omega$};
		\node [style=classcopoint] (1) at (-3, 0.75) {$x$};
		\node [style=label] (4) at (-3.75, 0) {$X$};
	\end{pgfonlayer}
	\begin{pgfonlayer}{edgelayer}
		\draw [style=cwire] (0) to (1);
	\end{pgfonlayer}
\end{tikzpicture}
 \in [0,1]
\end{equation}
\setcounter{nameOfYourChoice}{\value{enumi}}
\end{enumerate}
\end{definition} 

Thanks to these axioms, morphisms between classical objects allow us to encode (finite) classical probability theory. Indeed, by the last condition, normalised states $\omega$ of a classical object $X$ correspond as above to probability distributions over $X$:
\begin{equation} \label{eq:prob-dist}
\text{Normalised } \qquad 
\begin{tikzpicture}
	\begin{pgfonlayer}{nodelayer}
		\node [style=cmap] (0) at (-7, -0.5) {$\omega$};
		\node [style=none] (1) at (-7, 0.75) {};
		\node [style=label] (2) at (-7, 1.25) {$X$};
	\end{pgfonlayer}
	\begin{pgfonlayer}{edgelayer}
		\draw [style=cwire] (1.center) to (0);
	\end{pgfonlayer}
\end{tikzpicture}
 \qquad \iff \qquad \text{Probability distribution $\omega(X)$}
\end{equation}
In particular, the state for each element $x \in X$ is the point distribution $\delta_x$ at $X$, by \eqref{eq:deltaxy}. More generally, for any morphism $M \colon X \to Y$ between finite classical objects we can define:
\begin{equation} \label{eq:class-values}
\ M(y \mid x) \ \ := \ \ \input{classvalues.tikz}  \ \ \in [0,1]
\end{equation}
Then $M$ is a channel precisely when it is \emph{Stochastic}, i.e.~a probability channel, meaning that: 
\[
\sum_{y \in Y} M(y \mid x) = 1
\]
for all $x \in X$. Such a channel might more conventionally be denoted as $M(Y \mid X)$ or $P(Y \mid X)$.
\[
\text{Channel }  \ \ \input{Mxy.tikz}  \iff \text{ Probability channel } M(Y \mid X)
\]
Channels $M$ with multiple classical inputs $X_1, \dots, X_n$ and outputs $Y_1,\dots, Y_m$ correspond to probability channels $M(Y_1,\dots,Y_m \mid X_1,\dots,X_n)$ in the same way. For a channel $M(Y, Z \mid X)$, discarding output $Y$ corresponds to marginalisation over $Y$ in the usual sense:
\[
\input{margMxy.tikz} \ = \  \sum_{y \in Y} M(y,z \mid x)
\]

Note that we can make Definition \ref{def:prob-PT} more axiomatic as follows, proven in the Appendix.\footnote{More strongly than the probabilistic structure we have assumed, we could have alternatively required scalars to correspond to $\mathbb{R}^{\geq 0}$,  and that $\catC$ has \emph{addition}, meaning it is monoidally enriched in commutative monoids via an addition operation $f, g \mapsto f + g$ on morphisms $f, g \colon A \to B$. Categories with essentially this structure are considered in \cite{gogioso2017categorical} and allow working with `super-normalised' processes, while our approach allows our categories to contain only the more operationally meaningful scalars $[0,1]$ and sub-normalised processes, as in \cite{chiribella2010probabilistic,cho2015introduction,tull2016operational}.} 

\begin{proposition} \label{prop:probPT}
A process theory $\catC$ is
probabilistic iff it satisfies axiom \eqref{enum:scalars-are-probs} of Def.~\ref{def:prob-PT} along with the following, for all classical objects $X$.
\begin{enumerate}
	\setcounter{enumi}{\value{nameOfYourChoice}}
\item  \label{enum:joint-monic}
For all morphisms $\morf, \morg$ as below we have:
 \footnote{This states that the morphisms $x \otimes \id{B}$ are jointly monic as projections $X \cdot B \to B$, as in an effectus in partial form \cite{cho2015introduction}; see also \cite{tull2016operational}. }
\begin{equation} \label{eq:jointmonic}
\input{xafterfg.tikz} \ \  \forall x \in X
\quad \implies 
\quad
f = g
\end{equation}
\item \label{enum:sum-rule}
For every state $\omega$ of $X$ we have:
\[
\begin{tikzpicture}
	\begin{pgfonlayer}{nodelayer}
		\node [style=cmap] (0) at (-3, -0.75) {$\omega$};
		\node [style=upground] (1) at (-3, 0.75) {};
		\node [style=label] (4) at (-3.75, 0) {$X$};
	\end{pgfonlayer}
	\begin{pgfonlayer}{edgelayer}
		\draw [style=cwire] (0) to (1);
	\end{pgfonlayer}
\end{tikzpicture}
   \ \ = \ \  \sum_{x \in X} \  \ \omega(x)
\]
\item  \label{enum:normalisation}
Every non-zero state $\omega$ of $X$ has $\omega = (\discard{} \circ \omega) \cdot \normop(\omega)$ for some normalised state $\normop(\omega)$.
\item \label{enum:unif-state}
There exists a state $\discardflip{}$ of $X$ 
with $x \circ \discardflip{} = \frac{1}{|X|}$ for all $x \in X$. 
\end{enumerate} 
\end{proposition}

Intuitively, these conditions respectively state that each classical object $X$ is determined by `testing for' each element $x \in X$, that $\discard{}$ relates to marginalisation, that we can normalise sub-probability distributions, and that there exists a uniform distribution over each $X$.

\subsection{Examples}

Let us now meet some examples of process theories. We begin by describing a category of classical probabilistic processes. In fact we have essentially just met this example above, which simply amounts to working only with classical objects.

\begin{example}[\Class] \label{ex:Class-cat}
In the category $\Class$ of (finite) classical processes, the objects are finite sets $X, Y, \dots$ and every such object is a classical object. A morphism $M \colon X \to Y$ is thus given by a `matrix' $M \colon X \times Y \to [0,1]$ with values $M(y \mid x) \in [0,1]$ for $x \in X, y \in Y$. Composition of $M \colon X \to Y$ and $N \colon Y \to Z$ is matrix multiplication:
\[
(N \circ M)(z \mid x) := \sum_{y \in Y} N(z \mid y) M(y \mid x)
\]
The identity morphism on $X$ is the matrix with $\id{X}(y \mid x) = \delta_{x,y}$. The monoidal product is given on objects by $X \otimes Y = X \times Y$, and on morphisms $M \otimes N$ by the Kronecker product of matrices:
\[
(M \otimes N)(w,z \mid x ,y) = M(w \mid x)N(z \mid y)
\]
Here $I=\{\star\}$ is a singleton set. Discarding $\discard{} \colon X \to I$ is given by $\discard{X}(\star \mid x) = 1$ for all $x \in X$. 

Exactly as above, for each $x \in X$ there is a normalised state of $X$ corresponding to the point distribution $x(y \mid \star) = \delta_{x,y}$, as well as $x \colon X \to I$ given by $x(\star \mid y) = \delta_{x,y}$. Probabilities $M(y \mid x)$ are given just as in \eqref{eq:class-values}, 
normalised states  are probability distributions, channels $M \colon X \to Y$ are probability channels and $\discard{}$ amounts to marginalisation. Any probabilistic process theory $\catC$ contains $\Class$ as its sub-category of morphisms between classical objects, as we prove in Lemma \ref{lem:full-faithful-embedding} in Appendix \ref{app:proofs}. 
\end{example}

Next, we describe quantum processes. In fact we will do so via a category which includes both classical and quantum objects and processes.

\begin{example}[\QC] \label{ex:QC-catgeory}
In the category $\QC$ the objects are pairs $(\hilbH, X)$ where $\hilbH$ is a finite-dimensional complex Hilbert space (a quantum system) and $X$ is a finite set (a classical system). We draw quantum systems with thick wires, and classical objects with thin wires, as per our convention. A morphism $f \colon (\hilbH, X) \to (\hilbK, Y)$, drawn as: 
\begin{equation} \label{eq:QC-morphism}
\input{q-mor.tikz}
\end{equation}
is given by a collection of \emph{completely positive} linear maps:\footnote{Recall that a linear map $f \colon L(\hilbH) \to L(\hilbK)$ is positive when it sends positive operators $a=b^\dagger b \in L(\hilbH)$ to positive operators $f(a)$ in $L(\hilbK)$, and is completely positive when each map $f \otimes \id{\hilbH'}$ is also positive, for any other Hilbert space $\hilbH'$. We say $f$ is trace non-increasing when $\Tr(f(a)) \leq \Tr(a)$ for all $a$, and trace preserving when $\Tr(f(a)) = \Tr(a)$ for all $a$.} 
\[
\left(f(y \mid x) \colon L(\hilbH) \to L(\hilbK)\right)_{x \in X, y \in Y}
\]
such that $\sum_{y \in Y} f(y \mid x)$ is trace non-increasing for each $x \in X$. Here $L(\hilbH)$ denotes the set of linear operators on $\hilbH$.  Given $g \colon (\hilbK, Y) \to (\hilbW, Z)$ we define composition $g \circ f \colon (\hilbH, X) \to (\hilbW, Z)$ \rl{via}\ summation: 
\[
(g \circ f)(z \mid x) := \sum_{y \in Y} g(z \mid y) \circ f(y \mid x)
\]
We define $(\hilbH, X) \otimes (\hilbK, Y) := (\hilbH \otimes \hilbK, X \times Y)$ where $\hilbH \otimes \hilbK$ is the tensor product of Hilbert spaces, and on morphisms define $(f \otimes g)(y, z \mid x, w) := f(y \mid x) \otimes g(z \mid w)$, with unit object $I=(\mathbb{C},\{\star\})$. 
Classical processes correspond to those on thin wires only. For each finite set $X$ the corresponding classical object is $X=(\mathbb{C},X)$. A process between classical objects:
\[
\input{MXY.tikz}
\]
taking the form $(M(y \mid x) \colon \mathbb{C} \to \mathbb{C})_{x \in X, y \in Y}$, corresponds to a morphism $M \colon X \to Y$ in $\Class$, with values $(M(y\mid x))_{x \in X, y \in Y}$.\footnote{Here we view each value $M(y\mid x)$ as a function on $\mathbb{C}$ via $1 \mapsto M(y \mid x)$.}
There are also `entirely quantum' processes featuring thick wires only. We view any finite-dimensional Hilbert space $\hilbH$ as an object via $(\hilbH, \{\star\})$ where $\{\star\}$ is a singleton set. A morphism:  
\[
\input{fHK.tikz}
\]
is then precisely a completely positive trace non-increasing map $f \colon L(\hilbH) \to L(\hilbH)$. For any general morphism as in \eqref{eq:QC-morphism} its individual CP maps are given by composition: 
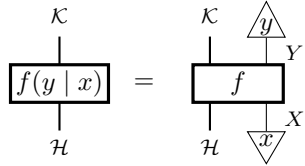
\[
\input{qmorxy.tikz}
\]
Within the quantum processes, of particular interest are the \emph{pure} quantum processes. For any linear map $f \colon \hilbH \to \hilbK$ with $\|f\| \leq 1$ there is an induced CP map $\hat f \colon L(\hilbH) \to L(\hilbK)$ given by $a \mapsto f a f^\dagger$. In particular, each unit vector $\psi \in \hilbH$ induces a normalised state $\hat \psi$ of $\hilbH$, whose density matrix is usually denoted with `Bra-ket' notation as $\ket{\psi}\bra{\psi}$. In diagrams we draw such a state simply as:
\[
\begin{tikzpicture}
	\begin{pgfonlayer}{nodelayer}
		\node [style=none] (0) at (0, 1) {};
		\node [style=point] (1) at (0, -0.25) {$\psi$};
		\node [style=label] (2) at (0, 1.5) {$\hilbH$};
	\end{pgfonlayer}
	\begin{pgfonlayer}{edgelayer}
		\draw (0.center) to (1);
	\end{pgfonlayer}
\end{tikzpicture}

\]

Discarding $\discard{} \colon (\hilbH, X) \to (\mathbb{C}, \{\star\})$ is given by setting $\discard{}(\star \mid x) \colon L(\hilbH) \to \mathbb{C}$ as the map $a \mapsto \Tr(a)$ for $a \in L(\hilbH)$, for all $x \in X$. Thus discarding a classical wire $X$ corresponds to marginalisation (summation over $X$) while discarding a quantum wire $\hilbK$ corresponds to taking the (partial) trace over $\hilbK$. 
\[
\input{ptrace.tikz}
\]
A morphism $f \colon \hilbH \to \hilbK$ between quantum objects is a channel precisely when it is \emph{CPTP (trace-preserving)} i.e.~
\[
\Tr(f(a)) = \Tr(a)
\] 
for all $a \in L(\hilbH)$. A normalised state $\rho$ of $\hilbH$ is then equivalent to a \emph{density matrix}, i.e.~a trace 1 positive operator $\rho \in L(\hilbH)$. 
\[
\text{Normalised} \ \ \ \ \begin{tikzpicture}
	\begin{pgfonlayer}{nodelayer}
		\node [style=none] (0) at (1, 1.25) {};
		\node [style=map] (1) at (1, 0) {$\rho$};
		\node [style=label] (2) at (1, 1.75) {$\hilbH$};
	\end{pgfonlayer}
	\begin{pgfonlayer}{edgelayer}
		\draw (0.center) to (1);
	\end{pgfonlayer}
\end{tikzpicture}
 \  \iff  \ \text{Density matrix $\rho \in L(\hilbH)$} 
\]
\end{example}

\begin{remark}[$\FCStar$]
Though we will not require it for this work, we can unify classical and quantum objects and processes in a more general way via the probabilistic process theory $\FCStar$ of finite-dimensional C*-algebras and trace non-increasing CP maps. There is an embedding $\QC \hookrightarrow \FCStar$ with classical objects $X$ corresponding to commutative algebras $\mathbb{C}^X$ and sending quantum objects $\hilbH$ to algebras $L(\hilbH)$, with more general algebras $\bigoplus_{x \in X} L(\hilbH_x)$ interpretable as `mixed' classical-quantum systems.
\end{remark}

\section{Decision Models} \label{sec:decision-models}

Let us now discuss how process theories can be used to model psychological data. We will consider the modelling of experimental scenarios in which participants are understood to make sequences of cognitive \emph{decisions}, as follows. 

\begin{definition}[Decision Data]
By a set of \emph{decisions} we simply mean a set $\decset$ whose elements we call decisions, where each decision $\decA \in \decset$ is given along with a finite set of possible \emph{outcomes} denoted $X_A$.  

A set of (sequential) \emph{decision data} $P$ for $\decset$ is given by a set of finite sequences $(A_1,\dots,A_n)$ of decisions from $\decset$, 
and for each such sequence a probability distribution $\PAton$ over $X_{A_1} \times \dots \times X_{A_n}$. We then denote the probability of an outcome sequence $(x_1,\dots,x_n) \in X_{A_1} \times \dots \times X_{A_n}$ by $\PAxton$.
\end{definition}

In a set of decision data, we interpret each sequence $(A_1,\dots,A_n)$ as an experimental scenario in which the participant is asked to make the given sequence of decisions in the order specified: $A_1$ followed by $A_2$ and so on, up to $A_n$. For each such sequence we have experimental data given by the distribution $\PAton$, specifying the relative probabilities of each possible sequence of outcomes. Each such distribution is denoted diagrammatically as a normalised state:
\[
\input{datadist.tikz}
\]
Thus any decision data is itself fully classical, being a state of the product of the classical objects $X_{A_1}, \dots X_{A_n}$. Many examples of decisions can be simply thought of as \emph{questions} asked to the participant's, with each $\PAxton$ being the probability a participant answers $x_1$ to question $A_1$, then $x_2$ to question $A_2$ and so on.

\begin{dataexample}[Clinton-Gore Data] \label{ex:Clinton-Gore}
A well-known example of question-style decision data is given by Moore \cite{moore2002measuring}. During a Gallup poll conducted during September 6-7, 1997, in which 1002 respondents were asked the questions $A$ = `Do you generally think Bill Clinton is honest and trustworthy?' and $B$ = `Do you generally think Al Gore is honest and trustworthy?'. Half of the participants were asked the questions in one order, and half in the other. 

In this case, the decisions are given by the two questions $\decset=\{A,B\}$, each with outcomes $X_A=X_B=\{y,n\}$ for `yes' and `no'. The decision data consists of data for the two question orders $(A,B)$ and $(B,A)$, and so is given by a distribution $\PAB$ over $X_A \times X_B$ and a distribution $\PBA$ over $X_B \times X_A$. Concretely, the data from the polls is as follows. 
\[
\begin{array}{l|lllll}
& AyBy  & AyBn  & AnBy  & AnBn  &  \\ \cline{1-5}
P & 0.4899 & 0.0447 & 0.1767 & 0.2887 &  
\end{array}
\qquad 
\begin{array}{l|lllll} 
& ByAy  & ByAn  & BnAy  & BnAn  &  \\ \cline{1-5}
P & 0.5625 & 0.1991 & 0.0255 & 0.2129 &  
\end{array}
\]
\end{dataexample}

Now let us describe what it will mean to model such decision data. The key idea will be to think of each participant as having an object $S$ describing their mental state, and each individual decision $A \in \decset$ as a process $A \colon S \to X_A \otimes S$ which acts on $S$, outputting both their reported outcome in the finite set $X_A$ as well as a new revised mental state in $S$. Throughout, let $\catC$ be a probabilistic process theory.

\begin{definition}[Instruments] \label{def:instrument}
Let $\cob$ be a finite set and $\obI$ an object in $\catC$. An \emph{instrument} $\instI$ on $\obI$ with outcomes $\cob$ is a channel in $\catC$ of the form: 
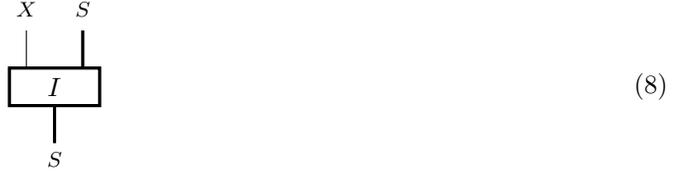
\begin{equation} \label{eq:instrument}
\input{inst.tikz}
\end{equation}
\end{definition}

Any such instrument describes a general process on $S$ which yields a (classical) outcome from the finite set $X$. Thanks to  \eqref{eq:jointmonic} it is fully determined by its collection of processes $(\instI(x) \colon \obI \to \obI)_{x \in \cob}$, where: 
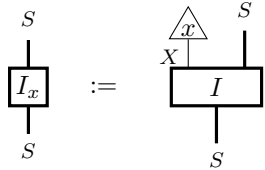
\[
\input{instcorresp.tikz}
\]
 Here $\instI_x$ can be thought of as the state change process occurring when the outcome is $x \in X$. We will meet many examples of instruments relevant to decision modelling shortly. 

We can now formally describe decision models in process theories. 

\begin{definition}[Instrument Model]
An \emph{instrument model} $\modelM$ of a set of decisions $\decset$ in a probabilistic process theory $\catC$ is given by specifying the following in $\catC$:
\begin{itemize}
\item an object $S$, called the \emph{state space}; 
\item a normalised state $\rho$ of $S$, called the \emph{initial state}; 
\item an instrument $A$ with outcomes $X_A$ for each decision $A \in \decset$ with outcomes $X_A$.
\end{itemize}
For each sequence of decisions from $\decset$ the induced distribution 
$\PMAton$ over $X_{A_1} \times \dots \times X_{A_n}$ is then defined by:
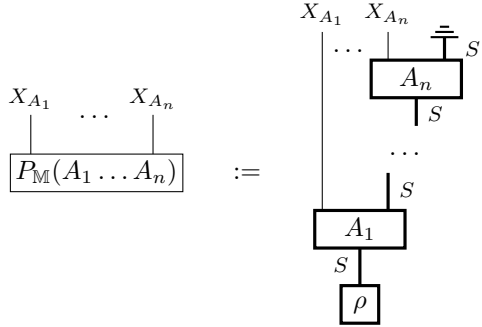
\[
\input{model-dec-data-smaller.tikz}
\]
We say that $\modelM$ is a \emph{model of} a given set of decision data $P$ over $\decset$ when for every such sequence in the data we have $\PAton = \PMAton$. 
\end{definition}

An instrument model thus specifies an object $S$ for the participants mental state, an initial state $\Mo$, and an instrument for each decision $A \in \decset$, such that these same instruments return the correct probability distributions for all of the different sequences of decisions in the data. We denote such a model of decisions $\decset=\{A_1,\dots,A_n\}$ by $\modelM=(S,\rho,A_1,\dots,A_n)$. 

\begin{example}[Modelling a pair of questions] \label{ex:C-G-modelling}
An instrument model $\modelM$ of the Clinton-Gore question data from Example \ref{ex:Clinton-Gore} in $\catC$ would be given by an object $S$, a normalised state $\rho$ of $S$ and a pair of instruments $A, B$ with outcomes $X_A=X_B=\{y,n\}$ which yield the data distributions $\PAB$ and $\PBA$, i.e.:
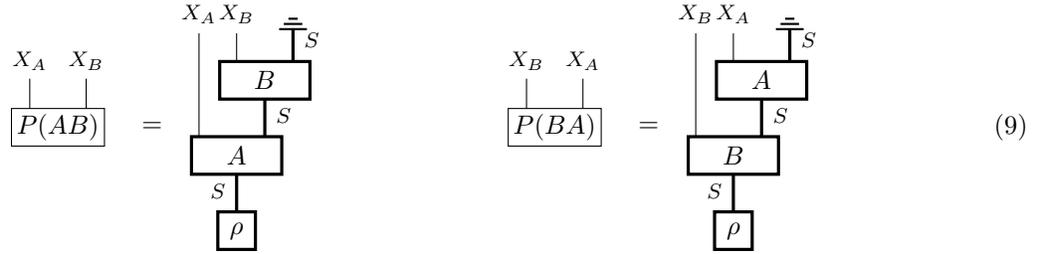
\begin{equation} \label{eq:AB-BA-model}
\input{A-B-model.tikz} \qquad \qquad \qquad 
\input{B-A-model.tikz}
\end{equation}
\end{example}

We will be interested in comparing various kinds of both classical and quantum instrument models of decision data. Let us now introduce several important classes of such models. 

\subsection{Classical models} \label{sec:class-insts}

By a \emph{classical} instrument model we mean an instrument model in $\catC=\Class$. Thus such a model $\modelM$ consists of a finite set $S$, an initial probability distribution $\rho$ over $S$, and for each decision $A \in \decset$ an instrument of the following form.

\begin{example}[Classical Instruments]
By a \emph{classical instrument} we mean an instrument $\instAI$ in $\Class$ on a finite set $\cHi$, with finite set of outcomes $\XAI$. 
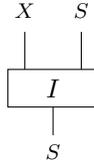
\[
\input{cinst.tikz}
\]
A classical instrument is thus precisely a probability channel: 
\[
\instAI(\XAI, \cHi \mid \cHi)
\] 
from $\cHi$ to $\XAI, \cHi$.
The instrument thus sends each element $\cst \in \cHi$, thought of as the current state $\cst$ of $\cHi$, to a joint distribution
$\instAI(\XAI, \cHi \mid \cHi = \cst)$ over outcomes $x \in \XAI$ and next states $\cst' \in \cHi$.
\end{example}

When classical models are studied in the quantum-like cognition literature, however, it is typically only the following kind which are considered.

\begin{example}[Bayesian Instruments]
By a \emph{Bayesian instrument} on $\cHi$ we mean a classical instrument $\instAI$ of the form:
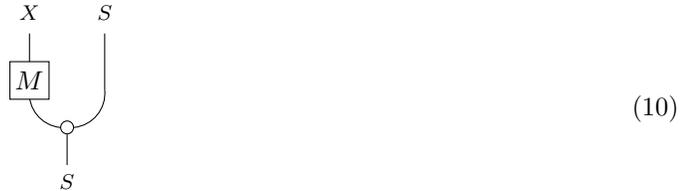
\begin{equation} \label{eq:bayes-inst}
\input{bayes-simple.tikz}
\end{equation}
where $\bayesmAI \colon \cHi \to \XAI$ is a probability channel. Equivalently, we have:
\[
\instAI(x, \cst' \mid \cst) := \bayesmAI(x \mid \cst) \delta_{\cst, \cst'}
\] 
Thus, a Bayesian instrument copies the state $\cst$ on $\cHi$ and leaves it undisturbed, simply recording a probabilistic outcome $x \in \XAI$ with probability $\bayesmAI(x \mid \cst)$. We might think of such an instrument as simply a conditional distribution $P(\XAI \mid S) := \bayesmAI(\XAI \mid S)$ of outcomes $\XAI$ given states $\cHi$. 

By a \emph{Bayesian model} we mean a classical instrument model all of whose instruments are Bayesian.  
\end{example}

A simple class of classical instruments which are not Bayesian are the following. Firstly, we view any function $f \colon X \to Y$ between finite sets as a probability channel $X \to Y$ given by $f(y \mid x) := \delta_{f(x),y}$, i.e.~sending each $x \in X$ to (the point distribution at) $f(x) \in Y$. We call a probability channel \emph{deterministic} when it is of this form, for some function $f \colon X \to Y$.

\begin{example}[Deterministic Instruments]
By a \emph{deterministic instrument} $\instAI$ on $\cHi$ we mean a classical instrument which is itself deterministic, or equivalently of the form:
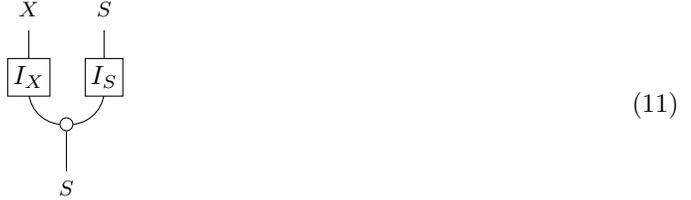
\begin{equation} \label{eq:det-inst}
\input{det-inst.tikz}
\end{equation}
where $\detAIx$ and $\detAIs$ are deterministic.
Thus such an instrument is given simply by a function $\detAIx \colon \cHi \to \XAI$ determining the outcome from the current state, and a function $\detAIs \colon \cHi \to \cHi$ determining the next state from the previous one. By a \emph{deterministic model} we mean a classical instrument model all of whose instruments are deterministic. 

A special case are deterministic Bayesian instruments, which fit both \eqref{eq:bayes-inst} and \eqref{eq:det-inst} with $\detAIx = \bayesmAI$ and $\detAIs = \id{}$. Any such instrument simply assigns a deterministic output $\bayesmAI(\cst)$ to each state $\cst \in \cHi$ and leaves the state undisturbed.

\end{example}

An important class of classical instruments featuring state updates and used widely in modelling are the following. 

\begin{example}[Markov Instruments] \label{ex:Markov-instruments}
By a \emph{Markov instrument} $\instAI$ on $\cHi$ we mean a classical instrument of the form: 
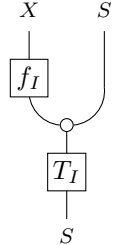
\[
\input{markovonly.tikz}
\]
where $\markovfAI$ and $\markovTAI$ are probability channels. 
Equivalently, we have 
\[
\instAI(x, s' \mid s) :=  \markovfAI(x \mid s') \markovTAI(s' \mid s)
\]
Thus such an instrument includes a \emph{transition} probability channel $\markovTAI \colon \cHi \to \cHi$ which first probabilistically updates the hidden state, and an \emph{observation} probability channel $\markovfAI \colon \cHi \to \XAI$ which determines the outcome.  By a \emph{Markov model} we mean a classical instrument model all of whose instruments are Markov instruments. By a \emph{simple Markov model} we mean one where $S$ comes with a distinguished element $\star$ and we have initial state $\rho = \delta_\star$ given by a point distribution (at $\star$), and such that every decision $A$ has the same outcome set $X$ and has $\markovfA = f$ for a fixed deterministic channel $f \colon S \to X$. 
\end{example}

We see from the above examples that most kinds of classical instrument (general, deterministic, Markov) do in fact allow for interaction with the hidden state $S$, while it is only Bayesian instruments which always leave it undisturbed. Hence there is nothing inherently `non-classical' about decision models which can alter this hidden state. 

\subsection{Quantum models} \label{sec:quant-insts}

Next we consider quantum models of decisions. By a \emph{quantum} instrument model we mean an instrument model in $\catC=\QC$ with $S=\hilbH$ for a finite-dimensional Hilbert space $\hilbH$. Thus, along with $\hilbH$, such a model specifies a density matrix $\rho \in L(\hilbH)$ for the initial state, and for each decision $A \in \decset$ an instrument of the following form.

\color{black}
\begin{example}[Quantum Instruments]
Let $\hilbH$ be a finite-dimensional Hilbert space. 
By a \emph{quantum instrument} on $\hilbH$ we mean an instrument $\instAI$ on $\hilbH$ in $\QC$ with some set of outcomes $\XAI$. 
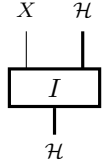
\[
\input{qinst.tikz} 
\]
Such an instrument is thus given by a family of CP maps: 
\[
(\instAI_x \colon L(\hilbH) \to L(\hilbH))_{x \in \XAI}
\]
whose sum $\sum_{x \in \XAI} \instAI_x$ is trace-preserving (CPTP). This is precisely the usual notion of instrument in quantum information, describing a general operation on a quantum system.\footnote{We can now describe arbitrary instruments in $\QC$. An instrument $\instAI$ on $(\hilbH, S)$ in $\QC$ with outcomes $X$ is given by a channel $\instAI \colon \hilbH \otimes S \to X \otimes S \otimes \hilbH$ in $\QC$. Thus, for each $s \in S$ it specifies an instrument $\instAI^s$ on $\hilbH$ with outcomes $X \times S$.} 
\end{example}

Within the quantum-like cognition literature, quantum models are however more usually taken to be of the following more specific kind. Recall that a \emph{projective measurement} on $\hilbH$ with outcomes $x \in X$ is given by a family $P=(P_x)_{x \in X}$ of positive operators $P_x \in L(\hilbH)$ which form projections, i.e.~each $P_x = P^\dagger_x P_x$, and satisfy $1_\hilbH = \sum_{x \in X} P_x$.

\begin{example}[Projective Instruments] \label{ex:proj-insts}
Let $\hilbH$ be a finite-dimensional Hilbert space. 
A \emph{projective instrument} on $\hilbH$ with outcomes $\XAI$ is a quantum instrument $\hat \projPAI$ on $\hilbH$ for which each $\hat \projPAI_x$ is the pure CP map induced by $\projPAI_x$, for some projective measurement $(\projPAI_x)_{x \in \XAI}$ on $\hilbH$. 
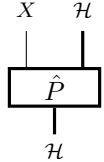
\[
\input{projpic.tikz} 
\]
That is, we have: 
\begin{equation} \label{eq:proj-inst-eq}
\hat \projPAI_x(a) = \projPAI_x a \projPAI_x
\end{equation}
for each $a \in L(\hilbH)$, $x \in \XAI$. 
An important special case of projective instruments are orthonormal \emph{basis measurements}, which take the form  $\projPAI_x = \ket{\psi_x}\bra{\psi_x}$ where $\{\psi_x\}_{x \in \XAI}$ is an orthonormal basis of $\hilbH$. In diagrams we have:
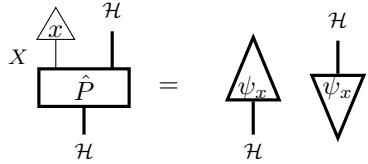
\[
\input{ONBpic.tikz}
\]
where the upward $\psi_x$ triangle denotes the CP map $\Tr(\ket{\psi}\bra{\psi}-) \colon L(\hilbH) \to \mathbb{C}$. 

By a \emph{projective instrument model} we mean a quantum instrument model all of whose instruments are projective instruments. 
\end{example}

The previous example in fact generalises beyond projective measurements to a more general form of quantum measurements. Recall that a \emph{positive operator valued measure (POVM)} on a finite-dimensional Hilbert space $\hilbH$ is given by a collection $(\povm_x)_{x \in X}$ of positive operators $\povm_x \in L(\hilbH)$ with $\sum_{x \in X} \povm_x = 1_\hilbH$.

\begin{example}[POVM Instruments] \label{ex:POVM-insts}
Given any POVM $(\povmAI_x)_{x \in \XAI}$ on $\hilbH$, we define a corresponding \emph{POVM instrument} $\upd{\povmAI}$ on $\hilbH$ with outcomes $\XAI$:
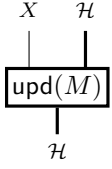
\[
\input{povminst.tikz}
\]
given by:
\begin{equation} \label{eq:update-povm}
\upd{\povmAI}_x(a) := \sqrt{\povmAI_x} a \sqrt{\povmAI_x}
\end{equation}
for $a \in L(\hilbH)$, $x \in \XAI$. Here 
$\sqrt{\povmAI_x}$ denotes the unique positive operator $b$ with $\povmAI_x = b^2$. 
Any projective measurement $P=(P_x)_{x \in X}$ is a special case of a POVM, and since $\sqrt{P_x} = P_x$, we indeed have $\hat P = \upd{P}$ as instruments. 
\end{example}

Note that while every quantum instrument defines a POVM (via the trace) and each POVM defines a canonical instrument as above, the correspondence is not one-to-one and quantum instruments can be more general than POVM instruments; we return to this in Section \ref{sec:meas-models}.

\section{Cognitive Effects} \label{sec:cog-effects}

In this section we will explore several effects in decision-making which have appeared in the literature on quantum-like modelling, with examples taken from \cite{busemeyer2012quantum}, at times in adapated form from \cite{jacobs2017quantum}. Throughout, $P$ denotes decision data for a pair of decisions $\decset=\{\decA, \decB\}$, and 
$\MABrho$ a candidate model consisting of instruments $A, B$ and initial state $\rho$ on $S$.

\subsection{Order effects} \label{subsec:order-effect}

Suppose we ask participants to answer a pair of questions $\decA, \decB$. If their answers simply amounted to verbalising some pre-determined information, even probabilistically, then the \emph{order} in which we ask the questions should make no difference to the final probabilities. However in real life data, question responses are frequently affected by those asked previously. This is one of the simplest and first cognitive effects used to argue for quantum models \cite{moore2002measuring,
aerts1995applications,khrennikov2004information,conte2009mental}.\footnote{Moore \cite{moore2002measuring} in fact introduces a number of classes of order effect, including those corresponding to what we here later term an `interference effect'.}

\begin{definition}[Order Effect] \label{def:order-effect}
We say that decision data satisfies an 
\emph{order effect} when: 
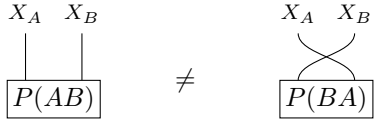
\[
\input{data-QOE-simple.tikz}
\]
That is, there exist outcomes $a \in X_A, b \in X_B$ such that $\PAaBb \neq \PBbAa$. Similarly, a pair of instruments $A, B$ on $S$ satisfy an \emph{order effect in state $\rho$} when: 
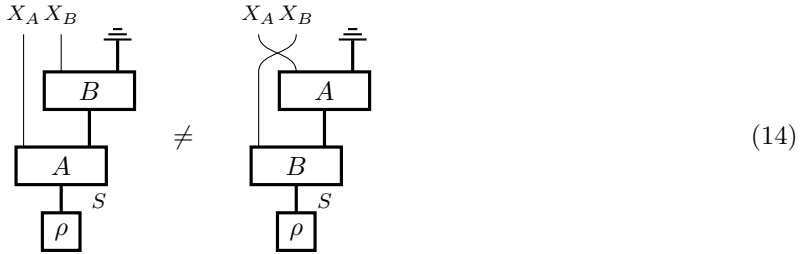
\begin{equation} \label{eq:OE-state}
\input{OE-state.tikz}
\end{equation}
Equivalently, the data $P_\modelM$ induced by the model $\MABrho$ exhibits an order effect. We simply say that instruments $A, B$ \emph{satisfy an order effect} when the same inequality holds but now without $\rho$ specified, so that each diagram has $S$ as an input (as in \eqref{eq:effects-simple}). 
We say that instruments 
$A, B$ on $S$ \emph{commute} when the converse holds even without discarding $S$, that is: 
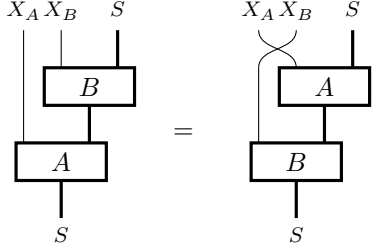
\[
\input{commute.tikz}
\]
Otherwise we say that $A, B$ are \emph{non-commutative}.
\end{definition}

Clearly, to exhibit any order effects, instruments must be non-commutative. A famous example of decision data satisfying an order effect is the following. 

\begin{dataexample}
The Clinton-Gore poll data (Example \ref{ex:Clinton-Gore}) exhibits an order effect. For example, $\PAyBy = 0.4899 \neq \PByAy = 0.5625$. 
\end{dataexample}

\begin{proposition} \label{prop:OE}  \ 
\begin{enumerate}
\item \label{enum:Bayes-commute}
Any pair of Bayesian instruments commute. 
\item \label{enum:detorder}
Deterministic instruments $\dI, \dJ$ on $\cHi$ exhibit an order effect iff $\dJx \circ \dIs \neq \dJx$ or $\dIx \circ \dJs \neq \dIx$. 
\item \label{enum:noOE-projec}
For projective instruments $\hat P, \hat Q$ the following are equivalent:
\begin{enumerate}
\item \label{enum:pqcom}
$\hat P, \hat Q$ commute; 
\item \label{enum:pqnooe}
$\hat P, \hat Q$ exhibit no order effect in any state $\rho$;
\item \label{enum:pqmeascom}
the projective measurements $P = (P_x)_{x \in X}$, $Q = (Q_y)_{y \in Y}$ commute, i.e.~$P_x Q_y = Q_y P_x \forall x, y$. 
\end{enumerate}
\end{enumerate} 
\end{proposition}
\begin{proof}
\eqref{enum:Bayes-commute}:
Let $\instone, \insttwo$ be Bayesian instruments. Then we have the following.
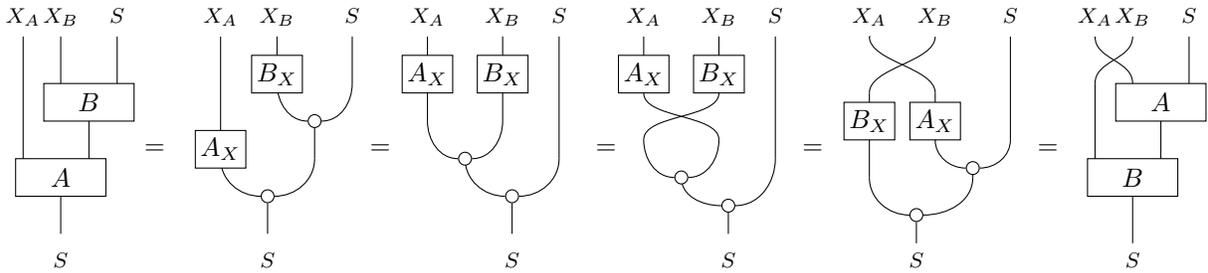
\[
\input{Bayesian-commute.tikz}
\]
\eqref{enum:detorder}: 
An order effect fails to hold iff: 
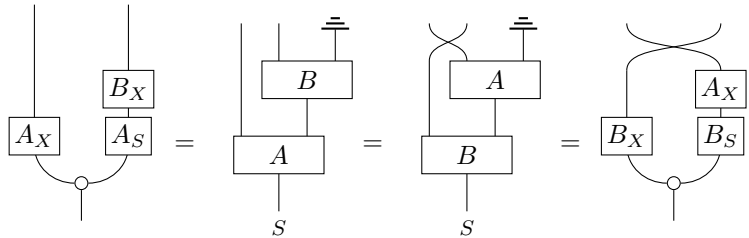
\[
\input{detorder2.tikz}
\]
which holds iff both $\dJx \circ \dIs = \dJx$ and $\dIx \circ \dJs = \dIx$. 

\eqref{enum:pqcom} $\implies$ \eqref{enum:pqnooe}: trivial. 

\eqref{enum:pqnooe} $\implies$ \eqref{enum:pqmeascom}: Let $\psi$ be an arbitrary pure state, and suppose that $\hat P, \hat Q$ exhibit no order effect on $\psi$. Then
\[
\|Q_yP_x\psi\|^2 = \Tr(Q_yP_x\ket{\psi}\bra{\psi}P_xQ_y) = \discard{} \circ \hat Q_y \circ \hat P_x \circ \ket{\psi}\bra{\psi} 
= \|P_xQ_y\psi\|^2
\]
In particular whenever $\psi$ is an eigentstate of $P_x$, with $P_x \psi = \psi$ we obtain $\|Q_y \psi \| = \|P_x(Q_y\psi)\|$ and so $Q_y \psi$ is again an eigenstate. It follows that $P_x, Q_y$ share all eigenstates and hence commute as projection. 

\eqref{enum:pqmeascom} $\implies$ \eqref{enum:pqcom}: 
For any state $\rho$ we have $\hat P_x \hat Q_y \circ \rho = P_x Q_y \rho Q_y P_x = Q_y P_x \rho P_x Q_y = \hat Q_y \hat P_x \rho$. 
\end{proof}

Thanks to the first result above, order effects cannot be modelled using straightforward probabilistic logic, i.e.~classical Bayesian instruments. Indeed any Bayesian model amounts to a single joint probability distribution $\PMAB = \PMBA$, given by: 
\begin{equation} \label{eq:Bayesian-induced-joint}
\input{Bayes-dist.tikz}
\end{equation}
For any $a \in X_A, b \in X_B$  we then have $\PMAaBb = \PMBbAa$. 

In contrast, thanks to result \eqref{enum:noOE-projec} above, most projective quantum models will be non-commutative and exhibit order effects, and indeed this served as an initial motivation for quantum models themselves. Note however that, by \eqref{enum:detorder} above, when making use of more general non-Bayesian classical instruments even simple determistic ones can exhibit order effects. 

\begin{example}
Consider the following simple projective quantum model exhibiting an order effect from \cite[Chapter 3]{busemeyer2012quantum}. Set $\hilbH=\mathbb{C}^2$ to be a two-dimensional Hilbert space, and set $A, B$ to be orthonormal bases measurements where 
\[
\ket{Ay} = \frac{1}{\sqrt{2}}(\ket{By} + \ket{Bn})
\qquad
\ket{An} = \frac{1}{\sqrt{2}}(\ket{By} - \ket{Bn})
\]
We set the initial state to be the pure state $\rho = \ket{s}\bra{s}$ where $\ket{s} = 0.8367 \ket{By} + 0.54777 \ket{Bn}$. This model shows an order effect since: 
\begin{align*}
\PAyBn &= | \braket{s \mid Ay}|^2 |\braket{Ay \mid Bn}|^2
 =(1/2)(1/2(0.8367 + 0.5477)^2) = 0.48 
\\ 
\PBnAy &= | \braket{s \mid Bn}|^2 |\braket{Bn \mid Ay}|^2
=(0.5477^2)(1/2) = 0.15 
\end{align*}
\end{example}

\begin{example}[Deterministic Order Effects] \label{ex:simple-det-OE}
Set $S=X_A=X_B=\{y,n\}$ and define deterministic instruments: 
\[
\input{simpledetOE.tikz}
\]
where $\bot$ is the function $y \mapsto n, n \mapsto y$. Then $A, B$ exhibit an order effect in the state $\rho = \delta_y$ given by the point distrubution on $y$. Defining $\MABrho$ we then have $\PMAyBy = 1 \neq 0 = \PMByAy$.
\end{example}

\subsection{Interference effects} \label{subsec:interference-effect}

A yet stronger form of cognitive effect is the following. In cognitive experiments, it is frequently observed that the probabilities of the outcomes of a decision $\decB$ differ in the case where we have first carried out a decision $\decA$, as follows.

\begin{definition}[Interference Effect]
Given decision data $P$, we say that decision $\decA$ \emph{interferes with} decision $\decB$ when:
\[
\input{data-interference.tikz}
\]
That is, $\sum_{a \in X_A} \PAaBb \neq \PBb$ for some $b \in X_B$. Similarly, for instruments $A, B$ on $S$, we say that \emph{$A$ interferes with $B$ in a state $\rho$} when we have:
\begin{equation} \label{eq:interference-rho}
\input{interference-rho.tikz}
\end{equation}
Equivalently, $\decA$ interferenes with $\decB$ in the data $P_\modelM$ for the model $\MABrho$. 
We simply say that $A$ \emph{exhibts an interference effect} on $B$ when the same inequality holds but now without $\rho$ specified, so that each diagram has $S$ as an input (as in \eqref{eq:effects-simple}). 
\end{definition}

This notion is not to be confused with `interference' in the sense of superpositions in quantum theory, and indeed we will see that even classical models can exhibit such effects. When discussing data $P$ exhibiting interference, the quantity 
\[
P_T(Bb) := \sum_{a \in X_A} \PAaBb 
\] 
is often referred to as the \emph{total probability} of $b \in X_B$, as opposed to the probability $\PBb$ for decision $B$ on its own. Then interference holds when that $P_T(Bb) \neq \PBb$ for some $b \in X_B$. Thus, the outcome probabilities for $B$ are altered by the earlier decision $A$, even when ignoring the actual outcome of $A$ by marginalising (i.e. discarding).
Interference effects are stronger than order effects, as follows.

\begin{proposition} \ \label{prop:interference} 
\begin{enumerate}
\item  \label{enum:inttoorder}
For instruments $A, B$, an interference effect (in state $\rho$) implies an order effect (in state $\rho$). 
\item \label{enum:commutenoint}
Commuting instruments exhibit no interference effects. 
\item \label{enum:Bayesnoint}
Bayesian models exhibit no interference effects. 
\item  \label{enum:Detint}
For deterministic instruments $\dI, \dJ$ on $\cHi$, $\dI$ interferes with $\dJ$ iff  $\dJx \circ \dIs \neq \dJx$. 
\end{enumerate}
\end{proposition}
\begin{proof}
\eqref{enum:inttoorder}: The arguments for the statements with or without $\rho$ are the same. Suppose that $\MABrho$ has no order effect, and consider the equality corresponding to $\PMBA = \PMAB$ up to swaps. Applying $\discard{}$ to $X_A$ yields precisely the failure of an interference effect  \eqref{eq:interference-rho}, using that $A$ is a channel. 

\eqref{enum:commutenoint}: Follows from \eqref{enum:inttoorder}.

\eqref{enum:Bayesnoint}: Follows from \eqref{enum:commutenoint}.

\eqref{enum:Detint}: An interference effect fails to hold iff:
\[
\input{detinterfere.tikz}
\]
\end{proof}

\begin{dataexample}[Categorisation-Decision Data] \label{ex:Cat-Dec-Exp}
Our first example of an interference effect explores the interaction between categorisation and decision-making \cite[1.2.2]{busemeyer2012quantum} \cite{busemeyer2009empirical}. In a study participants were shown pictures of faces and then two forms of decisions could be asked of them. Firstly, a categorisation decision $\catdecC$: were the faces of a `good guy' ($\good$)  or `bad guy' ($\bad$)? Secondly, an action decision $\actdecD$: should they `attack' ($\attack$) or `withdraw' ($\withdraw$)? The experiment includes data $\PB$ for conditions under which the participants merely make the decision $\actdecD$ alone, and data $\PAB$ for when they first categorise $\catdecC$ and then make the decision $\actdecD$.

The experimental data then consists of distributions $\PB$ and $\PAB$ with $X_\catdecC = X_\actdecD=\{y,n\}$. A specific dataset from an experimental trial discussed in \cite[1.2.2]{busemeyer2012quantum} is the following. 
\begin{equation} \label{eq:cat-dec-exp}
\begin{tikzpicture}
	\begin{pgfonlayer}{nodelayer}
		\node [style=cmap] (0) at (-2, 0) {$P(\actdecD)$};
		\node [style=none] (1) at (-2, 1.5) {};
		\node [style=label] (2) at (-2, 2) {$X_\actdecD$};
	\end{pgfonlayer}
	\begin{pgfonlayer}{edgelayer}
		\draw [style=cwire] (1.center) to (0);
	\end{pgfonlayer}
\end{tikzpicture}
 \quad 
\begin{array}{ll}
\actdecD \attack  & \actdecD \withdraw  \\ \cline{1-2}
0.69 & 0.31   
\end{array}
\qquad \qquad \qquad 
\input{PCD.tikz}
\quad 
\begin{array}{lllll}
\catdecC \good \actdecD \attack  & \catdecC \good \actdecD \withdraw  & \catdecC \bad \actdecD \attack  & \catdecC \bad \actdecD n  &  \\ \cline{1-4}
0.07 & 0.1 & 0.52 & 0.31 &  
\end{array}
\end{equation}
Taking marginals, we see that the categorisation $\catdecC$ interferes with the decision $\actdecD$. 
\[
\input{dataoff2.tikz} 
\qquad 
\begin{array}{ll}
\catdecC \actdecD \attack  & \catdecC \actdecD \withdraw   \\ \cline{1-2}
0.59 & 0.41   
\end{array}
\]
\end{dataexample}

While Proposition \ref{prop:interference} rules out Bayesian models of such data, quantum projective models can satisfy interference effects, such as the following.

\begin{example} \label{ex:cat-dec-quantum-simple}
A simple quantum model of the categorisation-decision experiment data \eqref{eq:cat-dec-exp} outlined in \cite[Section 8]{busemeyer2012quantum} is as follows. 
We set $\hilbH = \mathbb{C}^2$. We model $\catdecC, \actdecD$ as orthonormal basis measurements, for the bases $\ket{\goodstate}, \ket{\badstate}$ and $\ket{\attackstate}, \ket{\withdrawstate}$ respectively, and use an initial pure state $\psi = \psibad \ket{\badstate} + \psigood \ket{\goodstate}$. Then the $\actdecD$-alone condition means that we must have: 
\begin{equation} \label{eq:Dalone}
P(\actdecD \attack) = \left|\braket{\attackstate \mid \psi}\right|^2 = \left|\psibad \braket{\attackstate \mid \badstate} + \psigood \braket{\attackstate \mid \goodstate} \right|^2
\end{equation}
For the $\catdecC$ then $\actdecD$ condition we must have $P(\catdecC \bad) =  \left|\psibad \right|^2$ and:
\begin{align*}
P(\actdecD \attack \mid \catdecC  \bad) &= \left|\braket{\attackstate \mid \badstate}\right|^2 \\
P(\actdecD \attack \mid \catdecC  \good) &= \left|\braket{\attackstate \mid \goodstate}\right|^2   
\end{align*}
The other analogous conditions then follow automatically. 

To give such a model, set  $\ket{\attackstate} = \sqrt{P(\actdecD \attack \mid \catdecC \good)} \ket{\goodstate} + \sqrt{P(\actdecD \attack \mid \catdecC \bad} \ket{\badstate}$, $\psibad = \sqrt{P(\catdecC \bad)}$ and $\psigood = e^{i \theta} \sqrt{P(\catdecC \good)}$. 
Then all of the above conditions follow immediately apart from the $\actdecD$-alone condition \eqref{eq:Dalone}, which requires that 
\[
P(\actdecD \attack) = \left|\sqrt{\PAnBy + \sqrt{\PAyBy}e^{i\theta}}\right|^2 
\]
It follows that $\theta$ is given by: 
\[
\cos(\theta) = \frac{P(\actdecD \attack) - \PAnBy - \PAyBy}{2 \sqrt{\PAnBy \PAyBy}}
\]
Hence for this to yield a valid model the RHS above must lie in the range $[-1,1]$. Setting $\theta = 1.31 \pi$ then fits the data of Example \ref{ex:Cat-Dec-Exp}. Note that some experimental data will not fit such a model: in \cite[Section 8]{busemeyer2012quantum} an alternative set (the `wide face' data) is given for which no such $\theta$ is available. 
\end{example}

However, note that by using more general instruments, classical models of interference are available, as follows.

\begin{example}[Classical Interference] \label{ex:class-interference}
A classical Markov model of arbitrary decision data $\{\PAB, \PB\}$ for decisions $A, B$ with outcomes $\{y,n\}$, including the case where the data shows an interference effect of $A$ on $B$, is given as follows. We set $S=\{\star, \sany, \sann\}$ and define the following transition matrices. Bracketed entries $(x)$ indicate that the value of $x$ is arbitrary. Note that the final row in each case is determined by the rest since columns sum to $1$.
\begin{equation} \label{eq:CD-transition-matrices}
\begin{blockarray}{cccccc}
T_A & \star & \sany & \sann \\
\begin{block}{c(ccccc)}
  \star & 0 & 0 & 0 \\ 
  \sany & \PAy & (1) & (0)  \\
  \sann & \PAn & (0) & (1)  \\
\end{block}
\end{blockarray}
\color{black}
\qquad 
\qquad 
\qquad 
\begin{blockarray}{cccccc}
T_B & \star & \sany & \sann \\
\begin{block}{c(ccccc)}
  \star & 0 & 0 & 0 \\ 
  \sany & \PBy & \data(By \mid Ay) & \data(By \mid An)  \\
  \sann & \PBn & \data(Bn \mid Ay) & \data(Bn \mid An)  \\
\end{block}
\end{blockarray}
\end{equation}
For example, for the data from the categorisation-decision experiment we have $P(Ay) = 0.17$, $\PBy = 0.69$, $P(By \mid Ay) = \frac{0.07}{0.17}$ and $P(By \mid An) = \frac{0.52}{0.83}$.
\end{example}

Thanks to Proposition \ref{prop:interference} \eqref{enum:Detint} it is also straightforward to give deterministic classical models exhibiting interference effects.

\subsection{Conjunction fallacy}  \label{subsec:conj-fallacy}

A particular case of interference which deserves its own discussion is the following. Suppose we ask subjects yes-no questions such as whether a statement $A$ holds, statement $B$ holds, or whether both statements $A \wedge B$ hold. In experiments, participants at times judge the conjunction $A \wedge B$ as more likely than $A$ or $B$ individually, which would appear to be a clear `fallacy' in their reasoning. 

To model this, we treat statements $A$ and $B$ as decisions with outcomes $X_A=X_B=\{y,n\}$, and $A \wedge B$ as the participant first deciding that $A$ holds and then that $B$ does. We then characterise the fallacy as follows.

\begin{definition}[Conjunction Fallacy]
We say that data $P$ for yes-no decisions satisfies the 
\emph{conjunction fallacy} when: 
\begin{equation} \label{eq:conj-general}
\input{conj-general.tikz}
\end{equation}
That is, $\PAyBy > \PBy$. 
 Similarly, we say that instruments $A, B$ on $S$ with outcomes $\{y,n\}$ satisfy the conjunction fallacy in a state $\rho$ when we have:
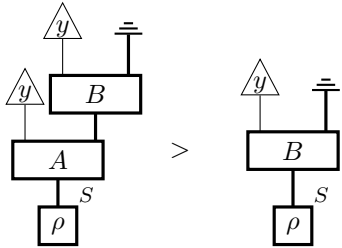
\[
\input{conj-fall-2.tikz}
\]
i.e.~the data $P_\modelM$ for $\MABrho$ satisfies a conjunction fallacy.
\end{definition}
Any data showing a conjunction fallacy requires decision $A$ to interfere with decision $B$. Indeed, otherwise we have:
\[
\input{conj-needs-int.tikz}
\]
By Proposition \ref{prop:interference} \eqref{enum:Bayesnoint} we then immediately have the following. 
 
\begin{corollary}
No Bayesian model can satisfy the conjunction fallacy. 
\end{corollary}

Perhaps the most famous example of a conjunction fallacy is the following. 

\begin{dataexample}[Linda experiment] \label{ex:linda-data}
\cite{tversky1983extensional,busemeyer2012quantum}
In a study, a group of participants are read a small piece of text about a person named Linda, and then asked to assess the likelihood of various statements about her, including that: 1) Linda is a bank teller and 2) Linda is a feminist and a bank teller. Participants frequently rate the second statement as more likely, showing a conjunction fallacy. 

We treat whether Linda is a feminist as a decision $\Fem $, and  whether she is a bank teller as a decision $\decB$, each with outcomes $X_\Fem =X_B=\{y,n\}$. 
\begin{center}
$\begin{array}{l|ll}
& By  & \Fem yBy   \\ \cline{1-3}
P & 0.024 &    0.09315
\end{array}$\end{center}
\end{dataexample}

General classical and quantum instrument models can show interference effects, and thus model data featuring conjunction fallacies.

\begin{example}
A simple quantum model $\modelM$ of the Linda experiment data from Example \ref{ex:linda-data} showing a conjunction fallacy is offered in \cite{busemeyer2012quantum}; we base our account here on the succinct formulation by Jacobs in \cite{jacobs2017quantum}.\footnote{Note that this should really require a distinction in the data between `bank teller and feminist' vs `feminist and bank teller'; Jacobs raises issues around this choice of ordering in \cite{jacobs2017quantum}. Nonetheless the experiment is modelled sequentially as `feminist' then `bank teller' in this manner in, for example, \cite{busemeyer2012quantum}.} For the quantum model we use $\hilbH=\mathbb{C}^2$. We set $\Fem$ to be a measurement for the orthonormal basis $\ket{\Fem y}, \ket{\Fem n}$, and $B$ a measurement for the orthonormal basis $\ket{By}, \ket{Bn}$, and use the pure initial state $\rho = \ket{\psi} \bra{\psi}$, for the following vectors. 
\[
\ket{\psi} = 0.987 \ket{\Fem y} -0.1564 \ket{\Fem n}
\qquad
\ket{By} = \cos\left(2\pi/5\right) \ket{\Fem y} + \sin(2\pi / 5) \ket{\Fem n}
\]
Then indeed this fits the Linda data since:
\begin{align*}
\PMBy
&= 
\left|\braket{By \mid \psi}\right|^2 = \left|0.987 \cos(2\pi/5) -0.1564\sin(2\pi/5)\right|^2 = 0.024
\\ 
\PAyBy
&= \  \left|\braket{\Fem y \mid \psi}\right|^2 \left|\braket{By \mid \Fem y}\right|^2 = 0.987^2 \left|\cos(2\pi/5)\right|^2  = 0.09315
\end{align*}
\end{example}

\begin{example}[Classical Conjunction Fallacy]
Fitting the data $\PBy$ and $\PAyBy$ for a conjunction fallacy such as Example \ref{ex:linda-data} offers little constraint on a model. For example we can give a trivial form of classical simple Markov model on $S=\{\sany,\sann\}$ with $\rho=\delta_n$ via: 
\[
\begin{blockarray}{cccccc}
T_A &  \sany & \sann \\
\begin{block}{c(ccccc)}
  \sany & 1 & \PAyBy  \\
  \sann & 0 & 1-\PAyBy \\
\end{block}
\end{blockarray}
\qquad \qquad 
\begin{blockarray}{cccccc}
T_B &  \sany & \sann \\
\begin{block}{c(ccccc)}
  \sany & 1 & \PBy  \\
   \sann & 0 & \PBn \\
\end{block}
\end{blockarray}
\]
In the case where we are given further data in the form of entire distributions $\PB$ and $\PAB$ we can fit the data with a Markov model just as for any interference effect, as in Example \ref{ex:class-interference}, or more generally still, apply the recipes for classical Markov models from Theorem \ref{Thm:alt-two}. 
\end{example}

\subsection{QQ-Equality} \label{sec:QQ}

Our next cognitive effect is a numerical condition which in \cite{wang2014context} the authors observe to be satisfied by many examples of polling datasets. Interestingly, the property also automatically holds in all quantum projective models.

\begin{definition}[QQ-Equality] \label{Def:QQ-data}
\cite{wang2014context} 
We say that data $P$ for yes-no decisions $A, B$ with outcomes $X_A=X_B=\{y,n\}$
satisfies the \emph{quantum question (QQ-)equality} when: 
\begin{equation} \label{eq:QQ-def-data}
\PAyBy + \PAnBn = \PByAy + \PBnAn
\end{equation}
Equivalently, the following equality holds: 
\[
\input{data-QQish.tikz}
\]
where for any classical object $X$, we define a `cap' process  $\tinycap \colon X \times X \to I$ by:
\begin{equation} \label{eq:cap-def}
\input{capdef2.tikz} \ \ := \ \ 
\begin{cases}
1 & x=z \\ 
0 & \text{otherwise}
\end{cases}
\end{equation}
Similarly, we 
say that instruments $A, B$ on $S$ with outcomes $X_A=X_B=\{y,n\}$ satisfy the \emph{QQ-equality} in a state $\rho$ when: 
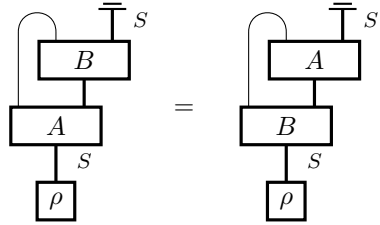
\begin{align*}
\input{QQone.tikz}  = \ \ \ \  \input{QQtwo.tikz} 
\end{align*} 
Equivalently, the data $P_\modelM$ induced by the model $\MABrho$ satisfies the QQ-equality. More strongly, we say that $A, B$ satisfy the \emph{state-independent QQ-equality} when the same equality holds without $\rho$ specified, so that $S$ is now an input (as in \eqref{eq:effects-simple}). 
\end{definition}

\begin{remark}
The QQ-equality was originally defined in the following equivalent form, given by subtracting each side of \eqref{eq:QQ-def-data} from one: 
\[
\PAyBn + \PAnBy = \PByAn + \PBnAy
\]
However here we prefer the graphically simpler form \eqref{eq:QQ-def-data}. 
\end{remark}

\begin{proposition} \label{prop:QQ} \ 
\begin{enumerate}
\item \label{enum:proj-QQ}
Any two-outcome projective instruments $\ppinst, \pqinst$ satisfy the QQ-equality. 
\item \label{enum:commute-QQ}
Any commuting two-outcome instruments satisfy the QQ-equality. 
\item \label{enum:det-QQ}
Deterministic instruments $\dI, \dJ$ on $\cHi$ satisfy the QQ-equality iff for all $s \in \cHi$ we have that $\dIx(s) = \dJx(\dIs(s)) \iff \dJx(s) = \dIx(\dJs(s))$.
\end{enumerate}
\end{proposition}

\begin{proof}
\eqref{enum:proj-QQ}:
For any state $\rho$ we have: 
\[
\input{PQa-proof.tikz}
\]
Here we use cyclicity of the trace in the second step, and in the last step that $\pp^\bot = 1-\pp$ and $\pq^\bot = 1-\pq$. Since this expression is symmetric in $\pp, \pq$ the result follows. 

\eqref{enum:commute-QQ}: Immediate from the definition. 

\eqref{enum:det-QQ}: 
The QQ-equality holds iff:
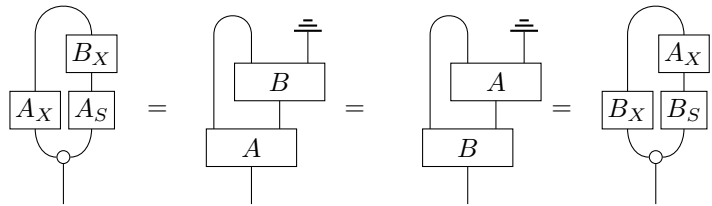
\[
\input{detQQ.tikz}
\]
The processses on each side pick out (i.e.~map those elements to 1 and the rest to 0) the subset of $s \in \cHi$ for which $\dIx(s) = \dJx(\dIs(s))$ and $\dJx(s) = \dIx(\dJs(s))$, respectively. 
\end{proof}

Since it is satisfied by all quantum projections, the QQ-equality was originally proposed in \cite{wang2014context} as an indicator of the `quantum' nature of question-answer data. 
However, despite the title of that work, on its own the QQ-equality should not be mistaken for any indication of quantum-ness. Indeed it is precisely a partial \emph{absence} of an order effect, stating that the quantities \eqref{eq:QQ-def-data} are unaffected by swapping the order of $A, B$. 
Thus it holds for any commuting instruments, and in particular any classical Bayesian instruments. 

Nonetheless, when combined with order effects the QQ-equality provides a non-trivial constraint, forcing a situation in-between commutation and full non-commutation which is suggestive of projective models. For general classical (or quantum) instruments, the QQ-equality may or may not hold.

\begin{example}
The pair of deterministic classical instruments $A, B$ on $S=X_A=X_B=\{y,n\}$ from Example \ref{ex:simple-det-OE} fail to satisfy the QQ-equality, since $\dIx(y) = \dJx(\dIs(y)) = y = \dJx(y)$ but $\dIx(\dJs(y)) = n$. 
\end{example}

\subsection{Response replicability effect} \label{subsec:RRE}

The final effect of this section can also be seen as a requirement that a model behaves in a 
 `more classical' fashion, but unlike the previous effect has posed a \emph{challenge} for quantum projective models. 

 In many question-answer scenarios, such as those with short timescales or more `objective' forms of questions, it is natural to expect that participants' responses to questions are consistent, in that if we ask a question sequence such as $\decA \decA$ or $\decA \decB \decA$ we receive the same answer to both instances of question $\decA$. We capture this as follows.

\begin{definition}[RRE]
We say that decision data $P$
satisfies the \emph{response replicability effect (RRE)} when the following hold:
 \[
\input{AArepsimple.tikz}  \qquad \qquad \qquad \input{ABrepsimple.tikz}
 \]
 or equivalently: 
 \begin{align*}
\PAaAad &= \delta_{a,a'} \PAa
&
\PAaBbAad &= \delta_{a,a'} \PAaBb
 \end{align*}
Furthermore we require the analogous conditions swapping $\decA$ and $\decB$. 

There are several related conditions at the level of instruments. We say that an instrument $A$ on $S$ is \emph{repeatable in state $\rho$} when we have: 
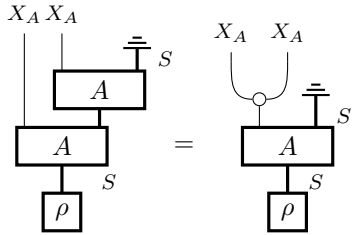
\[
\input{inst-RRE-1-state.tikz}
\]
We simply say that $A$ is \emph{repeatable} when the same equality holds without $\rho$ specified, now with $S$ as an input to the diagrams.  
More strongly still, we say that $A$ is \emph{idempotent} when: 
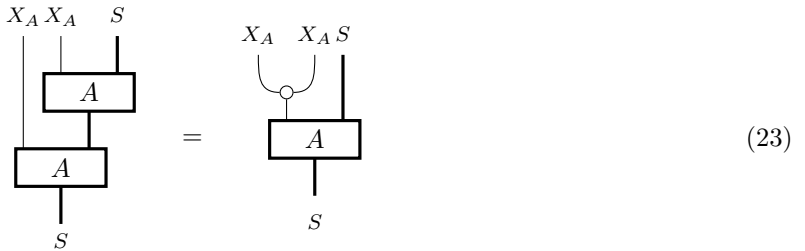
\begin{equation} \label{eq:idempotent}
\input{inst-idem-2.tikz}
\end{equation}
We say that a pair of instruments $A, B$ on $S$ \emph{satisfy RRE in state $\rho$} when they are each repeatable on $\rho$ and we have: 
\begin{equation} \label{eq:RRE-state-dep}
\input{inst-RRE-2-state.tikz}
\end{equation}
along with the same condition swapping $A$ and $B$. Equivalently, the data $P_\modelM$ induced by the model $\MABrho$ satisfies RRE. Finally, we simply say that instruments $A$ and $B$ satisfy \emph{(state-independent) RRE} when they are each repeatable and satisfy diagram \eqref{eq:RRE-state-dep} without $\rho$ specified, now with $S$ as an input to each side. 
\end{definition}

\begin{proposition} \label{prop:RRE} \ \ 
\begin{enumerate}
\item  \label{enum:idem-com-RRE}
Any instruments $A, B$ which are idempotent and commute satisfy RRE.
\item \label{enum:det-bayes-RRE}
Any deterministic Bayesian instruments $A, B$ are idempotent and satisfy RRE. 
\item \label{enum:detRRE} 
Deterministic instruments $\dI, \dJ$ on $\cHi$ satisfy idempotence iff $\dIx \circ \dIs = \dIx$ and $\dIs \circ \dIs = \dIs$ (and similarly for $\dJ$) and the $ABA$ and $BAB$ conditions of state-independent RRE iff 
$\dIx \circ \dJs \circ \dIs = \dIx \circ \dIs$ and $\dJx \circ \dIs \circ \dJs = \dJx \circ \dJs$. 
\item \label{enum:proj-q-idem}
Any projective quantum instrument $\hat P$ is idempotent. 
\item \label{enum:proj-q-rre-commute} \cite{ozawa2023nondistributivity}
Projective quantum instruments $\hat P,  \hat Q$ satisfy RRE iff they commute. 
\item 
\label{enum:cor-meas-POVM}
The instrument $\hat \measm$ for a POVM $\measm$ is repeatable iff $\measm$ is a projective measurement. 

\end{enumerate}
\end{proposition}
\begin{proof}
\eqref{enum:idem-com-RRE}: 
Straightforward by graphical reasoning. 

\eqref{enum:det-bayes-RRE}:
Idempotence is a special case of \eqref{enum:detRRE} below since any Bayesian instrument $A$  has $\dIs = \id{}$. 
Since all Bayesian instruments commute by Proposition \ref{prop:OE} \eqref{enum:Bayes-commute} they then satisfy RRE by \eqref{enum:idem-com-RRE}.

\eqref{enum:detRRE}:
Idempotence is equivalent to the central equation and hence that the LHS and RHS below are equal. 
\[
\input{detidem.tikz}
\]
In the first step and last step we used that $\dIs$ and $\dIx$ are deterministic. 
Since both sides are deterministic, equality holds iff the marginals (i.e. morphisms after each copy output) on each side are equal, which yields the conditions.  
The RRE condition may be verified by similar diagrammatic reasoning. 

\eqref{enum:proj-q-idem}: 
Applying an input state $\rho$, for given outcomes $x \in X, y \in Y$ we have that $\hat P_y \circ \hat P_x \circ \rho = P_y P_x \rho P_x P_y = \delta_{x,y} P_x \rho P_x$, as required. 

\eqref{enum:proj-q-rre-commute}: 
Suppose that RRE holds. Then in particular for every input state $\rho$ and for all $x, y$ we have:
\[
\discard{} \circ \hat P_x \circ \hat Q_y \circ \hat P_x \circ \rho =  \discard{} \circ  \hat Q_y \circ \hat P_x \circ \rho 
\]
Setting $\rho = \ket{\psi} \bra{\psi}$ to be a pure state and writing 
$\pp := \pp_x$, $\pq := \pq_y$, this is equivalent to: 
\[
\| \pp \pq \pp \psi \|^2 
=
\Tr(\pp \pq \pp \ket{\psi} \bra{\psi} \pp \pq \pp)
=
\Tr(\pq \pp \ket{\psi} \bra{\psi} \pp \pq) = \| \pq \pp \psi \|^2
\]
But since: 
\[
\| \pq \pp \psi \|^2 = \| \pp \pq \pp \psi \|^2  + \| \pp^\bot \pq \pp \psi \|^2 
\]
we obtain $\pp^\bot \pq \pp \psi = 0$. Since $\psi$ was arbitrary $\pp^\bot \pq \pp = 0$ and hence $\pq \pp = \pp \pq \pp$. But then $QP = PQP = (PQP)^\dagger = (QP)^\dagger = PQ$ using that $P=P^\dagger$ and $Q=Q^\dagger$. Hence $P_x$ and $Q_y$ commute for all $x,y$. Then by Proposition \ref{prop:OE} \eqref{enum:noOE-projec} $\hat P$ and $\hat Q$ commute as instruments. Conversely, if $\hat P$ and $\hat Q$ commute as instruments, they satisfy RRE by \eqref{enum:proj-q-idem} and \eqref{enum:idem-com-RRE}.
\eqref{enum:cor-meas-POVM}: Appendix \ref{app:proofs}.

\end{proof}

While natural for question-answer type data, RRE raises 
difficulties for quantum models which also wish to exhibit order effects. 

Recall that any quantum measurement given by a POVM or projection canonically defines an instrument as in Example \ref{ex:POVM-insts} (but not every quantum instrument is of this form). The following result was essentially established in \cite[Prop.~4]{khrennikov2014quantum} as a challenge to this sub-class of quantum instrument models.

\begin{theorem} \label{thm:RRE-no-order}
Quantum POVM instruments 
satisfy RRE in a state $\rho$ iff 
they are idempotent on $\rho$ and commute on $\rho$. 
\end{theorem}

\begin{proof}
Appendix.
\end{proof}

It follows from Theorem \ref{thm:RRE-no-order} that quantum instruments based on projective measurements or POVMs cannot model RRE at the same time as order effects. As considered in \cite{khrennikov2014quantum} this suggests the use of more general (classical or quantum) instruments.

\subsection{Deterministic model of the effects}

In \cite{ozawa2021modeling} a fully deterministic classical model providing a close fit to the Clinton-Gore poll data is presented, which in particular satisfies all of the cognitive effects we have seen. 

\begin{example}\cite{ozawa2021modeling} \label{ex:OK-det-model-CG}
There is a deterministic classical model of the Clinton-Gore data with 0.75\% error, whose instruments $A, B$ are idempotent and satisfy interference (and order) effects, RRE and the QQ-equality. 
\end{example}
\begin{proof}
Set $\cHi = \okA \times \okB \times \okP$ where $\okA=\okB=\ansone = \anstwo = \{y,n\}$ and $\okP=\{0,1,2\}$, and define the deterministic instruments $\dI$ and $\dJ$ by:
\begin{align}
\dIx(a,b,p) &:= a &   \dJx(a,b,p) &:= b \\ 
\dIs(a,b,p) &= (a,f_B(a,b,p),p)) & \dJs(a,b,p) &= (f_A(a,b,p),b,p) 
\end{align}
where: 
\[
f_B(\okAel,\okBel,\okPel) := \begin{cases}
\okBel
 & \okPel = 0 \\ 
\okAel^\bot & \okPel =1 \\ 
\okAel & \okPel=2
\end{cases}
\qquad \qquad \qquad 
f_A(\okAel,\okBel,\okPel) := \begin{cases}
\okAel
 & \okPel = 0 \\ 
\okBel^\bot & \okPel =1 \\ 
\okBel & \okPel=2
\end{cases}
\]
In \cite{ozawa2021modeling} it is shown that for a particular probability distribution $\rho$ over $S$ as the initial state this model fits the Clinton-Gore data with the stated accuracy. 
We now use Prop.~\ref{prop:OE} \eqref{enum:detorder}, Prop.~\ref{prop:QQ} \eqref{enum:det-QQ} and Prop.~\ref{prop:RRE} \eqref{enum:detRRE} to check these instruments have the desired properties. For idempotence, we have $\dIx\dIs(a,b,p) = a \dIx(a,b,p)$. Inspecting each case $p=0,1,2$ one may also check that $\dIs \dIs(a,b,p) = \dIs(a,b,p)$. For interference, note that $\dJx \dIs(a,a,1) = f_B(a,a,1) = a^\bot \neq a = \dJx(a,a,1)$.

For RRE, observe that $\dIx\dJs\dIs(a,b,p) = \dIx\dJs(a,f_B(a,b,p),p) = f_A(a,f_B(a,b,p),p)$. For each case $p=0,1,2$ one may see that this is equal to $a = \dIx\dIs(a,b,p)$, as required. The condition for $BAB$ is verified similarly. For the QQ-equality note that $\dIx(a,b,p) = \dJx\dIs(a,b,p)$ iff $a=f_B(a,b,p)$. For $p=0$ this holds iff $a=b$. For $p=1$ this never holds. For $p=2$ this always holds. Since these conditions are symmetric in $A, B$ the result follows. 
\end{proof}

\begin{remark}
The above model was originally represented in terms of decoherent quantum instruments, as noted in \cite[Eq. (84),(5)]{ozawa2021modeling}. However, this may equivalently be viewed as a deterministic classical model as above. 
\end{remark}

\section{Classical Models of Arbitrary Decision Data}

\label{sec:classical-models}

When making use of general classical instrument models, rather than merely Bayesian models, essentially any set of sequential decision data can in fact be modelled classically, as we will now prove.
The result depends on the data satisfying the following very mild property, which is in fact necessary in order to have any form of instrument model at all. 

\begin{definition}[No-Signalling from the Future]
We say that decision data satisfies \emph{no-signalling from the future} when for every sequence of decisions $(\q_1, \dots, \q_k)$, every extension to a sequence $(\q_1,\dots,\q_k,\q_{k+1},\dots\q_n)$ in the data has the same marginal on its first $k$ outcomes: 
\begin{equation} \label{eq:same-marginal}
\input{marginalBS.tikz}
\end{equation}
That is:
\begin{equation} \label{eq:subseq-dist}
\PAxtok := 
\sum_{\an_l \ k < l \leq n} \PAxton
\end{equation} 
is a well-defined distribution, independent of $\q_{k+1},\dots,\q_n$.
\end{definition}

We can think of the property as characteristic of temporally sequential data. Morally any experimental dataset should satisfy this condition if the experiment has been carried out correctly, since the choice of what decision the participant is asked to make in the future shouldn't affect the probabilities of answers to earlier decisions.\footnote{This assumes that participants are not told in advance what sequence of decisions they will be making, which would arguably make modelling via sequential instruments no longer natural, and this is made formal in Lemma \ref{lem:inst-means-no-BS}.} 
Each distribution \eqref{eq:subseq-dist} specifies the probabilities for \emph{aborting} an experiment after these first $k$ decisions.

\begin{lemma} \label{lem:inst-means-no-BS}
Any data with an instrument model satisfies no-signalling from the future.
\end{lemma}
\begin{proof}
Let $\modelM$ be an instrument model with $P_\modelM = P$ for data $P$. Since every instrument is a channel, we can `pull through' discards applied to both the output and outcomes of an instrument, and so \eqref{eq:same-marginal} is equal to $P_\modelM(\q_1,\dots,\q_k)$. 
\end{proof}

Moreover in many cases of interest the condition is in fact vacuously true. 

\begin{examples}
No-signalling from the future holds trivially for any data consisting only of decision sequences $(A, B)$ and $(B, A)$, such as the Clinton-Gore data (Example \ref{ex:Clinton-Gore}). Indeed, here each subsequence $(A)$ or $(B)$ has only a unique extension to one in the data, namely $(A, B)$ or $(B, A)$, respectively. 
For data also containing the sequences $(A, A)$ and $(B,B)$ the condition is equivalent to the distributions $\PAB$ and $\PAA$ having the same first marginal on $A$, and $\PBA$ and $\PBB$ the same first marginal on $B$. 
\end{examples}

We will now give our general results establishing the existence of classical models. Recall the notion of a simple Markov model based on Markov instruments from Example \ref{ex:Markov-instruments}.

\begin{theorem} \label{thm:-simple-M-Model}
Any decision data for sequences of fixed length $n$, satisfying no-signalling from the future, has a simple Markov model with state space: 
\[
\Hi := \coprod^{n-1}_{k=0} \events_k
\]
where $\events_k$ is the set of sequences of the form 
$\q_1 \an_1 \dots \q_k \an_k$
where $\an_j \in \ans_{\q_j}$ for $j=1,\dots,k$, and $\q_1 \dots \q_k$ is a subsequence of a full-length question sequence $\q_1 \dots \q_n$ in the data.
\end{theorem}
\begin{proof}
Appendix \ref{app:proofs}. 
\end{proof}
Note that by the above result, simple Markov models are also available for data containing sequences of varying length (e.g. $(A)$ and $(A,B)$) by arbitrarily extending the data for each shorter sequence to one of maximal length.

For datasets of a particular form, we can also simplify the state space used in the above result. We say decision data has \emph{alternating decisions} when it contains two decisions $\qone, \qtwo$ and just two sequences, both of length $n$, of the form $(\qone, \qtwo, \qone, \dots)$ and $(\qtwo, \qone, \qtwo, \dots)$. 

\begin{theorem} \label{Thm:alt-two}
Any question-answer data for alternating decisions $\qone, \qtwo$ with outcomes $\A$, with sequence length $n$, has a 
simple Markov model with state space:
\[
\Hi := \coprod^{n-1}_{k=0} \A^k
\]
\end{theorem}
\begin{proof}
Appendix \ref{app:proofs}.
\end{proof}

Of most relevance to us is a simple form of the result for data consisting of two yes-no decisions.

\begin{lemma} \label{lem:C-Gdata}
Any decision data $\{\PAB, \PBA\}$ for decisions $A, B$ with outcomes $X_A=X_B=\{y,n\}$ has a simple Markov model with state space:
\begin{equation} \label{eq:three-elements}
\Hi=\{\star,\sany,\sann\}
\end{equation}
and transition channels as follows, where the columns are input and rows output:
\[
\begin{blockarray}{cccccc}
T_A & \star & \sany & \sann \\
\begin{block}{c(ccccc)}
  \sany & \data(Ay) & \data(Ay \mid By) & \data(Ay \mid Bn)  \\
  \sann & \data(An) & \data(An \mid By) & \data(An \mid Bn)  \\
\end{block}
\end{blockarray}
\]
where: 
\[
\data(Ax' \mid Bx) := \frac{\PBxAxd}{\PBx} \qquad \qquad \PBx := \sum_{x'} \PBxAxd
\] 
for $x, x' \in \{y,n\}$, and $T_B$ is defined analogously.\footnote{As a simple Markov model we have initial state $\star$ (i.e. $\rho = \delta_\star$), and $f_A = f_B = f$ with $f(\sany) = y, f(\sann) = n$. We define $\data(Ax \mid Bx')$ arbitrarily if $\data(Bx') = 0$. }
\end{lemma}
\begin{proof}
Special case of Theorem \ref{Thm:alt-two}, proven in Appendix \ref{app:proofs}. 
\end{proof}

As an example, we can give such a classical instrument model of the Clinton-Gore polling data. This model will therefore satisfy all of the cognitive effects present in this data, i.e.~interference and order effects.

\begin{example}[Markov model for Clinton-Gore]
The Clinton-Gore opinion poll data from Example \ref{ex:Clinton-Gore} has a simple Markov model with three states $\{\star,\sany,\sann\}$, initial state $\star$, and transitions channels $T_A$ and $T_B$ given by: 
\[
\begin{blockarray}{cccccc}
T_A & \star & \sany & \sann \\
\begin{block}{c(ccccc)}
  \sany & 0.5346 & 0.7386 & 0.1070  \\
  \sann & 0.4654 & 0.2614 & 0.8930  \\
\end{block}
\end{blockarray}
\qquad \qquad 
\begin{blockarray}{cccccc}
T_B & \star & \sany & \sann \\
\begin{block}{c(ccccc)}
  \sany & 0.7616 & 0.9163 & 0.3797  \\
  \sann & 0.2384 & 0.08361 & 0.6203  \\
\end{block}
\end{blockarray}
\]
where the columns show the previous state and the rows the next state. 
 When the state transitions to $\sany$ or $\sann$ the outcome of the corresponding instrument is $y$ or $n$, respectively. 
\end{example}

The above results can also be strengthened such that the instruments satisfy RRE (see Section \ref{subsec:RRE}), as we show in Appendix \ref{sec:general-RRE-Markov-models}.

\section{Measurement Models} \label{sec:meas-models}

Given that any sequential decision data can be given a classical instrument model, why would one feel a need to use quantum (or other non-classical) models? In fact, the quantum-like modelling literature thus far has largely not allowed the use of such arbitrary instruments (classical or otherwise), and restricted only to those of the following form.

\begin{definition}[Measurements]
Let $\cob$ be a finite classical object. A \emph{measurement} $\measm$ on $\obI$ with outcomes $\cob$ is a channel $\measm \colon \obI \to \cob$. 
We say an instrument $I$ \emph{measures} $\measm$ when we have: 
\begin{equation} \label{eq:inst-measm}
\input{instmeasm.tikz}
\end{equation}
\end{definition}

While every instrument measures some measurement $\measm$, in general there are many instruments measuring the same $\measm$. However, in certain process theories there is a canonical choice, allowing us to define decision models more simply in terms of measurements.  

\begin{definition}[Measurement Models]
We say that a process theory $\catC$ has an \emph{update structure} when for each measurement $\measm \colon \obI \to \cob$ there is a chosen instrument $\upd{\measm}$ on $\obI$ which measures $\measm$, which we call the \emph{update} instrument for $\measm$. 
\[
\input{update-mapp.tikz}
\]
A \emph{measurement model} $\modelM$ of decisions $\decset$ in $\catC$ is an instrument model of $\decset$ in $\catC$ where every instrument is given by an updating instrument $A=\upd{m_A}$ for some measurement $\measm_A$. 
\end{definition}

Rather than a collection of instruments $\{A\}_{A \in \decset}$ a measurement model is thus equivalently specified merely by a collection of measurements $\{\measm_A\}_{A \in \decset}$ via $A=\upd{\measm_A}$. Typically measurements will require fewer parameters, making such a model simpler. Update structures are essentially introduced in \cite{jacobs2015new} (there referred to as `instruments') along with some natural structural axioms which we omit here for brevity. 

The process theories $\Class$ and $\QC$ both have update structures, and in fact we have already met the corresponding forms of model: classical measurement models are precisely Bayesian models, while quantum measurement models are those based on POVMs, including projections as a special case.  

\begin{example}[Classical measurement models] \label{ex:c-updates}
In $\Class$, a measurement $\measm \colon S \to X$ on a finite set $S$ with outcomes $X$ is a probability channel $\measm \colon S \to X$. Then $\upd{\measm}$ is precisely the Bayesian instrument \eqref{eq:bayes-inst} given by $\measm$. 
\end{example}

\begin{example}[Quantum measurement models] \label{ex:q-updates}
In $\QC$, a measurement $\povm$ on a quantum system $\hilbH$ is precisely a 
POVM $\povm=(\povm_x)_{x \in \cob}$ on $\hilbH$. The update instrument is the associated instrument $\upd{\povm} = \hat \povm$ as in Example \ref{ex:POVM-insts}. In particular, any projective measurement $P=(P_x)_{x \in X}$ on $\hilbH$ is a measurement $P \colon \hilbH \to X$, and its update instrument $\upd{P}=\hat P$ is the corresponding projective instrument as in Example \ref{ex:proj-insts}. 
\end{example}
For completeness, we note that we can unify classical and quantum measurement models as follows.

\begin{example}[$\QC$ measurement models]
In $\QC$, a measurement $\measm \colon (\hilbH, S) \to X$ on a general object $(\hilbH, S)$ is given by specifying a POVM $\measm_s := (\measm_{x \mid s})_{x \in X}$ on $\hilbH$ for each $s \in S$. The instrument $\upd{\measm}$ is defined by: 
\[\upd{\measm}(x, s' \mid s) := \upd{\measm_s}_x \delta_{s,s'}\]
for $x \in X$ and $s, s' \in S$. On a classical system $S$ (when $\hilbH$ is trivial) this yields precisely classical measurements and Bayesian instruments (Example \ref{ex:c-updates}). On a quantum system $\hilbH$ (when $S$ is trivial) this yields POVMs and quantum updates (Example \ref{ex:q-updates}). 
\end{example}

\paragraph{Advantages of measurement models.} 
When modelling cognitive decisions, should we then make use of arbitrary instruments or restrict ourselves to measurement models only? 
Most of the quantum-like modelling literature so far has implicitly focused on measurement models, usually comparing classical Bayesian models against quantum projective models, while works such as \cite{khrennikov2014quantum,jacobs2017quantum} have pointed out that one can extend quantum models more generally to POVMs. If we do make this restriction, then indeed quantum (projective) models can account for many of the cognitive effects from Section \ref{sec:cog-effects} (and most of our examples take this form) while classical (Bayesian) models cannot, providing the usual arguments for quantum-like approaches. 

Specifying arbitrary instruments can feel overly unconstrained, while restricting to measurements can be seen to provide a more structured way to model decisions. A major empirical benefit is that by encoding each decision more simply as a measurement rather than a full instrument, such models will also typically require fewer parameters, making them simpler to construct and potentially more verifiable. 

Measurements can also be interpreted as providing a kind of \emph{logic} or \emph{probability theory} of decisions. Indeed, much of the literature implicitly equates Bayesian models with \emph{classical logic} or \emph{classical probability theory} (CPT), and quantum projective models with \emph{quantum logic} or \emph{quantum probability} (QPT), comparing the two, see e.g.~\cite{pothos2022quantum}. Formally, the spaces of POVMs and projections on a Hilbert space form an effect algebra and an orthomodular lattice, respectively, under inclusion or `implication', compared to the Boolean algebra of classical (deterministic) measurements. 

\begin{remark}[Sequential measurement Logic]
While we have seen that instruments allow for an easy representation of sequential decisions, measurements can also define such a sequential logic, as follows. Given a measurement model $\upd{a}, \upd{b}$ of decisions $A, B$ in terms of measurements $a = \measm_A, b = \measm_B$, we can model decision $A$ followed by decision $B$ as the single joint measurement:
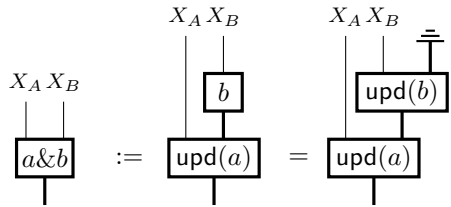
\[
\input{andthen2.tikz}
\]
This ability to combine a measurement $a$ followed by $b$ with the `and-then' operator $a \& b$ can be seen as the basic operation in a logic of sequential measurements. In \cite{jacobs2017quantum} Jacobs gives an account of various aspects of quantum-like cognition, extending the quantum perspective from projective to POVM models, in terms of such an operator (there defined on effects, i.e.~individual outcomes, rather than entire measurements). 
\end{remark}

\paragraph{Challenges for measurement models.} On the other hand, measurement models come with certain limitations that may force us to consider a more general class of instrument models. 

On the formal side, one issue is that measurement processes are not closed under composition. That is, for non-classical\footnote{Bayesian instruments are indeed closed under sequential composition, by associativity of copy morphisms. Note also that $\&$ is not associative in that in general $(a \& b) \& c \neq a \& (b \& c)$ for measurements $a, b, c$. This stems from the fact that $\upd{a \& b}$ is not given by the composite of instruments $\upd{a}$ and $\upd{b}$, for general POVMs $a, b$. For classical measurements $a, b, c$ however, the operation and resulting classical logic is indeed associative $(a \& b) \& c = a \& (b \& c)$, and commutative $a \& b = b \& a$ as we saw in Proposition \ref{prop:OE} \eqref{enum:Bayes-commute}.} measurements $a, b$, the composite of instruments $\upd{a}$ and $\upd{b}$ is not equal to $\upd{c}$ for any measurement $c$. Hence measurements and updates alone arguably do not yield a `closed' theory of sequential decision making, and to model the processes behind general sequences of decisions we must discuss more general instruments.\footnote{
    An alternative approach might be to focus only on the ultimate measurement given by a sequence of decisions, defined via $\&$.  However even in \cite{jacobs2017quantum} to model further aspects of decision making such as `changes of perspective', one must consider more general channels and hence instruments.}

A more significant challenge however, is that quantum measurement (projective or more general POVM) models cannot model all of the cognitive effects we met in Section \ref{sec:cog-effects}. 

\begin{corollary} \label{cor:upd-no-rep-OE}
Quantum measurement models cannot simultaneously satisfy order effects and RRE. 
\end{corollary}
\begin{proof}
This is precisely Theorem \ref{thm:RRE-no-order}. 
\end{proof}

This point is stressed for quantum projective instruments in \cite{khrennikov2014quantum}, as well as more general updating based on POVMs  \cite[Prop 4]{khrennikov2014quantum}, and used to argue for the use of more general instruments. In future it would be interesting to explore whether one can account for all of the cognitive effects either via alternative notions of measurement, providing a more general class than measurement models but more constrained than arbitrary instruments, or via measurements defined in an alternative process theory; we return to this in the discussion in Section \ref{sec:Discussion}.

\section{Joint Decision Models}

\label{sec:Bell-arguments}

If we do allow the use of arbitrary instruments, we have seen that any sequential decision data can be given a classical (not necessarily Bayesian) model. In order to strictly rule out classical models, one must therefore go beyond sequential composites of decisions. A natural place to look is instead in \emph{parallel} composites of decisions into `joint decisions'. Here we typically imagine a participant making a single decision with multiple aspects, represented by several internal mental representation spaces, rather than a combined decision made by distinct agents.\footnote{Note that using such factored spaces itself is not a radical idea: for example the generative models used in Bayesian approaches such as \cite{parr2022active} make use of multiple variables corresponding to distinct spaces.}

Within quantum foundations, the famous result of \emph{Bell's theorem} can in fact strictly rule out classical models based on parallel measurements, potentially allowing for strict arguments against classical models. In this section we will sketch what one would need in order to leverage this result, as a future research direction for quantum-like cognition.

\subsection{Modelling joint decisions}

Let us now introduce the kinds of joint decision data we will consider. 

\begin{definition}[Joint Decision Data]
A set of \emph{joint decision data} $\{P(A \otimes B)\}_{\decinA \in \decsetA, \decinB \in \decsetB}$ is given by:
\begin{itemize}
\item  
    a set of decisions $\decsetA$, each with the same set of outcomes $\Xone$; 
\item 
    a set of decisions $\decsetB$, each with the same set of outcomes $\Xtwo$; 
\item  
    data for a set of decisions denoted $\decset := \{\decinA \jdec \decinB\}_{\decinA \in \decsetA, \decinB \in \decsetB}$ each with outcomes $\Xone \times \Xtwo$. 
\end{itemize}
 We interpret each decision $\decinA \jdec \decinB$ as a joint decision of $\decinA$ and $\decinB$ made simultaneously, with data given by a distribution:
\begin{equation} \label{eq:joint-decision}
\input{PABjoint.tikz}
\end{equation}
over joint outcomes $\Xone \times \Xtwo$. 
We say that the overall data satisfies \emph{no-signalling} when we have: 
\begin{equation} \label{eq:no-signalling}
\input{nosig2.tikz} 
\qquad \qquad \input{nosig1.tikz}
\end{equation}
for all $\decinA, \decinA' \in \decsetA$, $\decinB, \decinB' \in \decsetB$.
\end{definition}

To model such joint decisions, we wish to combine instruments for $A, B$ with respective outcomes $\Xone$, $\Xtwo$ into a single `joint decision' instrument $A \jdec B$ with outcomes $\Xone \times \Xtwo$. Such an operation should satisfy the following, with the second condition ensuring no-signalling \eqref{eq:no-signalling}: 
\begin{equation} \label{eq:joint-dec}
\input{AdecB.tikz}
\qquad \qquad 
\input{nosignal.tikz}
\end{equation}
The simplest solution is to assume each parallel pair $A$ and $B$ commute and then define $A \jdec B$ as their sequential composite. Indeed one may verify that \eqref{eq:joint-dec} will then hold. However this will be too mild to rule out classical models, by the following result. 

\begin{theorem} \label{thm:joint-dec-thm}
Any joint decision data $\{P(A \jdec B)\}_{\decinA \in \decsetA, \decinB \in \decsetB}$ satisfying no-signalling has a classical Markov model where each $A \jdec B$ is the sequential composite of instruments $A, B$, and where $A$ and $B$ commute for all $\decinA \in \decsetA, \decinB \in \decsetB$.
\end{theorem}

\begin{proof}
Appendix \ref{app:proofs}.
\end{proof}

Instead, we can put further constraints on a model by explicitly using the monoidal product in our process theory $\catC$ to capture joint decisions. Thus we restrict attention to models of the following form.

\begin{definition}[Parallel Decision Model] \label{Def:par-meas-model}
A \emph{parallel decision model} of joint decision data $\{P(A \jdec B)\}_{\decinA \in \decsetA, \decinB \in \decsetB}$ is given by:
\begin{itemize}
\item objects $\Hione, \Hitwo$;
\item for each decision $\decinA \in \decsetA$ a measurement $A$ on $\Hione$ with outcomes $\Xone$; 
\item for each decision $\decinB \in \decsetB$ a measurement $B$ on $\Hitwo$ with outcomes $\Xtwo$; 
\item a normalised \emph{initial state} $\Mo$ of $\Hione \otimes \Hitwo$; 
\end{itemize}
such that for all $\decinA \in \decsetA, \decinB \in \decsetB$ we have:
\begin{equation} \label{eq:par-meas-model}
\input{Bell-data-meas.tikz}
\end{equation}
\end{definition}

It follows that any data with a parallel model satisfies no-signalling \eqref{eq:no-signalling}, and so parallel models are only suitable for such data. 

\begin{examples}
By a \emph{classical parallel decision model} we mean a parallel decision model in $\Class$.  Thus it consists of an initial joint distribution $\rho$ over a finite set $S=S_1 \times S_2$ and probability channels $A \colon S_1 \to \Xone$, $B \colon S_2 \to \Xtwo$ for each $A \in \decsetA$, $B \in \decsetB$, recovering each distribution as in \eqref{eq:par-meas-model}. 

A \emph{quantum parallel decision model} is one in $\QC$ on quantum objects. Thus it consists of a pair of Hilbert space $\hilbH_1, \hilbH_2$, a density matrix $\rho$ on the tensor $\hilbH_1 \otimes \hilbH_2$, and POVMs $A=(A_x)_{x \in \Xone}$ on $\hilbH_1$ and $B=(B_y)_{y \in \Xtwo}$ on $\hilbH_2$ for each $A \in \decsetA$, $B \in \decsetB$, satisfying \eqref{eq:par-meas-model}. 
\end{examples}

\begin{remark}
There is another convenient way to specify a parallel joint decision model. Writing $\decsetA = \{A_i\}_{i \in \mathcal{I}}$, $\decsetB = \{B_j\}_{j \in \mathcal{J}}$, we can view any joint decision data as a single  
 probability channel $P$ from the choices of decision on each side $\mathcal{I} \times \mathcal{J}$ to $\Xone \times \Xtwo$, as on the LHS below. Moreover we can combine the measurements for $\decsetA$ into a single controlled measurement here denoted $A$ with classical input $\mathcal{I}$, and similarly for $B_1, B_2$, as below. Then a parallel decision model amounts to specifying the state $\Mo$ and channels $A, B$ as below, such that the following single equation holds: 
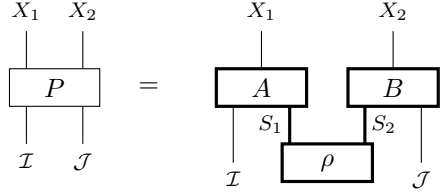
\[ 
\input{factor-data3.tikz}
\]
\end{remark}

\begin{remark}[Parallel Instrument Models]
Note that we only needed \emph{measurements} to represent each decision in a parallel decision model. More strongly however, we can define a \emph{parallel instrument model} just like Definition~\ref{Def:par-meas-model} but with an instrument $A$ on $\Hione$ with outcomes $\Xone$ for each $\decinA \in \decsetA$, and an instrument $B$ on $\Hitwo$ with outcomes $\Xtwo$ for each such $\decinB \in \decsetB$. Then we require each condition \eqref{eq:par-meas-model} to hold for the measurements induced by these instruments: 
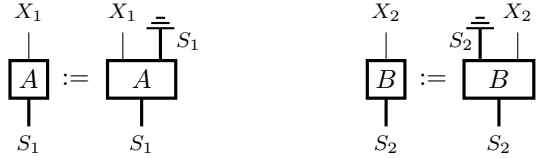
\[
\input{measpar.tikz}
\]
A parallel instrument model can be equivalently defined as a sequential decision model for which $S = \Hione \otimes \Hitwo$ and:
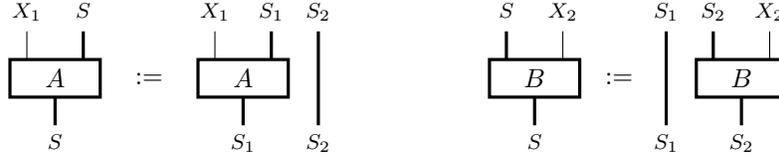
\[
\input{factor-inst.tikz}
\]
for all $\decinA \in \decsetA$, $\decinB \in \decsetB$. By construction each such parallel pair $A$ and $B$ commute on $S$, and their sequential composite $A \jdec B$ (with outcomes given by the product of their outcomes) automatically satisfies our requirements \eqref{eq:joint-dec}. Each joint decision is then modelled as a tensor product: 
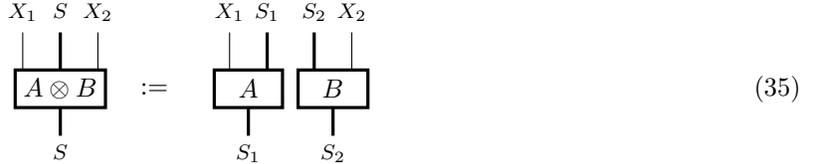
\begin{equation} \label{eq:parallel-joint}
\input{paralleljoint.tikz}
\end{equation}
We will see however that Theorem \ref{thm:joint-dec-thm} cannot be extended to Markov models of this specific form. 
\end{remark}

\subsection{Bell-CHSH inequality}

Suppose we have joint decision data $\{P(A \jdec B)\}_{\decinA \in \decsetA, \decinB \in \decsetB}$ with $\Xone=\Xtwo=\{y,n\}$. For each $\decinA \in \decsetA, \decinB \in \decsetB$ we introduce a quantity similar to that of Section \ref{sec:QQ}:
\[
\QQ(A,B) := P(Ay\jdec By) + P(An \jdec Bn)
\]

Henceforth, we will restrict our attention to the simplest form of joint decision data, with two decisions on each side, so that
$\decsetA=\{\decA,\decA'\}$ and 
$\decsetB=\{\decB,\decB'\}$.
Thus such data consists of four distributions:
\begin{equation} \label{eq:joint-data-simple}
 P(A \otimes B), \quad P(A' \otimes B), \quad P(A \otimes B'), \quad P(A' \otimes B')
 \end{equation}
 over $\Xone \times \Xtwo$. Such data is all that is required for the famous Bell-CHSH inequality, which we can state (in a slightly altered form) as follows.\footnote{
    More standardly, one may define $E(AB) := P(Ay \jdec By) + P(An \jdec Bn) - P(Ay \jdec Bn) - P(An \jdec By)$ and then Bell-CHSH states that $\left|E(AB) + E(A'B) + E(A'B') - E(AB')  \right| \leq 2$, which is easily seen to be equivalent since $E(AB) = 2\QQ(A,B) - 1$. }

\begin{proposition}[Bell-CHSH Inequality]\cite{clauser1969proposed}
Any joint decision data coming with a classical parallel decision model satisfies the inequality: 
\begin{equation} \label{eq:Bell-CHSH-QQ}
\left| \QQ(A,B) + \QQ(A',B) + \QQ(A',B') - \QQ(A,B') - 1 \right| \leq 1
\end{equation}
\end{proposition}

Since the Bell-CHSH inequality is not automatically satisfied by all joint decision datasets, there can exist such datasets for which no classical parallel decision model is available. In particular, there are quantum parallel decision models whose resulting data is able to violate the inequality. 

\begin{example}\cite{brunner2014bell}
A standard example of a quantum model violating the CHSH inequality is with $\Mo = \ket{\psi}\bra{\psi}$ where $\psi=\frac{1}{\sqrt{2}}(\ket{01} - \bra{10})$ is a pure singlet state, and the measurements $A,A',B',B'$ are pure measurements of the observables $\sigma_Z$, $\sigma_X$, $\frac{\sigma_Z-\sigma_X}{\sqrt{2}}$ and $\frac{\sigma_X-\sigma_Z}{\sqrt{2}}$, respectively. For this setup the LHS of \eqref{eq:Bell-CHSH-QQ} is equal to $2\sqrt{2}$, violating the inequality. In fact this violation (divided by $\frac{1}{2}$ in the more standard statement of the inequality), the \emph{Tsirelson bound}, is the largest possible for any parallel quantum model. 
\end{example}
 Any quantum parallel decision model violating the Bell-CHSH inequality requires the use of a state $\Mo$ which is \emph{entangled}. Recall that a state $\Mo$ of $S_1 \otimes S_2$ is a \emph{product state} when it factors as $\Mo = \Mo_1 \otimes \Mo_2$ for states $\Mo_1, \Mo_2$ of $S_1, S_2$ respectively. 
More generally, when $\catC$ admits convex mixtures of (normalised) states (as do $\Class$ and $\QC$) we say that $\Mo$ is \emph{separable} when it is a convex mixture of product states. Otherwise we call $\Mo$ entangled. Every joint state in $\Class$ is separable, and so in $\QC$ if $S_1$ and $S_2$ are classical objects one cannot violate the inequality.

\begin{remark}
Given joint decision data, stating a classical parallel decision model is equivalent to giving a sequential Bayesian model, as follows. 
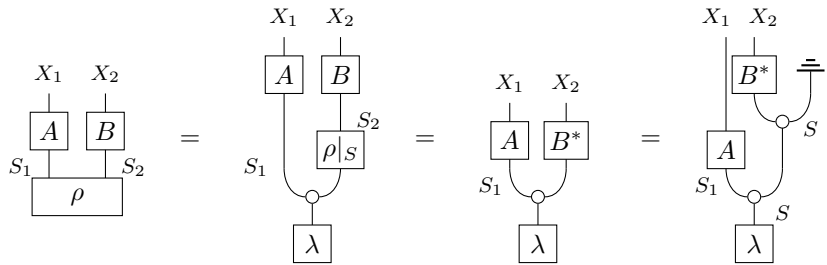
\[
\input{seqBayes2.tikz}
\]
In the first step we rewrite $\Mo=P(S,S')$ where $\lambda=P(S)$ and $\Mo|_S=P(S' \mid S)$. In the second step we define $B^* := B \circ \Mo|_S$. This form resembles the usual attempted factorisation of joint data in terms of a local hidden variable $\lambda$ in treatments of Bell's theorem. In the final step we see this is equal to the distribution induced by a Bayesian model.
Hence Bell-CHSH inequalities can be taken as tests for the availability of sequential Bayesian models.
\end{remark} 

\subsection{Bell arguments in cognition} \label{subsec:Bell-arguments}

Let us now discuss how Bell-type arguments could be applied in cognition, as alluded to above. Suppose we have the simple form of joint decision data given by the distributions $P(A \otimes B)$, $P(A' \otimes B)$, $P(A \otimes B')$, and $P(A' \otimes B')$ over $\Xone \times \Xtwo$,
coming from a cognitive experiment. 
Moreover suppose that the data:
\begin{enumerate}[label=(\alph*)]
    \item \label{enum:nat}
    is naturally represented as joint decision data and modelled with a parallel decision model; 
    \item \label{enum:no-signal}
    satisfies no-signalling \eqref{eq:no-signalling};
    \item \label{enum:Bell-violates}
    violates 
     the Bell-CHSH inequality \eqref{eq:Bell-CHSH-QQ};
     \item \label{enum:Tsirelson-bound}
    satisfies the Tsirelson bound.
\end{enumerate}

Then while \ref{enum:nat} means it is natural to model the data with a parallel decision model, and \ref{enum:no-signal} means it is possible to do so, thanks to \ref{enum:Bell-violates} and Theorem \ref{thm:joint-dec-thm} no such classical model can exist. Unlike the cognitive effects of Section \ref{sec:cog-effects} which ruled out Bayesian models but not simple functional models or probability channels, this would be a much stronger result, stating that there is strictly no way to implement the boxes in \eqref{eq:Bell-CHSH-QQ} as probability channels and recover the data.
Thus we would obtain a strong argument for a setting where non-classical models are needed in cognition. If moreover \ref{enum:Tsirelson-bound} holds then a quantum parallel decision model would instead be available, providing an argument for considering quantum models specifically.

\paragraph{The use of parallel models.}
While the remaining points are simple facts about the data, the strength of such an argument would hinge crucially on the argument for point \ref{enum:nat}, namely that the cognitive scenario should indeed be modelled with a parallel model. If this is unconvincing then one would always be free to simply use a sequential classical model, which is always available by Theorem \ref{thm:-simple-M-Model}.

To claim such naturality, our desiderata for the decisions in the scenario themselves are as follows. Ideally, it should be natural to think of each decision $A \jdec B$ as really consisting of two decisions $A$ and $B$, and moreover we should have data for asking the decisions $A, B$ sequentially and find that they exhibit no order effects, for each such parallel pair. More weakly, we may simply have a collection of decisions which have `two aspects'  to their outcomes, which we model as  $A \jdec B$, but we cannot ask each decision independently. In either case, it should be natural to treat $A, B$ as relating to aspects of the joint decision, pertaining to respective parallel cognitive spaces $\Hione, \Hitwo$.

As an aid to intuition, we tentatively offer some candidate proposals for experiments. We stress however that it would be the role of cognitive scientists to develop more refined proposals suitable for actual experiments. 

\begin{dataexample}
A candidate experiment would involve presenting a participant with two distinct stimuli $\Hione, \Hitwo$. The decisions $\decA, \decA'$ are yes-no questions about whether stimulus $\Hione$ has a given property, and $\decB, \decB'$ similarly for $\Hitwo$. On each run of the experiment a participant is asked one parallel pair of such questions. For example, participants could be shown:
\begin{itemize}

\item a video, with $\Hione$ being the sound and $\Hitwo$ the image, $\decA, \decA'$ are whether the sound was `loud' or `far away' and $\decB, \decB'$ whether it was `red' or `natural', respectively. More generally, $\Hione$ and $\Hitwo$ can refer to different modes (e.g. sound and taste) of one entity;

\item two images $\Hione$ and $\Hitwo$, where $\decA, \decA'$ are whether $\Hione$ is `red' or `natural' and $\decB, \decB'$ again the same properties for $\Hitwo$, respectively. 
\end{itemize}

No-signalling would amount to the probability that a participant says $\Hione$ has property $\decA$ is independent of which property $\decB, \decB'$ we ask about stimulus $\Hitwo$, etc.\footnote{To expect to use a quantum model with entanglement, we might also expect e.g.~the properties $\decA, \decA'$ to be suitably `incompatible', exhibiting order effects if asked sequentially.}
\end{dataexample}

\paragraph{Challenges.}
A large body of experimental evidence of joint decision data satisfying requirements \ref{enum:nat} - \ref{enum:Tsirelson-bound} would provide major evidence for non-classical and potentially quantum-like models of cognition. In practice, it may be difficult to find scenarios which can really be said to consist of two `independent' aspects so as to satisfy point \ref{enum:nat}. Even more so, the difficulty in finding cognitive data satisfying no-signalling \ref{enum:no-signal} is perhaps the main barrier towards applying such arguments in practice, due to likely presence of interconnections in cognitive systems. Note that these points are closely related: any data satisfying this condition would strongly suggest the use of parallel models as an explanation.  For this reason, it could be natural to explore areas already modelled via parallel models classically in so-called Bayesian-approaches such as active inference \cite{parr2022active}, in that case simply corresponding to the use of separate variables in the model.  

\begin{remark}
In the case where Bell violations were found in a cognitive setting, a major question would lie in identifying the correct interpretation of such non-classicality: must it require quantum physical aspects or `non-locality' in the brain itself? A related investigation of how Bell inequality violations can in fact arise from an underlying classical neuronal implementation, and a discussion of how such ‘mental entanglement’ can be tested experimentally, is found in the recent article \cite{khrennikov2025quantum}.
\end{remark}

\paragraph{Related work: Entangled concepts.} 

Arguments invoking entanglement and Bell's theorem have in fact already been made frequently in the quantum-like cognition literature, most commonly in the representation of human \emph{concepts}, e.g.~\cite{gabora2002contextualizing,
aerts2005theory,
aerts2009quantum}, \cite[Chapter 5]{busemeyer2012quantum}.  However, crucially none of these approaches satisfy desiderata \ref{enum:no-signal} above, and so refer to data for which no parallel decision model is available (classical, quantum or otherwise); for a detailed discussion see Appendix \ref{sec:entangled-concepts}. 

\subsection{Future direction: more general compositional models} \label{sec:comp-models}

If a valid Bell-type argument can be made to rule out classical models, this would motivate an exploration of cognitive modelling making more extensive use of the parallel composition of objects and processes, as well as sequential. It is also precisely such modelling which makes the use of process theories (i.e.~monoidal categories) most natural. In particular, we could consider the following generalisation of parallel instrument models. 
\begin{definition}
By a \emph{factorised} decision model we mean an instrument model for which $S$ comes with a given factorisation as a monoidal product: 
\[
S=S_1 \otimes \dots \otimes S_n
\]
of further objects $S_1, \dots, S_n$, called the \emph{factors}, and where each instrument $A$ is given along with a specified partition $S = S_A \coprod S'_{A}$ of the factors such that:
\[
\input{factored-inst2.tikz}
\]
for some instrument $A^*$ on $S_A$. We call the subset $S_A$ the \emph{domain} of $A$. 
\end{definition}

One area where factorised models are useful is in Gärdenfors' framework of \emph{conceptual spaces}, which uses mental 
representation spaces which are internally structured in terms of finitely many separate \emph{domains} such as those of 
colours, smells, and tastes. A focus of the approach is on concepts, which can be modelled as processes $C \colon S \to I$ on such a space.
The conceptual space framework can be generalised categorically in terms of various semantics: its usual form based 
on convex sets \cite{bolt2019interacting}, a more general form allowing fuzzy or graded concepts on measurable spaces \cite{tull2021categorical}, and also in a quantum semantics allowing for modelling `quantum concepts' \cite{tull2024conceptual}. Exploring more detailed connections with the quantum-like literature would be interesting for future work.

Another classical cognitive approach using factorised models are `Bayesian Brain' approaches based on Bayesian networks. The latter are also naturally modelled in string diagrams, allowing for categorical treatments of causal reasoning \cite{lorenz2023causal}, including the generative models used in Bayesian approaches such as active inference \cite{tull2023active}. In any area where such models are applied, if experimental data found a Bell inequality violation, it could motivate using instead a form of quantum or more general model with similar compositional structure. In such a case, while the desire for compositionality may have initially motivated the use of (classical) Bayesian networks, it would ultimately require the use of quantum models.

\section{Discussion} \label{sec:Discussion}

In this article we have aimed to demonstrate that probabilistic process theories provide a natural setting in which to define and study cognitive decision models, including both classical and quantum-like models.  

We gave diagrammatic accounts of many cognitive effects including order, interference, replicability effects, the QQ-equality, conjunction fallacies, and Bell inequalities, and in the appendices: disjunction fallacies, changes of perspective, reciprocity, double stochasticity, and temporal Bell inequalities. Our general approach allowed us to properly compare classical and quantum models in a uniform setting and assess whether there are any strict arguments for the use of quantum models. 

In Section \ref{sec:classical-models} we saw that in fact any sequential decision data can be given a classical instrument model, as proven in Theorem \ref{thm:-simple-M-Model}. Even simple deterministic classical models based on simple functions can capture all of the cognitive effects which usually motivate quantum ones, as exemplified in the model of \cite{ozawa2021modeling} in Example \ref{ex:OK-det-model-CG}. Crucially, general classical instruments (deterministic or otherwise) can allow for the state of the system to be altered and updated, allowing for such effects.

Hence if one is willing to use more general classical instruments, described by probability channels, then no sequential decision data, or any cognitive effects which can be derived from such data (interference, order, disjunction, conjunction effects, or the QQ-equality) will force one to consider quantum models. To fully explore the need for non-classical models within cognition, we then identified two possible routes forward.

\paragraph{Generalising measurement models.} 
The first option is to restrict attention to only allowing measurement models, as discussed in Section \ref{sec:meas-models}. These provide a structurally constrained class of instrument models, typically requiring fewer parameters and which can be equated with a logic or probability theory. The literature usually does so, comparing Bayesian models (`classical logic') against quantum projective or POVM models (`quantum logic'). After this restriction, classical models indeed fail to account for the cognitive effects while quantum (projective) models can. 

However, a crucial issue identified in \cite{khrennikov2014quantum} and here in Theorem \ref{thm:RRE-no-order} is that even quantum measurement models are unable to capture both order effects and RRE together. The key question is whether one can generalise measurement models in order to account for all of our cognitive effects at once. 

\begin{center}
\fbox{
    \begin{tabular}{@{}c@{}}
     Can we capture all of the cognitive effects with either \\ 
     measurement models in an alternative probabilistic process theory \\  or
     a generalised classical measurement theory?
    \end{tabular}%
}
\end{center}

One option could be to seek a class of instruments within broader than POVMs, satisfying the effects, but less general than arbitrary instruments, and which can still be thought of as measurements. Alternatively we could explore measurements within more general probabilistic process theories $\catC$. A minimal criteria would be that $\catC$ includes (update) instruments exhibiting interference effects, while also potentially satisfying RRE.

More generally, it would be interesting to apply our framework to  decision making in more general process theories $\catC$ beyond classical and quantum theory, particularly those which feature state disturbance as a fundamental aspect of interaction. An interesting candidate could be those based on epistemic restrictions, such as Spekkens’ toy theory \cite{spekkens2007evidence}. We note that the exploration of cognition in other generalised probabilistic theories has also been initiated in \cite{khrennikov2025coupling}. 

Whether by altering the notion of measurement, or replacing the process theory $\catC=\Class$ by an alternative, these investigations could ultimately amount to developing a more general form of \emph{classical measurement theory} which does allow for disturbance from measurements. It could turn out that ultimately classical models which are `quantum-like' in the simple sense that they feature disturbance under interaction may provide the best explanations for the cognitive effects explored here. 

\paragraph{Instruments and Bell Arguments.}
In the second route we instead allow the use of arbitrary instruments as models. Classical instruments include arbitrary probability channels, including deterministic models, while quantum instruments are given by families of jointly trace-preserving CP maps. Given that even classical deterministic models can then capture all the effects in Section \ref{sec:cog-effects} (as well as those in Appendices \ref{app:cond-data} and \ref{app:further-effects}) the key question to address is then the following. 

\begin{center}
\fbox{
    \begin{tabular}{@{}c@{}}
    Are there cognitive effects which can be modelled with quantum instruments, \\ but not (deterministic) classical instruments? 
    \end{tabular}%
}
\end{center}

Any such effects must go beyond sequential decision data and introduce further structure in the models. A natural candidate, and one which makes best use of monoidal categories, is to consider parallel composition of systems and decision processes, as we explored in Section \ref{sec:Bell-arguments}. In particular we saw that if experimental data for parallel joint decisions could be found satisfying suitable conditions, then Bell's Theorem can be applied to rule out any classical parallel decision model. This would provide a strong result, stating that such experiments simply cannot have any form of classical model based on parallel decision making processes. In this case one could then explore quantum parallel decision models as well as those in more general process theories. 

The challenge is thus in finding a compelling cognitive scenario which is natural to model through parallel processes, and for which real experimental data is found which violate a Bell inequality while satisfying no-signalling (in Appendix \ref{sec:entangled-concepts} we discuss how previous uses of Bell's Theorem in the quantum-like modelling literature have failed the latter condition). In particular, the greatest challenge may be that it may be difficult or impossible to find experimental data which satisfies no-signalling, due to the interconnected nature of cognition. 

\begin{center}
\fbox{
    \begin{tabular}{@{}c@{}}
    Are there compelling examples of parallel decision making in cognitive science \\ which experimentally satisfy no-signalling, but violate a Bell inequality? 
    \end{tabular}
}
\end{center}

Overall, while finding such valid cases in cognitive science remains a major challenge, we believe it is worth further investigation to look for such experimental scenarios. If achieved, such data would provide an argument of a strength previously unknown in quantum-like cognition, and one which could arguably have major import on cognitive science: the existence of data which simply cannot be modelled classically.

As well as the above there are several further aspects of quantum-like cognition which could be considered.

\paragraph{Empirical arguments.}
We have focused on looking for strict black-or-white structural arguments for quantum models over classical ones. However, even if no such arguments are available, and quantum models are not strictly essential to fit given data, they may still be valuable. For example, a quantum model may be seen to empirically provide the `best' model, either because it provides a richer explanatory model, or requires fewer parameters than alternatives. The classical models provided by Theorem \ref{thm:-simple-M-Model}, for example, essentially encode the entire probability distributions of the data within the instruments, and so could be deemed unsatisfactory on both accounts. In contrast, if for example a large decision dataset could be explained by a quantum projective model using very few parameters, this would be seen as compelling evidence for the use of such models, even if classical models are available. 

Thus, even without strict structural arguments, much of the assessment of the relative merits of models will come down to a more subtle empirical assessment of the utility of quantum vs classical models in psychology or cognitive science. In future it would be interesting to further understand any potential explanatory benefits of quantum cognitive models in particular.

\paragraph{Further topics in the literature.}
While we aimed to cover a wide range of aspects of quantum-like cognition from \cite{busemeyer2012quantum}, there are undoubtedly further aspects we have omitted which would be interesting to explore in terms of process theories. One such topic is in the use of quantum theory to model concepts, and the use of these in quantum AI (discussed briefly in Section \ref{sec:entangled-concepts}, see also \cite{tull2024conceptual}). Another interesting topic is the neuronal basis for quantum-like models, in which one justifies how a high-level quantum representation can emerge from lower-level classical neural models such as the brain or artificial neural networks \cite{khrennikov2025coupling}. We also did not in detail cover aspects of continuous time evolution and dynamics in decision models, such as the quantum Zeno effect \cite{yearsley2016zeno}.

\paragraph{Causal aspects.}
We have also overlooked causal aspects of cognition; for this we merely note that causal models (causal Bayesian networks) and reasoning are both naturally described in terms of string diagrams \cite{jacobs2019causal,fritz2023d,lorenz2023causal}. Another interesting topic in quantum foundations are indefinite causal orders \cite{kissinger2019categorical}, and it would be interesting to explore the extent to which these arise in human cognition and are best described by either classical or quantum causal models.

\section*{Acknowledgements}
Masanao Ozawa~was supported by JSPS KAKENHI Grant Numbers JP25K07108, JP25H00595, JP24H01566, and JST CREST Grant Number JPMJCR23P4, Japan. This work has benefited from discussions and input from a range of researchers and events. We thank Yasuyoshi Yonezawa, Haruki Emori, Miho Fuyama, Andrei Khrennikov, Emmanuel Pothos, Peter Bruza and Bob Coecke for valuable discussions. We would also like to thank participants in the discussion group on quantum-like cognition at the Ernst Strüngmann Forum \emph{Simplicity behind Absurdity: The Power of Quantum Thinking} held in Frankfurt in September 2025, as well as the organisers and participants at events where the work was presented: the satellite workshop \emph{Mathematical and philosophical models of consciousness: from classical to quantum-like} of ASSC27 at Ritsumeikan University, Kyoto in July 2024, and at the \emph{Models of Consciousness} conference at Hokkaido University in October 2025.

\bibliographystyle{alpha}
\bibliography{qcog}

\newcommand{\etalchar}[1]{$^{#1}$}
\begin{thebibliography}{CJWW15}

\bibitem[AA95]{aerts1995applications}
Diedrik Aerts and Sven Aerts.
\newblock Applications of quantum statistics in psychological studies of
  decision processes.
\newblock {\em Foundations of Science}, 1(1):85--97, 1995.

\bibitem[AAB{\etalchar{+}}18]{aerts2018spin}
Diederik Aerts, Jonito~Aerts Argu{\"e}lles, Lester Beltran, Suzette Geriente,
  Massimiliano Sassoli~de Bianchi, Sandro Sozzo, and Tomas Veloz.
\newblock Spin and wind directions {I}: Identifying entanglement in nature and
  cognition.
\newblock {\em Foundations of Science}, 23(2):323--335, 2018.

\bibitem[Aer09]{aerts2009quantum}
Diederik Aerts.
\newblock Quantum structure in cognition.
\newblock {\em Journal of Mathematical Psychology}, 53(5):314--348, 2009.

\bibitem[Aer14]{aerts2014quantum}
Diederik Aerts.
\newblock Quantum and concept combination, entangled measurements, and
  prototype theory.
\newblock {\em Topics in Cognitive Science}, 6(1):129--137, 2014.

\bibitem[AF10]{atmanspacher2010proposed}
Harald Atmanspacher and Thomas Filk.
\newblock A proposed test of temporal nonlocality in bistable perception.
\newblock {\em Journal of Mathematical Psychology}, 54(3):314--321, 2010.

\bibitem[AG05]{aerts2005theory}
Diederik Aerts and Liane Gabora.
\newblock A theory of concepts and their combinations {I}: The structure of the
  sets of contexts and properties.
\newblock {\em Kybernetes}, 34(1/2):167--191, 2005.

\bibitem[AS11]{aerts2011quantum}
Diederik Aerts and Sandro Sozzo.
\newblock Quantum structure in cognition: Why and how concepts are entangled.
\newblock In {\em International Symposium on Quantum Interaction}, pages
  116--127. Springer, 2011.

\bibitem[Bar07]{barrett2007information}
Jonathan Barrett.
\newblock Information processing in generalized probabilistic theories.
\newblock {\em Physical Review A—Atomic, Molecular, and Optical Physics},
  75(3):032304, 2007.

\bibitem[BB12]{busemeyer2012quantum}
Jerome~R Busemeyer and Peter~D Bruza.
\newblock {\em Quantum models of cognition and decision}.
\newblock Cambridge University Press, 2012.

\bibitem[BCG{\etalchar{+}}19]{bolt2019interacting}
Joe Bolt, Bob Coecke, Fabrizio Genovese, Martha Lewis, Dan Marsden, and Robin
  Piedeleu.
\newblock Interacting conceptual spaces i: Grammatical composition of concepts.
\newblock In {\em Conceptual spaces: Elaborations and applications}, pages
  151--181. Springer, 2019.

\bibitem[BCP{\etalchar{+}}14]{brunner2014bell}
Nicolas Brunner, Daniel Cavalcanti, Stefano Pironio, Valerio Scarani, and
  Stephanie Wehner.
\newblock Bell nonlocality.
\newblock {\em Reviews of modern physics}, 86(2):419--478, 2014.

\bibitem[Bel64]{bell1964einstein}
John~S Bell.
\newblock On the einstein podolsky rosen paradox.
\newblock {\em Physics Physique Fizika}, 1(3):195, 1964.

\bibitem[BKRS15]{bruza2015probabilistic}
Peter~D Bruza, Kirsty Kitto, Brentyn~J Ramm, and Laurianne Sitbon.
\newblock A probabilistic framework for analysing the compositionality of
  conceptual combinations.
\newblock {\em Journal of Mathematical Psychology}, 67:26--38, 2015.

\bibitem[BWLM09]{busemeyer2009empirical}
Jerome~R Busemeyer, Zheng Wang, and Ariane Lambert-Mogiliansky.
\newblock Empirical comparison of markov and quantum models of decision making.
\newblock {\em Journal of Mathematical Psychology}, 53(5):423--433, 2009.

\bibitem[CDP10]{chiribella2010probabilistic}
Giulio Chiribella, Giacomo~Mauro D’Ariano, and Paolo Perinotti.
\newblock Probabilistic theories with purification.
\newblock {\em Physical Review A—Atomic, Molecular, and Optical Physics},
  81(6):062348, 2010.

\bibitem[CDP11]{chiribella2011informational}
Giulio Chiribella, Giacomo~Mauro D’Ariano, and Paolo Perinotti.
\newblock Informational derivation of quantum theory.
\newblock {\em Physical Review A—Atomic, Molecular, and Optical Physics},
  84(1):012311, 2011.

\bibitem[CHSH69]{clauser1969proposed}
John~F Clauser, Michael~A Horne, Abner Shimony, and Richard~A Holt.
\newblock Proposed experiment to test local hidden-variable theories.
\newblock {\em Physical review letters}, 23(15):880, 1969.

\bibitem[CJ19]{cho2019disintegration}
Kenta Cho and Bart Jacobs.
\newblock Disintegration and bayesian inversion via string diagrams.
\newblock {\em Mathematical Structures in Computer Science}, 29(7):938--971,
  2019.

\bibitem[CJWW15]{cho2015introduction}
Kenta Cho, Bart Jacobs, Bas Westerbaan, and Abraham Westerbaan.
\newblock An introduction to effectus theory.
\newblock {\em arXiv preprint arXiv:1512.05813}, 2015.

\bibitem[CK18]{coecke2018picturing}
Bob Coecke and Aleks Kissinger.
\newblock Picturing quantum processes: A first course on quantum theory and
  diagrammatic reasoning.
\newblock In {\em International conference on theory and application of
  diagrams}, pages 28--31. Springer, 2018.

\bibitem[CKT{\etalchar{+}}09]{conte2009mental}
Elio Conte, Andrei~Yuri Khrennikov, Orlando Todarello, Antonio Federici,
  Leonardo Mendolicchio, and Joseph~P Zbilut.
\newblock Mental states follow quantum mechanics during perception and
  cognition of ambiguous figures.
\newblock {\em Open Systems \& Information Dynamics}, 16(01):85--100, 2009.

\bibitem[Coe16]{coecke2016terminality}
Bob Coecke.
\newblock Terminality implies no-signalling... and much more than that.
\newblock {\em New Generation Computing}, 34(1):69--85, 2016.

\bibitem[DK14]{dzhafarov2014selective}
Ehtibar~N Dzhafarov and Janne~V Kujala.
\newblock On selective influences, marginal selectivity, and bell/chsh
  inequalities.
\newblock {\em Topics in Cognitive Science}, 6(1):121--128, 2014.

\bibitem[DL70]{davies1970operational}
E~Brian Davies and John~T Lewis.
\newblock An operational approach to quantum probability.
\newblock {\em Communications in Mathematical Physics}, 17(3):239--260, 1970.

\bibitem[FK23]{fritz2023d}
Tobias Fritz and Andreas Klingler.
\newblock The d-separation criterion in categorical probability.
\newblock {\em Journal of Machine Learning Research}, 24(46):1--49, 2023.

\bibitem[Fri20]{fritz2020synthetic}
Tobias Fritz.
\newblock A synthetic approach to markov kernels, conditional independence and
  theorems on sufficient statistics.
\newblock {\em Advances in Mathematics}, 370:107239, 2020.

\bibitem[GA02]{gabora2002contextualizing}
Liane Gabora and Diederik Aerts.
\newblock Contextualizing concepts using a mathematical generalization of the
  quantum formalism.
\newblock {\em Journal of Experimental \& Theoretical Artificial Intelligence},
  14(4):327--358, 2002.

\bibitem[GS17]{gogioso2017categorical}
Stefano Gogioso and Carlo~Maria Scandolo.
\newblock Categorical probabilistic theories.
\newblock {\em arXiv preprint arXiv:1701.08075}, 2017.

\bibitem[Gud73]{gudder1973convex}
Stan Gudder.
\newblock Convex structures and operational quantum mechanics.
\newblock {\em Communications in mathematical Physics}, 29(3):249--264, 1973.

\bibitem[Har11]{hardy2011reformulating}
Lucien Hardy.
\newblock Reformulating and reconstructing quantum theory.
\newblock {\em arXiv preprint arXiv:1104.2066}, 2011.

\bibitem[HET{\etalchar{+}}25]{huang2025overview}
Jiaqi Huang, Gunnar Epping, Jennifer~S Trueblood, James~M Yearsley, Jerome~R
  Busemeyer, and Emmanuel~M Pothos.
\newblock An overview of the quantum cognition research program.
\newblock {\em Psychonomic Bulletin \& Review}, pages 1--50, 2025.

\bibitem[HP96]{hameroff1996orchestrated}
Stuart Hameroff and Roger Penrose.
\newblock Orchestrated reduction of quantum coherence in brain microtubules: A
  model for consciousness.
\newblock {\em Mathematics and computers in simulation}, 40(3-4):453--480,
  1996.

\bibitem[HV19]{heunen2019categories}
Chris Heunen and Jamie Vicary.
\newblock {\em Categories for Quantum Theory: an introduction}.
\newblock Oxford University Press, 2019.

\bibitem[Jac15]{jacobs2015new}
Bart Jacobs.
\newblock New directions in categorical logic, for classical, probabilistic and
  quantum logic.
\newblock {\em Logical Methods in Computer Science}, 11, 2015.

\bibitem[Jac17]{jacobs2017quantum}
Bart Jacobs.
\newblock Quantum effect logic in cognition.
\newblock {\em Journal of Mathematical Psychology}, 81:1--10, 2017.

\bibitem[Jed17]{jedlicka2017revisiting}
Peter Jedlicka.
\newblock Revisiting the quantum brain hypothesis: toward quantum (neuro)
  biology?
\newblock {\em Frontiers in molecular neuroscience}, 10:366, 2017.

\bibitem[JKZ19]{jacobs2019causal}
Bart Jacobs, Aleks Kissinger, and Fabio Zanasi.
\newblock Causal inference by string diagram surgery.
\newblock In {\em International conference on foundations of software science
  and computation structures}, pages 313--329. Springer, 2019.

\bibitem[KBDB14]{khrennikov2014quantum}
Andrei Khrennikov, Irina Basieva, Ehtibar~N Dzhafarov, and Jerome~R Busemeyer.
\newblock Quantum models for psychological measurements: an unsolved problem.
\newblock {\em PloS one}, 9(10):e110909, 2014.

\bibitem[Khr]{khrennikov2004information}
Andre{\u\i} Khrennikov.
\newblock {\em Information dynamics in cognitive, psychological, social, and
  anomalous phenomena}.

\bibitem[KOBS25]{khrennikov2025coupling}
Andrei Khrennikov, Masanao Ozawa, Felix Benninger, and Oded Shor.
\newblock Coupling quantum-like cognition with the neuronal networks within
  generalized probability theory.
\newblock {\em Journal of Mathematical Psychology}, 125:102923, 2025.

\bibitem[KU19]{kissinger2019categorical}
Aleks Kissinger and Sander Uijlen.
\newblock A categorical semantics for causal structure.
\newblock {\em Logical Methods in Computer Science}, 15, 2019.

\bibitem[KY25]{khrennikov2025quantum}
Andrei Khrennikov and Makiko Yamada.
\newblock Quantum-like representation of neuronal networks' activity: modeling
  “mental entanglement”.
\newblock {\em Frontiers in Human Neuroscience}, 19:1685339, 2025.

\bibitem[LT23]{lorenz2023causal}
Robin Lorenz and Sean Tull.
\newblock Causal models in string diagrams.
\newblock {\em arXiv preprint arXiv:2304.07638}, 2023.

\bibitem[Moo02]{moore2002measuring}
David~W Moore.
\newblock Measuring new types of question-order effects: Additive and
  subtractive.
\newblock {\em The Public Opinion Quarterly}, 66(1):80--91, 2002.

\bibitem[OK19]{ozawa2019application}
Masanao Ozawa and Andrei Khrennikov.
\newblock Application of theory of quantum instruments to psychology:
  Combination of question order effect with response replicability effect.
\newblock {\em Entropy}, 22(1):37, 2019.

\bibitem[OK21]{ozawa2021modeling}
Masanao Ozawa and Andrei Khrennikov.
\newblock Modeling combination of question order effect, response replicability
  effect, and qq-equality with quantum instruments.
\newblock {\em Journal of Mathematical Psychology}, 100:102491, 2021.

\bibitem[OK23]{ozawa2023nondistributivity}
Masanao Ozawa and Andrei Khrennikov.
\newblock Nondistributivity of human logic and violation of response
  replicability effect in cognitive psychology.
\newblock {\em Journal of Mathematical Psychology}, 112:102739, 2023.

\bibitem[PB22]{pothos2022quantum}
Emmanuel~M Pothos and Jerome~R Busemeyer.
\newblock Quantum cognition.
\newblock {\em Annual review of psychology}, 73(1):749--778, 2022.

\bibitem[Pit89]{pitowsky1989quantum}
Itamar Pitowsky.
\newblock {\em Quantum probability—quantum logic}.
\newblock Springer, 1989.

\bibitem[Pl{\'a}23]{plavala2023general}
Martin Pl{\'a}vala.
\newblock General probabilistic theories: An introduction.
\newblock {\em Physics Reports}, 1033:1--64, 2023.

\bibitem[PPF22]{parr2022active}
Thomas Parr, Giovanni Pezzulo, and Karl~J Friston.
\newblock {\em Active inference: the free energy principle in mind, brain, and
  behavior}.
\newblock MIT Press, 2022.

\bibitem[Sel10]{selinger2010survey}
Peter Selinger.
\newblock A survey of graphical languages for monoidal categories.
\newblock In {\em New structures for physics}, pages 289--355. Springer, 2010.

\bibitem[Spe07]{spekkens2007evidence}
Robert~W Spekkens.
\newblock Evidence for the epistemic view of quantum states: A toy theory.
\newblock {\em Physical Review A—Atomic, Molecular, and Optical Physics},
  75(3):032110, 2007.

\bibitem[TK83]{tversky1983extensional}
Amos Tversky and Daniel Kahneman.
\newblock Extensional versus intuitive reasoning: The conjunction fallacy in
  probability judgment.
\newblock {\em Psychological review}, 90(4):293, 1983.

\bibitem[TKS23]{tull2023active}
Sean Tull, Johannes Kleiner, and Toby St~Clere Smithe.
\newblock Active inference in string diagrams: A categorical account of
  predictive processing and free energy.
\newblock {\em arXiv preprint arXiv:2308.00861}, 2023.

\bibitem[TS92]{tversky1992disjunction}
Amos Tversky and Eldar Shafir.
\newblock The disjunction effect in choice under uncertainty.
\newblock {\em Psychological science}, 3(5):305--310, 1992.

\bibitem[TSZC24]{tull2024conceptual}
Sean Tull, Razin~A Shaikh, Sara~Sabrina Zemlji{\v{c}}, and Stephen Clark.
\newblock From conceptual spaces to quantum concepts: formalising and learning
  structured conceptual models.
\newblock {\em Quantum Machine Intelligence}, 6(1):21, 2024.

\bibitem[Tul16]{tull2016operational}
Sean Tull.
\newblock Operational theories of physics as categories.
\newblock {\em arXiv preprint arXiv:1602.06284}, 2016.

\bibitem[Tul21]{tull2021categorical}
Sean Tull.
\newblock A categorical semantics of fuzzy concepts in conceptual spaces.
\newblock {\em arXiv preprint arXiv:2110.05985}, 2021.

\bibitem[WSSB14]{wang2014context}
Zheng Wang, Tyler Solloway, Richard~M Shiffrin, and Jerome~R Busemeyer.
\newblock Context effects produced by question orders reveal quantum nature of
  human judgments.
\newblock {\em Proceedings of the National Academy of Sciences},
  111(26):9431--9436, 2014.

\bibitem[YP16]{yearsley2016zeno}
James~M Yearsley and Emmanuel~M Pothos.
\newblock Zeno's paradox in decision-making.
\newblock {\em Proceedings of the Royal Society B: Biological Sciences},
  283(1828):20160291, 2016.

\end{thebibliography}

\appendix

\section{Conditional Effects} \label{app:cond-data}

In this section we will meet further cognitive effects which apply not directly to sequential distributions such as $\PAB$, but instead to related conditional probability distributions.
Given any probability distribution $P(A,B)$ over $X_A, X_B$ recall that the conditional probability (partial) channel $P(B \mid A) \colon X_A \to X_B$ is defined by 
\begin{equation} \label{eq:cond-formula}
P(By \mid Ax) := \begin{cases} \frac{P(Ax, By)}{P(Ax)} & P(Ax) > 0 \\ 0 & \text{otherwise} \end{cases}
\end{equation}
for all $y \in X_B$, where $P(Ax) := \sum_{y \in X_B} P(Ax,By)$. This defines a channel whenever $P(Ax) > 0$ for all $x \in X$, and a \emph{partial} channel otherwise.

Conditioning can also be described diagrammatically; for more details see \cite{lorenz2023causal,cho2019disintegration}. Firstly, for any non-zero state $\omega$ of a classical object $X$ we define its \emph{normalisation} $\normop(\omega)$ to either be the zero state or the unique normalised state with $\omega = (\discard{} \circ \omega) \cdot \normop(\omega)$. More generally, for any classical process $f \colon X \to Y$ we define its \emph{normalisation} by $\normop(f)(x) := \normop(f(x))$ whenever this is well-defined and $\normop(f)(x) = 0$ otherwise. In diagrams we depict normalisation by drawing a dashed box as below. 
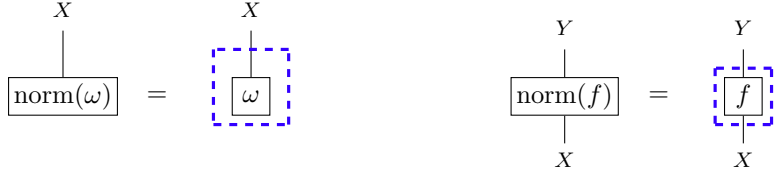
\[
\input{normomega.tikz} 
\qquad \qquad \qquad \qquad 
\input{normf.tikz}
\] 
Next, recall from \eqref{eq:cap-def} the `cap' process $\tinycap \colon X \times X \to I$ for a classical object $X$, which is equivalently defined by: 
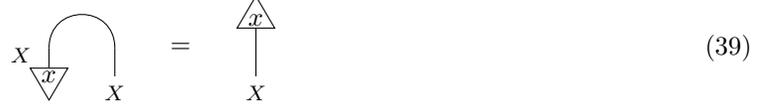
\begin{equation} \label{eq:classbend}
\input{classbend.tikz}
\end{equation}
for all $x \in X$.
Now, given any probability distribution $\PAB$ for decision $A$ then decision $B$, we can define the conditional probability $P(B \mid A)$ as in \eqref{eq:cond-formula} treating $\PAB$ as the joint distribution over $X_A \times X_B$. Equivalently: 
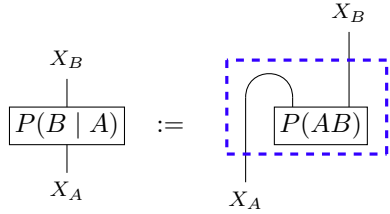
\[
\input{cond-channel.tikz}
\]
Conditioning on any particular outcome $x \in X_A$ then corresponds to the following state. 
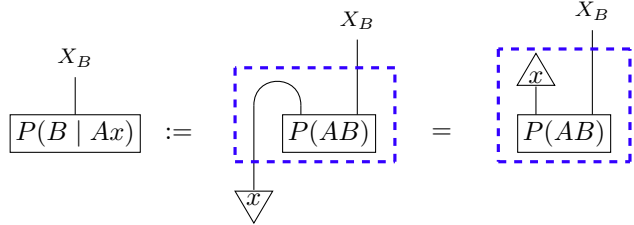
\[
\input{cond-state.tikz}
\]
Here we used \eqref{eq:classbend} along with the property that normalised states for point distributions can pass inside normalisation boxes (which does not hold for general states); see \cite{lorenz2023causal} for more details. 

One can more generally define distributions conditioning on an initial subsequence of decisions. Given a joint distribution $P(A_1,\dots,A_n,B_1,\dots,B_m)$ we define the conditional probability channel $P(B_1,\dots,B_m \mid A_1,\dots, A_n)$ by:
\begin{equation} \label{eq:cond-general}
P(B_1b_1,\dots,B_mb_m \mid A_1a_1,\dots,A_na_n) := \frac{P(A_1a_1,\dots,A_na_n,B_1b_1,\dots,B_mb_m)}{P(A_1a_1,\dots,A_na_n)}
\end{equation}
We can treat these diagrammatically in just the same way by bending multiple wires and normalising.

\begin{definition}
A set of \emph{conditional decision data} with decisions $\decset$ is defined just like a set of decision data but now containing probability channels of the form $\PBtomCAton$ for certain sequences of decisions $A_i, B_j \in \decset$. We say that an instrument model $\modelM$ of $\decset$ models the conditional decision data when for all such conditional distributions in the data we have $\PMBtomCAton = \PBtomCAton$.\footnote{Here $\PMBtomCAton$ is the conditional channel constructed from $\PAtonBtom$. Note that we can always assume each $P(A_1,\dots,A_n \mid B_1,\dots, B_m)$ in the data is a channel (and not merely a partial channel) as each such channel $P_\modelM(\dots)$ will be.}
\end{definition}

Note that the special case where $n=0$ means that such decision data can contain our usual decision sequences $\PBton$ with no conditioning.

\subsection{Law of reciprocity}

The following feature of quantum models based on basis measurements is identified in \cite{busemeyer2012quantum}. 

\begin{definition}[Law of Reciprocity]
 We say that conditional data $P$ for decisions $A, B$ satisfies the \emph{law of reciprocity} when: 
 \[
\input{recip-in-diags.tikz}
 \]
 That is, we have: 
\[
P(Ax \mid By) = P(By \mid Ax)
\]
for all $x \in X_A, y \in X_B$. Similarly, we say a pair of instruments $A, B$ on $S$ satisfy the \emph{law of reciprocity in state $\rho$} when the data $P_\modelM$ induced by the model $\MABrho$ does. 
\end{definition} 

Reciprocity may or may hold for any given decision data. General instrument models, and in particular Bayesian models, need not satisfy this property. However, quantum basis measurements always will.

\begin{lemma}
Let $A, B$ be projective quantum instruments for orthonormal basis measurements. Then $A, B$ satisfy the law of reciprocity in every state $\rho$. 
\end{lemma}
\begin{proof}
Suppose that $A, B$ measure the respective orthonormal bases $\{A_x\}^n_{x=1}, \{B_y\}^n_{y=1}$ of $\hilbH=\mathbb{C}^n$. Let $\modelM=(\hilbH,\rho,A,B)$ for some normalised state $\rho$. Then for each $x, y$ we have: 
\[
\PMAxBy  \ \ = \ \ \input{ONBproof.tikz}  \ \ = \ \ P_\modelM(Ax) \left|\braket{By \mid Ax}\right|^2
\]
Hence: 
\[
P_\modelM(B y \mid A x) := \frac{\PMAxBy}{P_\modelM(Ax)}
= 
\left|\braket{B_y \mid A_x}\right|^2 = \left|\braket{A_x \mid B_y}\right|^2 = P_\modelM(B_y \mid A_x)
\]
\end{proof}

\subsection{Double stochasticity}

Another feature of quantum basis measurements identified in \cite{busemeyer2012quantum} is the following. Say that data $\PAB$ is \emph{well-supported} if its marginals have full support, i.e.~for each $x \in X_A$, $\PAxBy$ is non-zero for at least one $y \in X_B$, and conversely for each such $y$ it is non-zero for at least one $x \in X_A$. Whenever this holds, 
$P(B \mid A)$ will be a probability channel and so stochastic as a matrix. 
Hence we will have: 
\[
\input{stoch.tikz}
\]
Equivalently 
\[
\sum_{y \in X_B} P(By \mid Ax) = 1
\]
for all $x \in X_A$. 
For certain datasets, however a form of converse may also hold. Recall from Proposition~\ref{prop:probPT} that on each classical object $X$ there is a normalised state $\discardflip{}$ such that:
\[
\begin{tikzpicture}
	\begin{pgfonlayer}{nodelayer}
		\node [style=none] (0) at (0, 0.75) {};
		\node [style=downground] (1) at (0, -0.75) {};
		\node [style=classcopoint] (2) at (0, 0.75) {$x$};
		\node [style=label] (3) at (-0.75, 0) {$X$};
	\end{pgfonlayer}
	\begin{pgfonlayer}{edgelayer}
		\draw [style=cwire] (0.center) to (1);
	\end{pgfonlayer}
\end{tikzpicture}
 \ \ = \ \ \frac{1}{|X|}
\]
for all $x \in X$. 

\begin{definition}[Double Stochasticity]
We say that conditional decision data $P(B \mid A)$ satisfies \emph{double stochasticity} when: 
\[
\sum_{x \in X_A} P(By \mid Ax) = 1
\]
for all $y \in X_B$. Rescaling both sides of the equation we can write this diagrammatically as:  \[
\input{doublestochscalars.tikz}
\]
We say that instruments $A, B$ on $S$ satisfy \emph{double stochasticity in state $\rho$} (resp.~is well-supported) when the model $\MABrho$ induces data $P_\modelM$ which does.
\end{definition}

\begin{lemma}
Any well-supported dataset (or instruments) satisfying reciprocity satisfies double stochasticity. 
\end{lemma}
\begin{proof}
We have $\sum_{x \in X_A} P(By \mid Ax) = \sum_{x \in X_A} P(Ax \mid By) = 1$ since $P(A \mid B)$ is a channel. Equivalently, in diagrams:
\[
\input{dstochproofscalar.tikz}
\]
from which the result follows using that the cap is cancellative (which follows from \ref{eq:snake} shown later).
\end{proof}

\begin{corollary}
Any quantum orthonormal basis measurements satisfy double stochasticity.\footnote{More precisely, in any state $\rho$ with non-zero inner product with all the basis vectors.} 
\end{corollary}

\begin{dataexample}
In \cite{busemeyer2012quantum} the categorisation-decision data of Example \ref{ex:Cat-Dec-Exp} is given in two forms, for `narrow' and `wide' faces respectively. In conditional form it appears as follows. 
\[ 
\begin{bmatrix}
P(\actdecD y \mid \catdecC y) & P(\actdecD y \mid \catdecC n) \\ 
P(\actdecD n \mid \catdecC y) & P(\actdecD n \mid \catdecC n)
\end{bmatrix}_\mathsf{narrow}
=
\begin{bmatrix}
0.41 & 0.63 \\ 
0.59 & 0.37
\end{bmatrix}
\qquad \qquad
\begin{bmatrix}
P(\actdecD y \mid \catdecC y) & P(\actdecD y \mid \catdecC n) \\ 
P(\actdecD n \mid \catdecC y) & P(\actdecD n \mid \catdecC n)
\end{bmatrix}_\mathsf{wide}
=
\begin{bmatrix}
0.35 & 0.52 \\ 
0.65 & 0.48
\end{bmatrix}
\]
One can see that the narrow face data approximately satisfies double stochasticity, while the wide face data does not, and hence cannot be given a quantum model using orthonrmal basis measurements. 
\end{dataexample}
\subsection{Disjunction effect}

Like the conjunction fallacy, the following effect, first noted by Tversky and Shafir \cite{tversky1992disjunction} can be understood as a logical fallacy in human reasoning. We follow the account in \cite{busemeyer2012quantum}. 

\begin{definition}[Disjunction Effect]
We say that conditional data for decisions $A, B$ with outcomes $\{y,n\}$
satisfies a \emph{disjunction effect} when both of the following hold: 
\[
\input{gdatay4.tikz} \ \ > \begin{tikzpicture}
	\begin{pgfonlayer}{nodelayer}
		\node [style=cmap] (0) at (4, 0) {$\data(\gtwo)$};
		\node [style=none] (1) at (4, 0.25) {};
		\node [style=classcopoint] (2) at (4, 1.5) {$\gam$};
	\end{pgfonlayer}
	\begin{pgfonlayer}{edgelayer}
		\draw [style=cwire] (2) to (1.center);
	\end{pgfonlayer}
\end{tikzpicture}
 \ \ < \ \ \input{gdatay3.tikz}
\]
i.e.~$P(\Bet y) < P(\Bet y \mid \Win y)$ and $P(\Bet y) < P(\Bet y \mid \Win n)$. Similarly we say instruments $\Win, \Bet$ on $S$ with outcomes $\{y,n\}$ satisfy a \emph{disjunction effect in state $\rho$} when the data $P_\modelM$ induced by the model $\modelM=(S,\rho,\Win,\Bet)$ does. 

\end{definition}

The disjunction effect only requires as data a distribution $P(B)$ and channel $P(B \mid A)$.  If we however are also given a distribution $P(A)$ and hence the sequential distribution $\PAB := P(B \mid A)P(A)$ over $X_A \times X_B$, we
can see that a disjunction effect will always require an interference effect of $\Win$ on $\Bet$. 

\begin{proposition} \ 
\begin{enumerate}
\item 
\label{enum:disj-needs-int}
For data of the form $\{\PAB$, $\PB\}$ a disjunction effect requires an interference effect of $\Win$ on $\Bet$.
\item \label{enum:nobayesdis}
Bayesian models fail to satisfy disjunction effects.
\end{enumerate}
\end{proposition}
\begin{proof}
\eqref{enum:disj-needs-int}: 
For data showing no interference effect we have: 
\begin{align*}
\PBy & =  
\PAyBy + \PAnBy \\  
& = P(By \mid Ay)\PAy + P(By \mid An)\PAn \\ 
& \geq \mathsf{min}(P(By \mid Ay), P(By \mid An))
\end{align*}
since $\PAy + \PAn = 1$, which contradicts the disjunction effect.

\eqref{enum:nobayesdis}: Any instrument model $\modelM$ has conditionals induced by $P_\modelM(\Win, \Bet)$. Hence by \eqref{enum:disj-needs-int}
it must show interference of $\Win$ on $\Bet$, but Bayesian models never show interference by Proposition \ref{prop:interference} \eqref{enum:Bayesnoint}.
\end{proof}

The (failure of the) derivation in the proof of \eqref{enum:disj-needs-int} above corresponds to the usual view of the disjunction effect as a logical `fallacy'. An experimental example of a disjunction effect is the following.

\begin{dataexample}[Gambler's Fallacy\footnote{
Originally the Tversky and Shafir study \cite{tversky1992disjunction} was conceived as a test of the `Sure-thing' principle of decision theory: if under both state of the world $X$ and state of the world $X^\bot$ individually you prefer decision $A$ over decision $B$, then you should prefer decision $A$ over $B$ even when you do not know the state of the world. As noted in \cite[9.2.3]{busemeyer2012quantum} however, it is difficult to interpret the experiment this way since the latter requires an experiment of the same participant under different conditions, rather than proportions for each condition.}
\footnote{
When modelling Example \ref{ex:Gamblers} the information as to whether the previous bet was successful is not really a `decision' made by the participant. Hence the data is given only as conditionals $P(\Bet \mid \Win)$ and $P(\Bet)$. Nonetheless when modelling with instruments this makes no difference and amounts to modelling $\Win$ the same as a previous decision, providing a distribution $\PMAB$.}]
\label{ex:Gamblers}
In an experiment, participants were asked whether they wish to gamble ($\Bet$) their money in a game. They were told they had played the game previously, with separate conditions under which participants were either told they had won the previous game ($\Win y $), lost the previous game ($\Win n $), our the outcome was unknown to them. 

It was found that the proportion of participants willing to gamble in the unknown condition $P(\Bet y)$ was lower than both the proportion when told they had won previously $P(\Bet y \mid \Win y)$ \emph{and} for when told they had lost previously $P(\Bet y \mid \Win n)$, with respective probabilities: 
\begin{center}
$\begin{array}{l|ccc}
& \gtwo\gam  &  \gtwo\gam \mid \gone\win  & \gtwo \gam \mid \gone\lose\\ \cline{1-4}
P &0.37 &   0.69 & 0.58
\end{array}$
\end{center}
\end{dataexample}
The disjunction effect can be modelled by quantum instruments or non-Bayesian classical instruments allowing interference. 

\begin{example}
A projective quantum model of Example \ref{ex:Gamblers} on $\hilbH=\mathbb{C}^3$  was given by Aerts \cite{aerts2009quantum}, \cite[F.2]{busemeyer2012quantum}. We model $\Win $ as a measurement of the orthonormal basis $\ket{0}, \ket{1}, \ket{2}$ interpreted respectively as 
 knowing the previous game was won, lost, or `neither' (unknown).
We define another basis $\ket{0'}, \ket{1'}, \ket{2'}$ 
 via a unitary $U$ with $\ket{i'} = U^\dagger \ket{i}$ for $i=0,1,2$. given by: 
\[
U = \begin{bmatrix}
    \sqrt{a}       & e^{i \beta} \sqrt{\frac{(1-a)(1-b)}{a}} & -\sqrt{\frac{(1-a)(a+b-1)}{a}}\\
    0      & e^{i \beta} \sqrt{\frac{(a + b -1)}{a}} & \sqrt{\frac{1-b}{a}} \\
    \sqrt{1-a}      & -e^{i \beta} \sqrt{1-b} & \sqrt{a + b -1}
\end{bmatrix}
\]
 where $0 \leq a,b < 1$ and $\beta \in [-\pi, \pi]$. 
 We then model $\Bet$ as the projective measurement $\{\Bet_y,\Bet _n\}$ with 
 \[
B_y = U^\dagger(\ket{0}\bra{0} + \ket{1}\bra{1})U \qquad \qquad B_n=U^\dagger\ket{2}\bra{2}U
\]
and use pure initial state $\rho = \ket{\psi}\bra{\psi}$ where 
$\psi = \psi_0 \ket{0} + \psi_1 \ket{1}$.
For the $P(By)$ condition we require: 
\begin{align*}
P(\Bet y) = \| \Bet_y \psi \|^2 
&= \|(\ket{0}\bra{0} + \ket{1}\bra{1})U \psi \|^2 \\ 
&= 
\left| \sqrt{a} \psi_0 + e^{i \beta} \psi_1
\sqrt{\frac{(1-a)(1-b)}{a}} \right|^2 + \frac{a+b-1}{a} \left|\psi_1\right|^2 \\ 
&=
a\left|\psi_0\right|^2 + b\left|\psi_1 \right|^2 + 2 \cos(\beta)\psi_0\psi_1\sqrt{(1-a)(1-b)}
\end{align*}
in the last step assuming $\psi_0, \psi_1 \in \mathbb{R}$.
 Setting $\psi_0 = \psi_1 = \frac{1}{\sqrt{2}}$ gives: \[
 {P(\Bet y) = \frac{a + b}{2} + \cos(\beta)\sqrt{(1-a)(1-b)}}
 \]
For the $P(\Bet y \mid \Win y)$ condition, after conditioning on outcome $\Win y$ we can assume  $\psi = \ket{0}$ and then calculate as above with $\psi_1 = 0$, obtaining $P(\Bet y \mid \Win y) = a$. For the known lose condition we similarly can set $\psi_0 = 0$ and obtain $P(\Bet y \mid \Win n) = b$. Now $\beta=0.7875 \pi$ produces $P(\Bet y)=0.36$. 
\end{example}

\begin{example} \label{ex:disj-fallacy}
We can give a simple Markov model of any data for a disjunction effect on $S=\{\star, \sany, \sann\}$, by fitting the distributions $P(\Win )$ and $\PAB$ exactly as in \eqref{ex:class-interference}. If we merely wish to fit the distributions $P(\Bet )$, $P(\Bet \mid \Win y)$ and $P(\Bet  \mid \Win n)$ we are free to define $T_\Win $ as any arbitrary channel with $T_\Win(\sany \mid \star) > 0$ and $T_\Win(\sany \mid \star) > 0$. 
\end{example}

\section{Further Cognitive Effects} \label{app:further-effects}

Thanks to Theorem \ref{thm:-simple-M-Model}, none of the cognitive effects of Section \ref{sec:cog-effects} or Appendix \ref{app:cond-data} can structurally rule out classical instrument models. 
In this Section we consider two further aspects of modelling from \cite{busemeyer2012quantum} which one may hope to rule out classical models instead, and see how to describe both aspects diagrammatically in any process theory. In each case  we find once again that while Bayesian models are ruled out, general classical instrument models remain available. 

\subsection{Temporal Bell inequality}

Beyond modelling decisions in discrete sequences, we could consider models of decisions taking place over \emph{time}. 
The time evoluton of the participant's mental state on $S$ from a state $\rho_0$ at time $t_0$ to a state $\rho_1$ at time $t_1$ can be modelled in a process theory by a \emph{time evolution} channel $U$ on $S$:
\[
\input{time.tikz} 
\]

\begin{example}
 In a quantum model a time evolution $U= U(t_1, t_0)$ is a unitary operator determined by the Hamiltonian of the system. In a classical model, $U$ would instead be given by a Markovian time evolution. 
\end{example}

Thus we wish for ways to distinguish data arising from a classical versus quantum time evolution. 
 One way proposed in \cite{atmanspacher2010proposed} is through the so-called `Temporal Bell Inequalities' which, like the conventional Bell inequalities, are structural inequalities satisfied by all Bayesian classical models. Here we will follow the presentation in  \cite[8.1.3]{busemeyer2012quantum}. 

We consider a model $\modelM$ featuring an object $S$ in initial state $\rho$. The state then evolves between three time steps: an initial time $t_1$ followed by an evolution $U_1$ to later time $t_2$ and evolution $U_2$ to a later time $t_3$. At each time $t_1, t_2, t_3$ we respectively apply a decision $D_1, D_2, D_3$, each given by one of the two instruments:
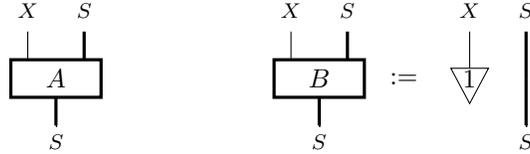
\[
\input{TBIB2.tikz}
\]
 with outcomes $X = \{1,-1\}$. Here the decision $B$ is `thought of as doing `nothing', always returning outcome $1$. 
Like any channels, we can also think of the evolutions $U_1$ and $U_2$ as decisions with a single outcome (more specifically given by time evolution channels). We are then interested in the probability distributions:
\[
P(D_1 \then D_2 \then D_3) := P_\modelM(D_1 \then U_1 \then D_2 \then U_2 \then D_3)
\]
where each $D_1, D_2, D_3 \in \{A,B\}$. Explicitly:
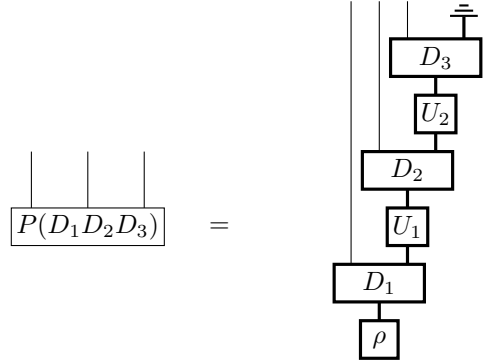
\[
\input{PD1D2D3.tikz}
\]
We write:
\[
P(D_1D_2D_3 = -1) = \sum_{a,b,c \mid abc = -1} P(D_1aD_2bD_3c)
\] 
for the induced probability that the multiplication of the outcomes of each decision is equal to $-1$. The inequality applies only to the distributions $P(A \then A \then B)$, $P(A \then B \then A)$ and $P(B \then A \then A)$, for the three scenarios in which $B$ appears once, and two decisions $A$ are made. 

\begin{lemma}[Temporal Bell Inequality] \label{lem:TBI}
Suppose that the instrument $A$ is \emph{non-disturbing} in that: 
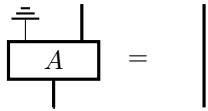
\[
\input{mnondist.tikz}
\]
Then we have:
\begin{equation} \label{eq:TBI}
P(AAB=-1) + P(BAA=-1) \geq P(ABA=-1) 
\end{equation}
\end{lemma}
\begin{proof}
Using that $A$ is non-disturbing and the definition of $B$, note that: 
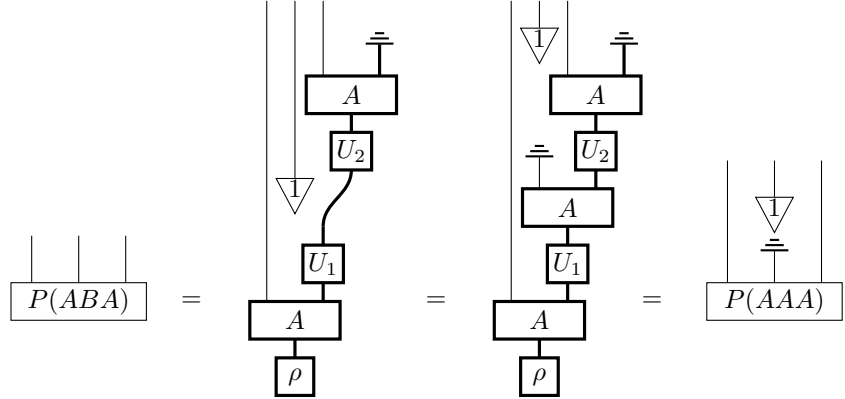
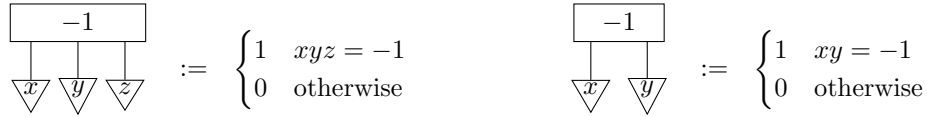
\[
\input{PABA.tikz}
\]
Now define:
\[
\input{minus1def.tikz} \ \ := \ \ \begin{cases} 1 & xyz = -1 \\ 0 & \text{otherwise}\end{cases}
\qquad \qquad \qquad 
\input{minus1def2.tikz} \ \ := \ \ \begin{cases} 1 & xy = -1 \\ 0 & \text{otherwise}\end{cases}
\]
Then from the above we have: 
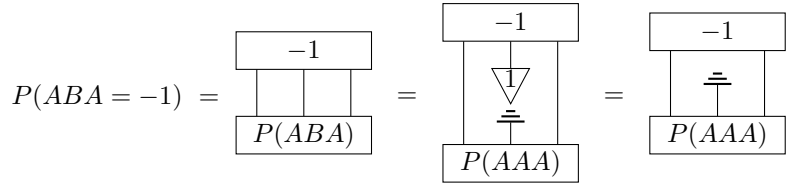
\[
P(ABA=-1) \ = \ 
\input{minus1.tikz}
\]
Then $P(BAA=-1)$ and $P(AAB)=-1$ are given similarly by discarding the relevant wire for $B$. Moreover it is easy to see that: 
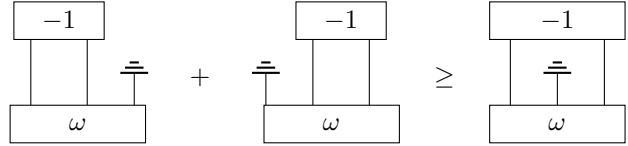
\[
\input{TBIpf2.tikz}
\]
on any state $\rho$. Indeed for any $(x,y,z) \in X \times X \times X$ if $xz = -1$ then $x \neq z$ and so either $x \neq y$, and $xy = -1$, or $z \neq y$ and $yz = -1$. Hence:
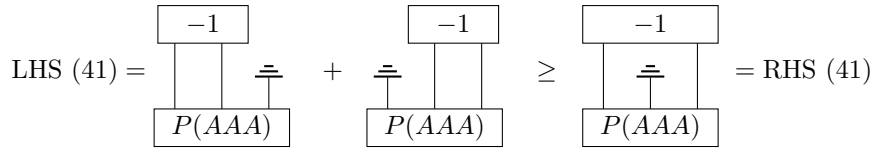
\begin{align*}
\text{LHS \eqref{eq:TBI}}
=
\input{AAB.tikz}  \ \ + \ \ \input{BAA.tikz}  \ \ 
 \geq 
\ \ 
\input{ABA.tikz} = \text{RHS \eqref{eq:TBI}}
\end{align*}
\end{proof}

\begin{corollary}
Any Bayesian model will fail to satisfy the Temporal Bell Inequality.
\end{corollary}
\begin{proof}
Any Bayesian instrument is non-disturbing since it shows no interference with the identity channel on $X$. Explicitly: 
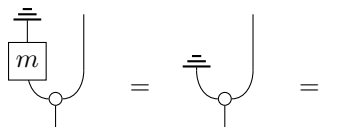
\[
\input{bayesnodist.tikz}
\]
\end{proof}
In contrast, quantum models can violate the inequality, including those based on basis measurements, for an example see \cite[8.1.3]{busemeyer2012quantum}. As such, in \cite{atmanspacher2010proposed} the inequality is proposed as a means to identify quantum evolution of a participants mental state. However, more general classical instruments can exhibit interference and are easily seen not to be non-disturbing, and so may fail to satisfy the inequality. Hence once again this result fails to rule out more general classical instruments.
Another issue is that the essence of the temporal Bell inequality derives from the sequential decisions and not the time evolutions $U_1, U_2$. Indeed we did not need to specify any of their properties in the proof above and in fact they could simply be absent with $U_1=U_2 = \id{} $. Hence the inequality can be seen simply as a property of decision data, and not fundamentally about the time evolution.  

We note that further work has been done on identifying quantum aspects of evolution in psychology, including so-called `Quantum Zeno effects' \cite{yearsley2016zeno}; we leave an analysis of such further work for the future. 

\subsection{Changes of perspective} 

Another area of cognition in which quantum models have been applied is in modelling `changes of perspective'. Recall that a channel $U \colon S \to S$ is reversible when it comes with an inverse $U^{-1}$ satisfying $U\circ U^{-1} = U^{-1} \circ U = \id{S}$. In process theories such as $\Class$ and $\QC$ these can be used to relate states as follows.

\begin{examples}
Any two sharp states $a, b$ of a (classical) finite set $S$, corresponding to point distributions, are related by any (non-unique) isomorphism $f \colon S \to S$ with $f \circ a = b$. For any quantum system $\hilbH$ any two pure states $\psi, \phi$ are related by a (non-unique) unitary transformation with $U \circ \psi = \phi$. 
\end{examples}

Any such reversible transformation can be used to represent a change of perspective in the following manner.

\begin{definition}
Let $I$ be an instrument on $S$ and $U$ a reversible channel on $S$. We define the \emph{change of perspective} instrument $U(I)$ on $S$ by: 
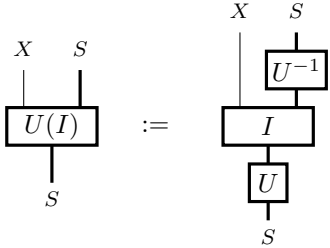
\[
\input{change-perspective.tikz}
\]
\end{definition}

The following hypothetical example is used to introduce various aspects of quantum-like modelling in \cite[Section 2]{busemeyer2012quantum}, including changes of perspective, and given a more succinct treatment in \cite{jacobs2017quantum}, which we here follow and generalise. 

Consider a scenario in which a man and a woman are each making the same decision from each of their perspectives. Their decisions are respectively denoted by $\carinst$ and $W$ with outcomes $X$. In the original example they are choosing which car to purchase from options $X=\{a,b,c\}$ from $a$ (Audi), $b$ (BMW) and $c$ (Cadillac). 

We model the man's decision as an instrument $\carinst$ on $S$ with outcomes $X$, with preferences captured via an initial state $\initM$. We consider a reversible channel $U$ on $S$ which captures the change to the woman's perspective, so that the woman's decision is given by the instrument $W = U(\carinst)$. Thus we consider an overall model $\modelM=(S,\initM,\carinst,W)$ of the form below. 
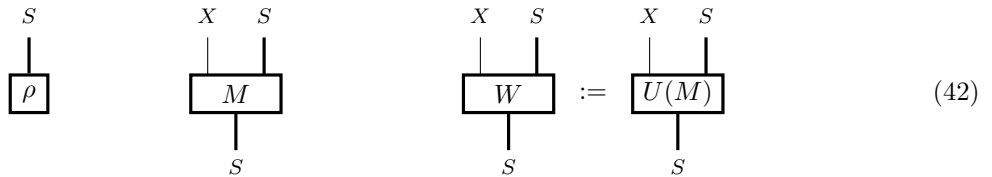
\begin{equation} \label{eq:carmodel}
\input{carmodel.tikz}
\end{equation}
Given such a model $\modelM$, the woman's initial preferences are then in turn derived from the state $\initW := U \circ \initM$ of $S$. Concretely, the man's initial preferences for the options $X$ are given by the distribution $P_\modelM(M)$ over $X$, and the woman's by $P_\modelM(W)$, which we may find to be given as below. 
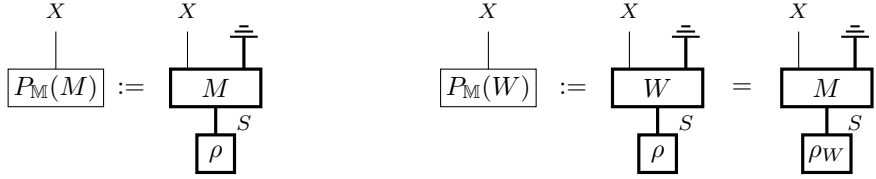
\[
\input{car-dist.tikz}
\]
Additionally the man, beginning from his own preferences, can first evaluate an option $X$ from the woman's perspective, and then a further option $Y$ from his own perspective. This is captured by the distribution $P_\modelM(W \then M)$ over $X \times Y$ where $Y=X$, given by starting from $\initM$ and then applying $W$ and then $M$. 
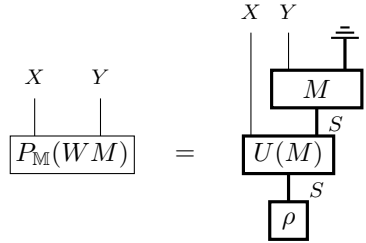
\[
\input{WMcarstart.tikz}
\]
\begin{remark}
Alternatively, the woman could consider option $X$ from her own perspective, by starting in initial state $\initW$ and then applying $\carinst$, and then consider option $Y$ from the man's, by switching perspective back via $U^{-1}$. This would yield a distribution $\PMW$ over $X \times Y$.  
In fact we can see that this is precisely $P_\modelM(W,M)$ again, as below. This fact is derived in \cite{jacobs2017quantum} but yet more easily here diagrammatically.  
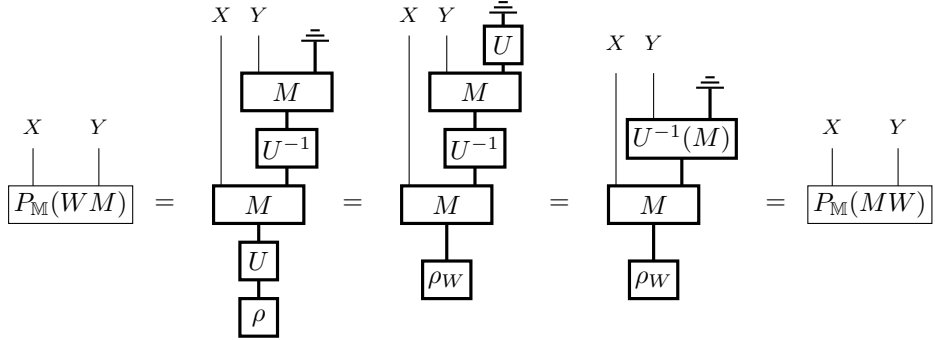
\[
\input{WMcar3.tikz}
\]
\end{remark}
These change of perspective distributions can be used to model various psychological phenomena. For example, suppose the woman prefers option $c$ (Cadillac). The man's preference for option $c$ may be increased by first considering the woman's preference for $c$, and then his own preferences, represented by a model satisfying the following inequality.
\begin{equation} \label{eq:MW-effect}
\input{WM-effect-simpler.tikz} 
\end{equation}

From the inequality above one can see that this requires an interference effect of $W = U(\carinst)$ on $\carinst$. 

\begin{lemma}
Let $\upd{m}$ be a Bayesian instrument and $U$ a deterministic reversible channel on $S$. Then $U(\upd{m}) = \upd{m \circ U}$ and in particular is again a Bayesian instrument.
\end{lemma}
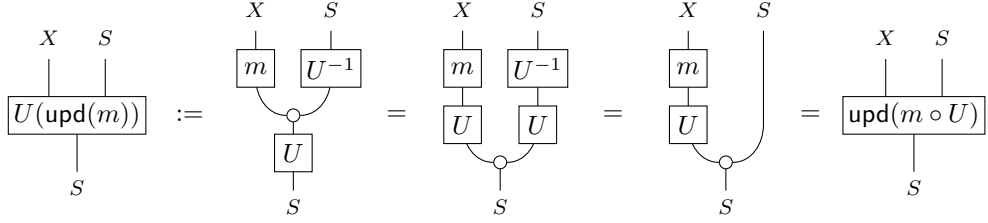
\begin{proof}
\[
\input{change-perspective-bayes.tikz}
\] 
\end{proof}
Since Bayesian models admit no interference effects, Bayesian models (with deterministic reversible channels) are therefore unable to account for change of perspective effects.
Instead quantum projective models have been applied to yield models such as the following.

\begin{example}
The following quantum projective model of the above scenario is spelled out geometrically in \cite[Section 2]{busemeyer2012quantum} and more succinctly in terms of `effect logic' in \cite{jacobs2017quantum}, which we recount here.\footnote{In our terminology the latter is a special case where the instruments are given by quantum update instruments; see Section \ref{sec:meas-models}. Our presentation here extends this beyond updating to general (quantum, classical or otherwise) instruments exhibiting interference.}  
We model $\carinst$ as a measurement on $\hilbH=\mathbb{C}^3$ for the orthornomal basis $\ket{0} = \ket{a}$, $\ket{1} = \ket{b}$, $\ket{2} = \ket{c}$. We set $\initM = \ket{m}\bra{m}$, $U$ to be a channel induced by a unitary $U$ given below, so that $\initW = U^\dagger M U = \ket{w}\bra{w}$ with $\ket{w} = U^\dagger \ket{m}$, as follows.

\[ 
\ket{m} = \begin{bmatrix} -0.6963 \\ 0.6963 \\ 0.1741 \end{bmatrix}
\qquad \qquad 
U = \begin{bmatrix} 
\frac{1}{\sqrt{2}} & \frac{1}{2} & \frac{-1}{2} \\ 
\frac{1}{\sqrt{2}} & \frac{-1}{2} & \frac{1}{2} \\ 
0 & \frac{1}{\sqrt{2}} & \frac{1}{\sqrt{2}}
 \end{bmatrix}
 \qquad \qquad 
 \ket{w} = \begin{bmatrix} 0 \\ -0.5732 \\ 0.8194 \end{bmatrix}
\]
Then we obtain the following probabilities:
\begin{align*}
\PM(a) &= \left \lvert {\braket{a \mid m}} \right \rvert^2 = 0.485 &  
\PW(a) &= \left \lvert {\braket{a \mid w}} \right \rvert^2 = 0 \\ 
\PM(b) &= 0.485 &
\PW(b) &= 0.329\\  
\PM(c)& = 0.030 & 
 \PW(c) &= 0.671
\end{align*}
The distribution $\PWM$ has for example:
\[
\PWMcc = \left|\braket{c \mid w}\right|^2  \left|\braket{c \mid U^\dagger c}\right|^2 
=
\left|\braket{c \mid w}\right|^2  \left|\braket{c \mid U c}\right|^2 
=
0.671 \times 0.5
=
0.336
\]
so that when the man takes the woman's perspective into account, the probability of choosing $c$ is increased from $0.030$ to $0.336$.
\end{example}

We might wonder whether the ability to model interference effects specifically with change of perspective instruments (i.e.~giving a model $\modelM = (\initM, M, W)$ with $W=U(M)$ of the form \eqref{eq:carmodel}) is a quantum feature not possible for classical models. However, we can in fact do so using general non-Bayesian classical instruments, thanks to the following result. 

\begin{proposition}
For any classical instrument model $\modelM = (S, \rho, A, B)$ one can specify a classical model $\modelM' = (S', \rho', A', B')$ where $B':= U(A')$ for an isomorphism $U$ on $S'$, inducing the same distributions $P_\modelM=P_{\modelM'}$, i.e.~$P_\modelM(A \then B) = P_{\modelM'}(A' \then B')$, $P_\modelM(B \then A)=P_{\modelM'}(B' \then A')$ etc. 
\end{proposition}
\begin{proof}
Define: 
\[
S'= S_A \coprod S_B \qquad S_A := \{s_A \mid s \in S\} \qquad S_B :=\{s_B \mid x \in S\}
\]
For $s, s' \in S$  and $x \in X_A, y \in X_B$, define: 
\begin{align*}
U(s_A) &= s_B & \rho'(s_A)&=\rho(s) & A'(s'_A, x \mid s_A) &= A(s', x \mid s) \\ 
U(s_B) &= s_A & \rho'(s_B)&=0 & A'(s'_B, y \mid s_B) &= B(s', y \mid s)
\end{align*}
Then one may verify the result straightforwardly. By construction the initial state $\rho$ begins in $S_A$ and the state remains there when applying $A'$ or $B'=U(A')$. The instrument $A'$ applies $A$ within $S_A$, while $B'$ moves the state via $U$ to $S_B$ then applies $B$ and moves back to the corresponding state in $S_A$. Hence we obtain $P_\modelM=P_{\modelM'}$. 
\end{proof}

In particular, we can model interference via change of perspective such as \eqref{eq:MW-effect} classically by first defining a classical instrument model $(\rho,M,W)$ of any given distributions $P(M)$, $P(W)$ and $\PWM$ (in particular those with interference effects) and then using the above result to define an equivalent model $\modelM = (S',\rho',M',W')$ yielding the same distributions and now with $W'=U(M')$. Thus, like other sequential effects, changes of perspective rule out Bayesian models but can be handled by more general classical instruments.

\section{Proofs} \label{app:proofs}

\begin{lemma} \label{lem:full-faithful-embedding}
Any probabilistic process theory $\catC$ has a full and faithful monoidal embedding $\Class \hookrightarrow \catC$.  
\end{lemma}
\begin{proof}
We send each finite set $X$ to the corresponding classical object of $\catC$. We send each sub-channel $M \colon X \to Y$ in $\Class$ to the unique such morphism in $\catC$ satisfying \eqref{eq:class-values}. Indeed since $M(y \mid x)$ is a sub-distribution of $Y$, for each $x \in X$, there exists a unique such morphism in $\catC$ by axiom \eqref{enum:states-to-sub-dist} of Def.~\eqref{def:prob-PT}. It is easy to see that this mapping respects identities, $\otimes$ and $\discard{}$. Finally we check that it respects composition. For any $M \colon X \to Y$ and $N \colon Y \to Z$ in $\catC$ we have: 
\[
\input{functor.tikz}
\]
\end{proof}

\begin{proof}[Proof of Proposition \ref{prop:probPT}]

\eqref{enum:sum-rule}, \eqref{enum:normalisation}: 
Let $(\omega(x))_{x \in X}$ be a non-zero sub-distribution over $X$. Write $p = \sum_{x \in X} \omega(x)$. Define $(\hat \omega(x))_{x \in X}$ to be the normalised distribution over $X$ with $\hat \omega(x) := \frac{\omega(x)}{p}$ for all $x \in X$. Then by axiom \eqref{enum:states-to-sub-dist} $\hat \omega$ corresponds to a normalised state of $X$. Then since $\omega(x) = \hat \omega(x) \cdot p$ for all $x \in X$ we have $\omega = p \cdot \hat \omega$ as states of $X$, again by axiom \eqref{enum:states-to-sub-dist}.  But then since $\hat \omega$ is a normalised state $\discard{} \circ \omega = p$ as required for \eqref{enum:sum-rule}, and then \eqref{enum:normalisation} holds with $\normop(\omega) = \hat \omega$. 

\eqref{enum:unif-state}: a special case of axiom \eqref{enum:states-to-sub-dist}.

\eqref{enum:joint-monic}:
Recall the `cap' process $\tinycap \colon X \times X \to I$ defined by \eqref{eq:classbend}.
Then composing with each state for $x \in X$ one may verify that: 
\begin{equation} \label{eq:snake}
\input{snake2.tikz}
\end{equation}
Now $f, g \colon A \to X \otimes B$ agree on all $x \colon X \to I$ as in the LHS of \eqref{eq:jointmonic} iff $f, g$ agree when composing with $\tinycap$. Composing again with $\tinycopy \circ \discardflip{}$ and applying \eqref{eq:classbend} we obtain $\frac{1}{|X|} \bullet f = \frac{1}{|X|} \bullet g$. Hence $f =g $ by axiom \eqref{enum:scalar-cancellative}.
Axiom \eqref{enum:scalar-cancellative}: 
Suppose that $f, g \colon A \to B$ have $p \cdot f = p \cdot g$ for some  $p \in (0,1]$. Letting $\frac{1}{n} < p$ and multiplying by $\frac{1}{pn}$ we obtain $\frac{1}{n} \cdot f = \frac{1}{n} \cdot g$. Then letting $X$ be a classical object with $|X|=n$ we have:
\[
\input{fgx.tikz}
\] 
for all $x \in X$, using \eqref{enum:unif-state}. Hence $\discardflip{} \otimes f = \discardflip{} \otimes g$ by \eqref{enum:joint-monic}. Applying $\discard{X}$ we have $f = g$. 

Axiom \eqref{enum:states-to-sub-dist}:
Any state $\omega$ of $X$ is uniquely determined by the collection of probabilities $(\omega(x))_{x \in X}$ thanks to \eqref{enum:joint-monic}. Moreover since $\sum_{x \in X} \omega(x) = \discard{} \circ \omega \in [0,1]$ these are indeed a sub-distribution, and a normalised distribution precisely when $\omega$ is a normalised state. Conversely, let $(\omega(x))_{x \in X}$ be any sub-distribution. It remains to show there exists a state $\omega$ of $X$ satisfying \eqref{eq:probs}. For this, first define $e \colon X \to I$ by $e \circ x = \omega(x)$ for all $x$, and then define $\omega$ to be the normalisation of the state:
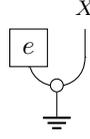
\[
\input{def2.tikz}
\]
\end{proof}

\begin{proof}[Proof of Theorem \ref{thm:-simple-M-Model}]
Using that the data satisfies no-backwards signalling, for each subsequence $\q_1,\dots,\q_k$ of a full-length sequence in the data define a distribution $P(\q_1 \an_1 \thenv \dots \thenv \q_k \an_k)$ over $\ans_1 \times \dots \times \ans_k$ by \eqref{eq:subseq-dist}. 

Each element $s$ of $\Hi$ is by definition a sequence $\q_1 \an_1 \dots \q_k \an_k$ of decision-outcome pairs $(\q_j, \an_j)$
 of length $k \leq n-1$. Define $\star = ( ) \in \Hi$ as the empty sequence.
Define each $f_\q = f \colon \Hi \to \A$ to send a sequence $s$ to its last answer $\an_k$, and set $f(\star)$ to be an arbitrary element of $\Hi$. For each $\q \in \Q$ and $s \in \Hi$ define 
\begin{equation} \label{eq:def-cond-prob}
\ctrans{\q}(\hi' \mid \hi) := 
\frac{P(\q_1 \an_1 \then \dots \thenv \q_{k+1} \an_{k+1})}{P(\q_1 \an_1 \then \dots \thenv \q_k \an_k)}
\end{equation}
for each pair with $s = \q_1 \an_1 \dots \q_k \an_k \in \events_k$  and
\begin{equation}
\hi' = 
\begin{cases}
(\q_1 \an_1 \dots \q_k \an_k\q\an_{k+1}) & 1 \leq k < n-1 \\ 
(\q,\an_{k+1}) & k \in \{0,n-1\}
\end{cases}
\end{equation}
where $(\q_1 \an_1 \dots \q_k \an_k\q\an_{k+1})\in \events_{k+1}$, setting $\ctrans{\q}(\hi' \mid \hi) = 0$ otherwise. For the case $k=0$ we mean that $\ctrans{\q}(\q \an_1 \mid \star) = \data(\q \an_1)$. We now verify this model fits the data . Setting $\hi_0 = \star$ we have:
\begin{align*}
\sum_{\hi_1,\dots,\hi_n \text{ s.t. } 
f(\hi_j) = \an_j } \prod^n_{j=1} \ctrans{\q_j}(\hi_j \mid \hi_{j-1})
&= \prod^n_{k=0} 
\frac{\data(\q_1 \an_1 \thenv \dots \thenv \q_{k+1} \an_{k+1})}{\data(\q_1 \an_1 \dots \q_k \an_k)} \\ 
&= \data(\q_1 \an_1 \thenv \dots \thenv \q_{n} \an_{n})
\end{align*}
\end{proof}

\begin{proof}[Proof of Theorem \ref{Thm:alt-two}]
Let $(\an_1,\dots,\an_k)$ be an alternating sequence of length $1 \leq k \neq n$. 
For any outcomes $\an_1,\dots, \an_k \in \ans$ define $P(\q_1 \an_1 \thenv \dots \thenv \q_k \an_k)$ just as in \eqref{eq:subseq-dist}, where $(\q_1,\dots,\q_n)$ is the unique extension to an alternating sequence.
For each sequence $(\an_1,\dots,\an_k)$ of outcomes of length $1 \leq k \neq n$, denote the corresponding state in $\Hi$ by $\hi = \san_1\dots\san_k$. Define $\star = ( \ )$ as the empty sequence. 
Define $f(\hi) = \an_k$.
For $k= 0$ we mean $\hi = \star$ and define $f(\star)$ arbitrarily.

We define each $T_\q$ as follows. 
Set $T_\q(\hi' \mid \hi) = 0$ unless
\begin{equation} 
\hi' = 
\begin{cases}
\san_1 \dots \san_k \san_{k+1} & 1 \leq k < n-1 \\ 
\san_{k+1} & k=n-1
\end{cases}
\end{equation}
When $\hi'$ is of the above form, define $\ctrans{\q}(\hi' \mid \hi)$ by the expression \eqref{eq:def-cond-prob}
where the sequence $(\q_1,\dots,\q_{k+1})$ is given by $(\q^\bot,\q,\dots,\q^\bot, \q)$ when $k$ is even and $(\q,\q^\bot,\dots,\q^\bot, \q)$ when $k$ is odd, always ending in $\q$. Here we mean $\ctrans{\q}(\an_1 \mid \star) = \data(\q \an_1)$. 

Now let us verify that this yields a valid model of the data. 
Inspecting the values for which the sum below is non-zero we see the following. 
\begin{align*}
& \sum_{\hi=(\hi_j)^n_{j=1} \text{ s.t. } 
f_{\q_j}(\hi_j) = \an_j } \ctrans{\q_j}(\hi_j \mid \hi_{j-1})
\\ 
&= 
\ctrans{\q_1}(\san_1 \mid \star) 
\ctrans{\q_2}(\san_1 \san_2 \mid \san_1)
\dots
\ctrans{\q_{n-1}}(\san_1\dots\san_{n-1} \mid \san_1\dots\san_{n-2})
\ctrans{\q_n}(\san_n \mid \san_1 \dots \san_{n-1})
\\ 
&=
\data(\q \an_1) \prod^{n-1}_{k=1} \frac{P(\q_1 \an_1 \thenv \dots \thenv \q_{k+1} \an_{k+1})}{P(\q_1 \an_1 \thenv \dots \thenv \q_k \an_k)}
\\ 
&= \data(\q_1 \an_1 \thenv \dots \thenv \q_n \an_n)  
\end{align*}
\end{proof}

\begin{proof}[Proof of Theorem \ref{thm:RRE-no-order}]
Let $(\povmA_x)_{x \in X}$ and $(\povmB_y)_{y \in Y}$ be POVMs on $\hilbH$. If the conditions are satisfied then RRE is satisfied by Proposition \ref{prop:RRE} \eqref{enum:idem-com-RRE}. Conversely, suppose that they satisfy RRE in the state $\rho$. Let $a_x = \sqrt{\povmA_x}$ and $b_y = \sqrt{\povmB_y}$ for each $x \in X, y \in Y$.
First assume $\rho = \ket{\psi}\bra{\psi}$ is a pure state corresponding to a normalised vector $\psi$. RRE on $\rho$ amounts to precisely the following equalities (and similar swapping $\povmA$ and $\povmB$).
\begin{align}
\|a_y a_x \psi\|^2 = \delta_{x,y} \|a_x \psi\|^2 \label{eq:pf-rep} \\ 
\|a_z b_y a_x \psi\|^2 = \delta_{x,z} \|b_y a_x \psi\|^2 \label{eq:pf-RRE}
\end{align}
Since the $\povmA_x$ form a POVM we have: 
\[
a_x \psi = \sum_y a^2_y a_x \psi = a^3_x \psi
\]
using that $a_y a_x \psi = 0$ for $x \neq y$ by repeatability \eqref{eq:pf-rep}. It follows from Lemma \eqref{lem:helpful} that we have:
\begin{equation}
a_x^2 \psi = a_x \psi
\end{equation}
for all $x$. Combined with \eqref{eq:pf-rep} it follows that $\hat \povmA$ is idempotent on $\rho$. It also follows that they show no order effect on $\rho$ since:
\begin{align*}
\| a_x b_y \psi \| &= \left\| \sum_z a_x b_y a_z \psi \right\| \\ &=_\eqref{eq:pf-RRE} \| a_x b_y a_x \psi \|  \\  &=_\eqref{eq:pf-RRE} \|b_y a_x \psi \|
\end{align*}
More strongly we will show that they commute on $\rho$. First note that we have:
\begin{equation} \label{eq:step}
a_x b_y \psi = \sum_z a_x b_y a_z \psi = a_x b_y a_x \psi
\end{equation}
Now from \eqref{eq:pf-RRE} it follows that for $x \neq z$ we have:
\[
a_x a^2_z b_y a_x \psi = 0
\]
Hence:
\[
a_x b_y \psi 
=_\eqref{eq:step} a_x b_y a_x \psi
= \sum_z a_x a^2_z b_y a_x \psi = a^3_x b_y a_x \psi 
= a^2 (a_x b_y a_x) \psi 
= a^2_x a_x b_y \psi 
= a^3_x b_y \psi
\]
It follows from Lemma \ref{lem:helpful} again that $a_x b_y \psi = a^2_x b_y \psi$. But then: 
\[
a_x b_y \psi 
= a^2_x b_y \psi 
= a^2_x b_y a_x \psi 
= \sum_z a^2_z b_y a_x \psi 
= b_y a_x \psi
\]
and so $a_x, b_y$ commute on $\psi$ as desired. 

Finally, we claim the result generalises to the general case where $\rho$ is a density matrix. Write $\rho = \sum p_j \ket{\psi_j}\bra{\psi_j}$ as a convex mixture of pure states where each $p_j > 0$. Then we claim that if $(\hat \povmA, \hat \povmB, \rho)$ satisfies RRE then so does $(\hat \povmA, \hat \povmB, \ket{\psi_j} \bra{\psi_j})$ for each $j$. Indeed repeatability and RRE amount to precisely the following equalities. 
\begin{align*}
\sum_j p_j \|a_y a_x \psi_j\|^2 &= \delta_{x,y} \sum_j p_j \|a_x \psi_j\|^2 \\ 
\sum_j p_j \|a_z b_y a_x \psi_j \|^2 &= \delta_{x,z} \|b_y a_x \psi_j\|^2
\end{align*}
Since each operator $a_y$ has norm $\leq 1$ we have $\|a_y \phi\|^2 \leq \|\phi\|^2$ for all $\phi$. Hence in each of the above equations equality can only hold if each $j$ indexed summand is an equality. Hence $(\hat \povmA, \hat \povmB, \ket{\psi_j} \bra{\psi_j})$ satisfies RRE for each $j$. Hence $\hat \povmA, \hat \povmB$ are idempotent and commute (exhibit no order effect) on $\ket{\psi_j}\bra{\psi_j}$ for all $j$, from which the result follows for the mixture $\rho$.
\end{proof}

\begin{proof}[Proof of Proposition \ref{prop:RRE} \eqref{enum:cor-meas-POVM}]
Let $\measm=(\measm_x)_{x \in X}$ be a POVM on $\hilbH$ for which $\hat \measm$ is repeatable, and $\psi$ be a unit vector in $\hilbH$. Then since $\hat \measm$ is repeatable on $\ket{\psi}\bra{\psi}$, just as in the proof of Theorem \ref{thm:RRE-no-order} we have $\measm_x \circ \measm_y \circ \psi = 0$ for $x \neq y$ and $\measm^2_x \psi = \measm_x \psi$, for all $x$. Since $\psi$ was arbitrary, $\measm$ is a projection. 

Conversely, if $\measm$ is a projective measurement then $\hat \measm$ is repeatable by Proposition \ref {prop:RRE} \eqref{enum:proj-q-idem}. 
\end{proof}

\begin{proof}[Proof of Theorem \ref{thm:joint-dec-thm}]
Treat the joint decision data as a sequential dataset of distributions $\PAB := P(A \jdec B) =: \PBA$ for $A \in \mathcal{A}, B \in \mathcal{B}$. Then no signalling \eqref{eq:no-signalling} now ensures that the data satisfies no backward-signalling, and so we can define a simple Markov model as in Theorem \ref{thm:-simple-M-Model}. That is, set 
\[ 
S = \{\star\} \cup \{Ax\}_{A \in \mathcal{A}, x \in X_A} \cup \{Bx\}_{B \in \mathcal{B}, x \in X_B} \cup \{AxBy\}_{A \in \mathcal{A}, B \in \mathcal{B}, x \in X_A, y \in X_B}
\]
For each $A \in \mathcal{A}$, define $\markovfA(AxBy) = x = \markovfA(A_x)$ and define $\markovfA(s)$ arbitrarily otherwise. Then define
\[
\markovTA(s' \mid s) := 
\begin{cases}
P(Ax) & s= \star, s'=Ax \\ 
P(Ax \thenv By \mid By) & s'=AxBy, s=By \\ 
\delta_{s,s'} & \text{otherwise}
\end{cases}
\]
Then one may verify that this forms a model of the $\PAB, \PBA$ data as required, and for for $A \in \mathcal{A}, B \in \mathcal{B}$ that $\markovTA \circ \markovTB = \markovTB \circ \markovTA$ and so $A$ and $B$ commute. 
\end{proof}

\begin{lemma} \label{lem:helpful}
If $a$ is a positive bounded operator on a Hilbert space and $a^3 \psi = a \psi$ then $a^2 \psi = a \psi$. 
\end{lemma}
\begin{proof}
Choose an orthonormal basis $e_i$ of eigenvalues of $a$ with $a e_i = \lambda_i e_i$ with all $\lambda_i \geq 0$. Write $\psi = \sum c_i e_i$. Then 
\[
a^3 \psi - a \psi = \sum_i \lambda_i (\lambda_i^2 -1) c_i e_i = 0
\]
Hence for all $i$ either $c_i  = 0$ or $\lambda_i \in \{0,1\}$. In either case we have $\lambda_i^2 c_i = \lambda_i c_i$ for all $i$. Hence 
\[
a^2 \psi = \sum c_i \lambda_i^2 e_i = \sum c_i \lambda_i e_i = a \psi
\]
as required. 
\end{proof}

\section{Markov Models satisfying RRE} \label{sec:general-RRE-Markov-models}

We have seen that any decision data satisfying no-signalling from the future can be given a classical model. If we wish to take the subjects memories of their previous answers into account, it is natural to seek classical models moreover satisfying RRE. By using a slightly larger state space we can indeed provide classical models with this and further properties.

\begin{proposition} \label{prop:alt-two-RRE}
Any decision data of the form $\{\PAB,~\PBA\}$ 
has a Markov model with idempotent instruments satisfying RRE, with state space:
\[
\Hi = (\ansone \cup \{\sundef\}) \times (\anstwo \cup \{\sundef\})
\] 
\end{proposition}

\begin{proof}[Proof of Proposition \ref{prop:alt-two-RRE}]
A hidden state $\hi \in \Hi$ is thus a pair $\hi = (x_\qone,x_\qtwo)$ where $x_\qone \in \ansone$ or $x_\qone = \sundef$ (is `undefined') and $x_\qtwo \in \anstwo$ or $x_\qtwo = \sundef$ (is `undefined'). We use initial state $\star = (\sundef, \sundef)$. We define $f_\qone(\hi) = x_\qone$ if $x_\qone \in \ansone$, and define $f_\qone(\hi)$ arbitrarily if $x_\qone = \sundef$. Similarly $f_\qtwo(\hi) = x_\qtwo$ when this is in $X_\qtwo$ and define it arbitrarily otherwise. For any $x$ in $\ansone, \anstwo$ respectively define: 
\[
P(\qone x) := \sum_{x' \in \anstwo} P(\qone x \thenv \qtwo x')
\qquad \qquad 
P(\qtwo x) := \sum_{x' \in \ansone} P(\qtwo x \thenv \qone x')
\]
We define the channel $\ctrans{\qone}$ as follows.
\[
\ctrans{\qone}(x'_\qone x_\qtwo \mid x_\qone x_\qtwo) \ \ := \ \ 
\begin{cases}
\delta_{x_\qone, x'_\qone} & x_\qone \in \ansone \\ 
\frac{P(\qtwo x_\qtwo \thenv \qone x'_\qone)}{P(\qtwo x_\qtwo)} & x_\qone = \star, x'_\qone \in \ansone, x_\qtwo \in \anstwo \\ 
P(\qone x'_\qone) & x_\qone = \star, x_\qtwo = \star, x'_\qone \in \ansone \\ 
0 & \text{otherwise}
\end{cases}
\]
and with $\ctrans{\qone}(x'_\qone x'_\qtwo \mid x_\qone x_\qtwo) = 0$ for $x'_\qtwo \neq x_\qtwo$. We define $\ctrans{\qtwo}$ analogously. Then by inspecting the values for which each channel is non-zero, we see that: 
\begin{align*}
\sum_{\substack{s_1, s_2 \text{ s.t. } \\  f_\qone(s_1) = a, f_\qtwo(s_2) = b}} \ctrans{\qtwo}(s_2 \mid s_1)\ctrans{\qone}(s_1 \mid \star)& =
\ctrans{\qtwo}((a,b) \mid (a,\sundef))\ctrans{\qone}((a,\sundef) \mid \star) 
\\ & = 
\frac{P(\qone a \thenv \qtwo b)}{P(\qone a)} P(\qone a)
\\ &  = P(\qone a \thenv \qtwo b)
\end{align*}
The other condition is verified similarly. Finally we need to verify that the induced instruments are idempotent and satisfy RRE. 

One may see that to show that the instrument for $\qone$ is idempotent we must show that: 
\[
\ctrans{\qone}(a_2b_0 \mid a_1b_0) \ctrans{\qone}(a_1 b_0 \mid a_0 b_0)
= \delta_{a_1, a_2} \ctrans{\qone}(a_1 b_0 \mid a_0 b_0)
\]
Whenever the LHS above is non-zero we have that $a_1 \neq \sundef$ so that $a_1 \in \ansone$. In this case the left-hand term is equal to $\delta_{a_1,a_2}$ as required. 

Next, we verify RRE. One may see that to verify the $\qone \qtwo \qone$ condition of RRE we need to prove that we have:  
\[
\ctrans{\qone}(a_2b_1 \mid a_1b_1) \ctrans{\qtwo}(a_1b_1 \mid a_1b_0)\ctrans{\qone}(a_1b_0 \mid a_0 b_0) 
=
\delta_{a_1,a_2} \ctrans{\qtwo}(a_1b_1 \mid a_1b_0)\ctrans{\qone}(a_1b_0 \mid a_0 b_0)
\]
Indeed note that whenever the above is non-zero we have $a_1 \neq \sundef$, so that $a_1 \in \ansone$. But then from the definition of $\ctrans{\qone}$ we see that the left-hand term will be equal to $\delta_{a_1,a_2}$ as required. The $\qtwo \qone \qtwo$ condition for RRE holds similarly. 
\end{proof}

For yes-no decisions $\ans=\{y,n\}$ such as the Clinton-Gore data (Example \ref{ex:Clinton-Gore}), the Markov model provided by Proposition \ref{prop:alt-two-RRE} uses state space:
\[
\Hi = \{\star, \sundef\sany, \sundef\sann, \sany \sundef, \sany \sany, \sany \sann, \sann \sundef, \sann \sany, \sann \sann\}
\]
and where $\ctrans{A}$ is defined as follows. 
For each state with no outgoing arrows we transition back to the same state with probability $1$, resulting in the instruments being idempotent.
\[
\input{Trans-2.tikz}
\]
The instrument for $B$ and $T_B$ are defined in the same way, exchanging $A$ and $B$ and swapping the two components e.g. $\sann\sany \mapsto \sany\sann$, $\sany \sundef \mapsto \sundef \sany$ etc. We interpret the elements of each pair as the outcomes to $A, B$ respectively from `yes', `no' or `not yet asked'. The outcome to $A$ or $B$ returns the corresponding element of the pair, following the transition.

\section{Related Work: Entangled Concepts} \label{sec:entangled-concepts}

Bell-type arguments around the entangled nature of mental states have already been explored in the quantum-like cognition literature, most commonly in the representation of human \emph{concepts}, as in the work of Aerts and collaborators \cite{
aerts2005theory,
aerts2009quantum,
aerts2018spin}, as well as \cite[Chapter 5]{busemeyer2012quantum}, beginning with Gabora and Aerts \cite{gabora2002contextualizing}. 

In these approaches, multi-aspect concepts are (implicitly) understood as being modelled by joint decision data 
where each decision $A,A', B, B'$ has two outcomes, gathered from real world experiments on human participants who are asked to interpret a concept by choosing from a set of possible exemplars. 

\begin{dataexample} \label{ex:wind-directions}
In \cite{aerts2018spin} 
 participants are asked to choose, from a given set of options, two different examples of wind directions which together best form an example of the concept \confont{Two Different Wind Directions}.\footnote{In the original experiment, the directions are represented graphically as arrows on a compass.} The decisions $A, A'$ are the possible sets of options for the first wind direction, and $B, B'$ for the second, given as follows.
\[
A = \{N, S\}, A'=\{E, W\}, B=\{NE, SW\}, B'=\{SE, NW\}
\] 
For example, in a run of the joint decision $A \jdec B$ participants must choose whether (N, NE), (N, SW), (S, NE), or (S, SW) best fits the concept. The following joint decision data gives the proportions of choices for each pair under each condition, where for short we write $A1$ for $AN$, $B'2$ for $B'NW$ etc. 
\begin{center} \label{eq:table-PAB}
    $\begin{array}{c|cccc}\toprule
         &  $B1$&  $B2$&  $B'1$& $B'2$\\\midrule
         $A1$&  0.13&  0.55&  0.47& 0.12\\
         $A2$&  0.25&  0.07&  0.06& 0.35\\
         $A'1$&  0.13&  0.38&  0.09& 0.44\\
         $A'2$&  0.42&  0.07&  0.38& 0.09\\ \bottomrule
    \end{array}$
    \end{center}
\end{dataexample}

A related prominent area in which entanglement has been argued for in cognition has been in so-called \emph{concept combinations}, in which two or more concepts (words) are combined.

\begin{dataexample}
In \cite{aerts2011quantum}, Aerts and Sozzo consider the combined concept \confont{The Animal acts}. In an experiment, participants are asked to jointly choose a pair of words from given sets of options which best fit the concept. The first sets of decisions $A = \{\confont{Horse}, \confont{Bear}\}$ and $A' = \{\confont{Tiger}, \confont{Cat}\}$ relate to the first word \confont{Animal} and the second $B = \{\confont{Growls}, \confont{Whinnies}\}$ , $B' = \{\confont{Snorts}, \confont{Meows}\}$ to \confont{Acts}. For example, under conditions $A \jdec B$ participants must choose from the options \confont{Horse growls}, \confont{Horse whinnies}, \confont{Bear growls}, or \confont{Bear whinnies}. Again the data can be represented in a table similar to \eqref{eq:table-PAB}.
\end{dataexample}

\begin{dataexample}
Bruza et al.~in \cite{bruza2015probabilistic}, summarised in \cite[Chapter 5.5]{busemeyer2012quantum}, study entanglement in bi-ambiguous concept combinations such as \confont{Boxer bat} and \confont{Spring plant}. In particular the concept combination \confont{Ring pen} is argued to be `non-classical'. It is bi-ambiguous in that the word \confont{Ring} and \confont{pen} may each be interpreted in one of two senses; either as a shape or as an object. An experiment consists in asking participants for an interpretation of the concept from which ultimately we deduce which sense of interpretation of each word has been used. Roughly, before being presented with the concept, participants are primed for one of the senses of the word and then the test determines whether the priming was successful in each case. $A, A'$ refer to priming for each of the two senses for the first word, and $B, B'$ for the second word, with outcome $y$ when successful and $n$ unsuccessful. 
The data for \confont{Ring pen} is as follows. 
\begin{center} 
    $\begin{array}{c|cccc}\toprule
         &  By&  Bn&  B'y& B'n\\\midrule
         Ay&  0.26&  0&  0.21& 0\\
         An&  0.33&  0.41&  0.03& 0.76\\
         A'y&  0.69&  0&  0.49& 0\\
         A'n&  0&  0.31&  0& 0.51\\ \bottomrule
    \end{array}$
\end{center}
\end{dataexample}

For each of these examples, the probability tables indeed violate the Bell-CHSH inequality.\footnote{The quantity $\QQ(A,B)$ is the trace of the upper left 2x2 matrix, $\QQ(A',B)$ that below, and so on, so that in \ref{eq:table-PAB} we see that: 
$\left| \QQ(A,B) + \QQ(A',B) + \QQ(A',B') - \QQ(A,B') - 1 \right| = |-1.42| > 1$
}
 Unfortunately however, these approaches do not satisfy the requirements \ref{enum:nat} - \ref{enum:Tsirelson-bound} for a Bell argument laid out in Section \ref{subsec:Bell-arguments}. Most crucially, each of the above datasets fails to satisfy the no-signalling condition \eqref{eq:no-signalling}, as pointed out in \cite{dzhafarov2014selective}. Hence it is not possible to give parallel decision models of any of the datasets, defeating the use of Bell's Theorem to rule out classical such models.

 In response, Aerts et al.~argue that \eqref{eq:no-signalling} is a sufficient but not necessary condition to prove that no signalling occurs in the data \cite{aerts2014quantum}, and also suggest that data violating the condition should be explored more seriously in physical settings such as physical Bell experiments. 
More importantly here, rather than using parallel decision models, they propose to model concepts with general joint measurements of the form:
\begin{equation} \label{eq:product-meas}
\input{joint-meas-aerts.tikz}
\end{equation}
The authors then apply results of \cite{pitowsky1989quantum} to show that CHSH violations rule out Bayesian (`Kolmogorov') models of the above form. They thus propose the use of quantum models, which they assume to be based on pure quantum measurements, from which it follows that each joint decision $A \jdec B$ is now modelled by an entangled quantum measurement.

However, unfortunately there is no analogous Bell-like theorem to rule out more general classical models of the form \eqref{eq:product-meas}, based on probability channels. For example we can give a trivial classical model in this sense with $S_1 = S_2 = \emptyset$, $\omega=1$, and set $A \jdec B = P(A \jdec B)$. To structurally rule out classical models would require either alternative no-signalling data, or an alternative Bell-style theorem applicable to effects forcing the measurements \eqref{eq:product-meas} to be entangled).

Nonetheless, though we are not forced structurally to do so, there may be good reasons for using (entangled) joint quantum measurements of the form \eqref{eq:product-meas} as representations of concepts. Indeed one of the authors of this work has co-authored a quantum model inspired by Gärdenfors' framework of conceptual spaces in \cite{tull2024conceptual}, as discussed in Section \ref{sec:comp-models}. One argument for using quantum joint measurements could be a potential advantage in the use of quantum computers (via entanglement) to model correlations between variables, when doing so may be `more expensive' classically. It would be interesting to further explore the arguments for the modelling of concepts in this manner, within cognitive science or AI, beyond strict structural constraints. 

\end{document}